\documentclass[11pt]{article} %

\usepackage[margin=1in]{geometry} %
\usepackage{fontawesome5}
\usepackage[utf8]{inputenc}
\usepackage[T1]{fontenc}
\usepackage{mathpazo} %

\usepackage{comment,xspace}
\usepackage{authblk}
\usepackage[normalem]{ulem}
\usepackage[dvipsnames]{xcolor} %
\usepackage{todonotes}
\usepackage{booktabs,multicol,multirow,subcaption} %

\usepackage{enumitem} %
\definecolor{ptblue}{RGB}{15,76,129} %
\definecolor{ptemerald}{HTML}{009473} %
\definecolor{ptgray}{HTML}{939597} %
\definecolor{cobalt}{rgb}{0.0, 0.28, 0.67}
\usepackage{tikz} %
\usetikzlibrary{arrows.meta, positioning, calc, decorations.pathreplacing, fit, backgrounds}
\usepackage{amsmath,amsfonts,amssymb,amsthm,bbm,mathtools}
\usepackage{thm-restate}

\DeclareMathOperator{\Acc}{Acc}
\newcommand{\fancyP}{\mathcal{P}}
\newcommand{\fancyQ}{\mathcal{Q}}
\newcommand{\fancyX}{\mathcal{X}}
\newcommand{\fancyC}{\mathcal{C}}

\usepackage{needspace,tabularx,microtype}
\newcommand{\R}{\mathbb{R}}
\let\OLDforall\forall
\renewcommand{\forall}{\;\OLDforall\:}
\let\OLDexists\exists
\renewcommand{\exists}{\;\OLDexists\,}

\usepackage[ruled,linesnumbered,vlined]{algorithm2e}
\usepackage{xfrac}
\usepackage{array}

\SetCommentSty{mycommentfont}

\usepackage{minted}
\setminted{
  fontfamily=tt,
  fontsize=\small
}
\newmintinline[inlineLean]{lean4}{bgcolor=white}
\newminted[leancode]{lean4}{fontsize=\footnotesize,breaklines=true}
\usepackage[T1]{fontenc}
\usepackage[utf8]{inputenc}
\usepackage{fontspec}

\usepackage{hyperref}
\hypersetup{
linktocpage,
colorlinks=true,
citecolor=cobalt, %
urlcolor=ptblue, %
linkcolor=Plum, %
}
\usepackage{cleveref} %
\usepackage{crossreftools} %
\theoremstyle{plain}
\newtheorem{theorem}{Theorem}[section]
\newtheorem*{claim-non}{Claim}

\newtheorem{corollary}[theorem]{Corollary}
\newtheorem{lemma}[theorem]{Lemma}
\newtheorem{proposition}[theorem]{Proposition}

\theoremstyle{definition}
\newtheorem{definition}[theorem]{Definition}

\newtheorem{observation}[theorem]{Observation}
\newtheorem{claim}[theorem]{Claim}

\theoremstyle{remark}

\definecolor{DarkRedZB}{rgb}{0.8,0,0}

\newcommand{\HAnewcite}[2][]{{\hypersetup{citecolor=ForestGreen}\color{ForestGreen}\if\relax\detokenize{#1}\relax\cite{#2}\else\cite[#1]{#2}\fi}}

\makeatletter %
\newcommand{\EF}[1]{\if\relax\detokenize\expandafter{\@firstofone#1{}}\relax \text{EF}\xspace\else \text{EF#1}\fi}
\makeatother

\newcommand{\EFX}{\EF{X}\xspace}

\title{EFX Need Not Exist}

\author[1,2]{\mbox{Hannaneh Akrami}}
\author[3]{\mbox{Simon Mackenzie}}
\author[1,4]{\mbox{Alexander Mayorov}}
\author[1]{\mbox{Kurt Mehlhorn}}
\author[5]{\mbox{Shreyas Srinivas}}
\author[3]{\mbox{Mashbat Suzuki}}
\author[1]{\mbox{Christoph Weidenbach}}
\affil[1]{Max Planck Institute for Informatics \& Saarland Informatics Campus, Saarbr\"ucken, Germany}
\affil[2]{Hertz Chair for Algorithms and Optimization, University of Bonn, Germany}
\affil[3]{UNSW Sydney, Australia}
\affil[4]{Saarbr\"ucken Graduate School for Computer Science}
\affil[5]{CISPA Helmholtz Center for Information Security, Saarbr\"ucken, Germany}
\hypersetup{pdfauthor={Hannaneh Akrami, Simon Mackenzie, Alexander Mayorov, Kurt Mehlhorn, Shreyas Srinivas, Mashbat Suzuki, Christoph Weidenbach}}
\date{}

\begin{document}
\maketitle
\begin{abstract}
The existence of envy-free up to any good (EFX) allocations has long been
a central open problem in discrete fair division. We resolve it in the
negative by constructing a counterexample with three agents and eight
goods. The instance is invariant under every simultaneous permutation of
the agents and their associated good types, and this symmetry yields a
short, self-contained proof of nonexistence. The counterexample can be
realized with weighted-coverage, and hence submodular, valuations, and
it extends to $n$ agents and $m$ goods for all $n\ge3$ and $m\ge n+5$.

For three agents, eight goods is the minimum: every instance with three
agents and at most seven goods admits an EFX allocation under arbitrary
monotone valuations. The proof relies on a new analysis, valid
for any number of goods, of when the invariant-based technique
of~\cite{AACGMM25} produces an EFX allocation. 
This analysis strictly extends the class of MMS-feasible valuations.

We also give explicit bounds for approximate EFX. For subadditive
valuations, $\alpha$-EFX allocations need not exist for any
$\alpha>2^{-1/6}\approx0.891$, and for submodular valuations for any
$\alpha>0.958$. On the positive side, if two of three agents have
arbitrary monotone valuations and the third has a subadditive valuation,
there is always an allocation that is EFX for the first two and
$2^{-1/3}$-EFX ($\approx0.794$) for the third.

A first version of the counterexample and of the seven-goods result was
obtained using SAT/SMT solving, and the insights it provided guided our
classical, pen-and-paper proofs. We also describe the encoding, which we
believe is a valuable tool in its own right: to our knowledge, this is
the first use of SAT/SMT solving for EFX.
\end{abstract}

\section{Introduction}\label{sec:introduction}

The mathematical study fairness in resource allocation can be traced
back to the work of Steinhaus~\cite{Ste48}. Over the past several decades, many
fairness notions have been proposed to capture different desiderata
across a range of applications. Among these, envy-freeness has emerged
as a central benchmark: an allocation is envy-free if no agent prefers
another agent's bundle to her own.

For indivisible goods, envy-freeness cannot always be guaranteed. This
is evident even with two agents and a single good, the agent who does
not receive the good may envy the one who does. The discrete fair
division literature therefore studies relaxations that account for
indivisibility without abandoning the substance of envy-freeness. One
of its strongest and closest analogues is \emph{envy-freeness up to any
good} (EFX), introduced by Caragiannis, Kurokawa, Moulin, Procaccia, Shah, and Wang~\cite{CKM16}. An
allocation is EFX if any envy disappears after the removal of any
single good from the envied bundle.

In his CACM technical perspective
\emph{An Answer to Fair Division's Most Enigmatic Question}, Ariel Procaccia
summarized the state of the problem as follows~(\cite{procaccia2020technical})\footnote{ Procaccia stated the question for additive valuations.}:
\begin{quote}
    ``Is an EFX allocation guaranteed to exist? This fundamental and
    deceptively accessible question is open.''
\end{quote}

In this paper, we resolve the general monotone version of the EFX
existence question in the negative. We give the \emph{first
counterexample to EFX existence under general
monotone valuations}. In fact, our counterexample uses weighted coverage
valuations, a highly structured subclass of submodular valuations. We do so by showing that the EFX existence question is equivalent for general monotone and weighted coverage valuations: EFX allocations exist for every instance with monotone valuations if and only if they exist for every instance with weighted coverage valuations.

\begin{theorem}\label{thm:intro-counterexample}
There exists a fair division instance with three agents and eight goods in which all agents have weighted coverage valuations and no EFX allocation exists.
\end{theorem}

This rules out a universal EFX existence theorem not only for monotone
valuations but even for submodular valuations. The construction is small,
symmetric, and can be verified entirely by hand.

At the same time, the counterexample is minimal for three agents. Previous
positive results include EFX existence for identical valuations
\cite{Plaut-Roughgarden}, for three agents with additive valuations
\cite{Chaudhury-Garg-Mehlhorn:EFX}, for three agents with nice-cancelable
valuations \cite{berger2021almost}, and, more generally, for three
agents when one agent has an MMS-feasible valuation and the other two have
arbitrary monotone valuations \cite{AACGMM25}.
EFX is also known to exist under arbitrary monotone valuations whenever
$m\leq n+3$~\cite{Mahara-EFX}. We settle the remaining boundary case for
three agents.

\begin{restatable}{theorem}{threeseventhm}
    \label{thm:intro-seven}
    Every fair-division instance with three agents, at most seven goods, and monotone valuations admits an EFX allocation.
\end{restatable}
Combining Theorems~\ref{thm:intro-counterexample} and~\ref{thm:intro-seven}
gives the exact threshold for $n = 3$ agents:
\[
    m\leq 7 \quad\Longrightarrow\quad \text{EFX always exists},
    \qquad
    m\geq 8 \quad\Longrightarrow\quad \text{EFX may fail}.
\]

We extend our construction to $n$ agents and $m$ goods for every $n\geq 3$ and $m\geq n+5$. Together, our results clarify the limits of EFX existence.

\begin{restatable}{theorem}{mainextended}
\label{thm:general-nonexistence}%
       For every $n\ge3$ and $m\ge n+5$, there exists an instance with $n$ agents and $m$
goods, all with weighted-coverage valuations, that admits no EFX allocation. In particular, 
EFX allocations need not exist even for submodular valuations.
\end{restatable}

Any counterexample to EFX existence rules out $\alpha$-EFX for some $\alpha<1$. We exploit the structure of our constructions to turn this qualitative observation into strong, explicit inapproximability bounds.

\par\medskip\noindent\textbf{Inapproximability of EFX.}\quad
An allocation is $\alpha$-EFX if every agent values her own bundle at least $\alpha$ times as much as any other agent's bundle after the removal of any single good from that bundle. For general monotone valuations, our results rule out every positive universal approximation guarantee: for each $\alpha>0$, there exists an instance that admits no $\alpha$-EFX allocation. For more restricted valuation classes, we establish the following inapproximability bounds:
\begin{itemize}
\item For monotone subadditive valuations, there exists an instance that admits no $\alpha$-EFX allocation for any $\alpha>2^{-1/6}\approx0.890899$.
\item For monotone submodular valuations (in fact, weighted coverage valuations), there exists an instance that admits no  $\alpha$-EFX allocation for any $\alpha>\lambda_*\approx0.957455$, where $\lambda_*$ is the unique positive real root of $4z^4+4z^3-3z-4=0$.
\end{itemize}
The subadditive and weighted coverage counterexamples each use three agents and eight goods, and both constructions extend to every $n\geq3$ and $m\geq n+5$.

For comparison, $\frac{1}{2}$-EFX allocations are known to exist for monotone subadditive valuations, and hence also for monotone submodular valuations, for arbitrary numbers of agents~\cite{Plaut-Roughgarden, BS2026}. Together with these guarantees, our results place the optimal universal approximation factor in $[1/2,2^{-1/6}]$ for monotone subadditive valuations and in $[1/2,\lambda_*]$ for monotone submodular valuations.

For three agents with subadditive valuations, we improve the best known approximation factor for EFX existence from $\frac{1}{2}$ to $2^{-1/3}\approx0.793701$. Our result is stronger in two respects. First, only one agent needs to have a subadditive valuation; the other two may have arbitrary monotone valuations. Second, we guarantee EFX for the two monotone valuation agents and $2^{-1/3}$-EFX for the subadditive agent.

\begin{restatable}{theorem}{mixedapprox}
\label{thm:mixed-approx}
Let $p,q,c$ be three agents with normalized, nonnegative, monotone
valuations over a finite set of goods, and suppose that $v_p$ is
subadditive. Then there exists an allocation that is EFX
for $q$ and $c$ and $2^{- \frac{1}{3}}$-EFX for $p$.
\end{restatable}

\subsection{Our techniques and organization}\label{sec:techniques}

\par\medskip\noindent\textbf{The counterexample.}\quad
The instance has three agents $a,b,c$ and eight goods: two goods of each
of the types $A,B,C$, associated with $a,b,c$ respectively, and two
goods $x,y$. All agents use the same valuation rule relative to their
own type, so the instance is invariant under simultaneous relabelling
of agents and types. Hence each case needs to be argued only once.

The construction has two stages. First, each agent values a bundle by
its best pair, preferring its own type with $y$, then the other two
types together, then its own type with $x$. These preferences compete
for the scarce goods and force, up to relabelling, a unique EFX
allocation $Ay\mid BCCx\mid AB$ (Proposition~\ref{prop:shape}).
Second, we make the triple $BCx$, which remains after removing a $C$
from $BCCx$, more valuable to $a$ than any pair, and analogously for
$b$ and $c$. This destroys the forced allocation. No new EFX allocation
arises: every such triple contains $x$, so at most one agent can hold
one, and then the other two agents are forced onto pairs that leave no
place for $y$ (Lemma~\ref{lem:triple}).

Replacing one agent by copies and adding goods valued only by the copies
extends the counterexample to all $n\ge3$ and $m\ge n+5$
(Lemma~\ref{lem:padding}). An ordinal coverage representation
(Lemma~\ref{prop:coverage-representation}) turns it into a
weighted-coverage instance. Rescaling its values gives the subadditive
barrier, and a separate weighted-coverage instance with the same
structure gives the submodular barrier (Section~\ref{sec:quantitative}).

\par\medskip\noindent\textbf{Existence for three agents and seven goods.}\quad
We refine the analysis of the algorithm of~\cite{AACGMM25}, which
maintains a partition $\{A,B,D\}$ in which an agent $c$ accepts $A$ and
$B$ and another agent accepts $D$, and repairs it until an EFX
allocation is found. We show that whenever the repairs fail, the other
two agents $p,q$ are in a \emph{critical conflict}: each splits the
goods outside $B$ into two parts that it values above $B$ but that the
other values below $B$ minus one good, and both $p$ and $q$ accept only
$D$. Hence either EFX exists or every pair of agents is in a critical
conflict (Theorem~\ref{theorem:critical_conflict_determine_efx}). This
covers MMS-feasible agents~\cite{AACGMM25} and also one leveled
agent (Corollary~\ref{corollary:efx_for_three_agents_one_leveled}).
The same reduction, run with approximate EFX for one subadditive agent,
also yields our $2^{-1/3}$ guarantee: its failure point gives three
comparisons that contradict subadditivity (Appendix~\ref{app:positive}).

A critical conflict needs at least seven goods, and with exactly seven
its structure is rigid (Lemma~\ref{lemma:three_seven_basic_structure}).
Testing a few candidate allocations then shows that, unless EFX exists,
all agents agree on a unique most valuable good $b$. Using $\{b\}$ as a
common threshold, we show that the three critical conflicts constrain
each other's pivot bundles and
cannot coexist (Appendix~\ref{section:efx_existence_up_to_seven_appendix}).

\par\medskip\noindent\textbf{Historical Perspective.}\quad
We first obtained Theorems~\ref{thm:intro-counterexample} and \ref{thm:intro-seven} with the help of a lean-verified \cite{lean4ITP} encoding in a SAT solver. For each $n$ and $m$, we encode the EFX problem as a CNF formula $\Phi_{n,m}$ that is unsatisfiable iff every instance with $n$ agents and $m$ items admits an EFX allocation; a satisfying assignment encodes a counterexample. SPASS~\cite{SPASS} proved $\Phi_{3,7}$ unsatisfiable, via the DRAT-trim tool \cite{WetzlerHH14}, and a case distinction derived from this proof yielded a counterexample for $m = 8$.

For both results, we then worked towards proofs that can be verified by hand. For the counterexample, we searched for an instance with a symmetric structure to find one; this led to the instance presented here. For seven items, we developed a direct existence proof. We used the SMT solver Z3~\cite{Moura:TACAS08-337} in both cases. The main body of the paper contains only these human-verifiable proofs. Readers interested in the computer-assisted approach find the details in Appendix~\ref{sec:old-sat-approach}.

\par\medskip\noindent\textbf{Organization.}\quad
Section~\ref{sec:preliminaries} gives preliminaries.
Section~\ref{sec:counterexample} presents the counterexample and its
extensions, and Section~\ref{sec:quantitative} the approximation
results. Sections~\ref{section:reduction_to_conflicts}
and~\ref{section:efx_existence_up_to_seven} develop critical conflicts
and prove existence for seven goods. Section~\ref{sec:conclusion}
discusses open problems. Remaining proofs and the SAT-assisted approach
are in the appendices.

\subsection{Related Work}\label{sec:related-work}

The study of indivisible fair division seeks fairness guarantees that
remain meaningful when goods cannot be split. Surveys by
Moulin~\cite{moulin2019internet} and Amanatidis et
al.~\cite{amanatidis2022survey,survey2022} describe the broader economic
and algorithmic development of this area. A central distinction is
between envy-freeness up to one good (EF1), where some removal must
eliminate envy, and envy-freeness up to any good (EFX), where every
removal must do so. EF1, formalized by
Budish~\cite{budish2011combinatorial}, can always be achieved for
monotone valuations using the envy-cycle method of Lipton et
al.~\cite{lipton2004approximately}. The corresponding existence
question for EFX, introduced by Caragiannis et al.~\cite{CKM16},
has proved substantially harder. We follow the convention of Plaut
and Roughgarden~\cite{Plaut-Roughgarden}, which tests the removal of
every good, including goods of zero marginal value; the original
additive definition tested positively valued goods.

Much of the progress on EFX existence has come from understanding
how far positive results extend as assumptions on valuations are
relaxed. Plaut and Roughgarden~\cite{Plaut-Roughgarden} proved
existence for identical monotone valuations and, using cut-and-choose,
for two agents with arbitrary monotone valuations. Moving from two
agents to three required new ideas: Chaudhury, Garg, and
Mehlhorn~\cite{Chaudhury-Garg-Mehlhorn:EFX} established existence for
three additive agents. Subsequent work showed that additivity could
be relaxed. Berger et al.~\cite{berger2021almost} extended the
three-agent result to a class that includes unit-demand and
budget-additive valuations.
For additive valuations, EFX is also known to exist for
any number of agents with at most three distinct valuation
types~\cite{PrakashGNV2025}.
Akrami et al.~\cite{AACGMM25} required
a structural condition on only one agent, allowing the other two
to have arbitrary monotone valuations. Their condition,
\emph{MMS-feasibility}, rules out two partitions of the same subset
of goods into two bundles such that every bundle in one partition
is strictly more valuable than every bundle in the other.
Their algorithm is the starting point for our positive results.

These extensions raise the question of whether the remaining
valuation restrictions are essential. Our counterexample answers
the general EFX existence question posed by Plaut and
Roughgarden~\cite[p.~1040]{Plaut-Roughgarden} negatively, already
for three agents with submodular valuations. In particular,
diminishing returns alone do not guarantee existence. The
weighted-coverage realization strengthens this conclusion: the
obstruction occurs when goods provide overlapping features and a
bundle's value is the total weight of the features it covers.
Our representation argument explains the connection between this
structured setting and general monotone valuations. EFX depends
on comparisons between bundles, and preserving the relevant
comparisons preserves the obstruction. The additive
existence question remains open for four or more agents.
Computer-assisted methods have recently been applied to
this case as well: Alkassar, Fouz, and
Mehlhorn~\cite{alkassar2026completeefxallocationsexist} show that EFX allocations exist for four additive agents and up to nine
goods. For a weaker notion, Ashuri, Gkatzelis, and
Sgouritsa~\cite{AshuriGkatzelisSgouritsa2024} showed that EF2X
allocations, in which envy disappears after the removal of any two
goods from the envied bundle, exist for four agents with cancelable
valuations.

The number of goods provides another way to understand where
existence can fail. Mahara~\cite{Mahara-EFX} proved that every
instance with monotone valuations and $m\leq n+3$ admits an EFX
allocation. For three agents, this guarantees existence through
six goods. Our seven-good theorem closes the remaining gap before
the eight-good counterexample, giving an exact threshold without
any additional restriction on the valuations. The connection to
Akrami et al.~\cite{AACGMM25} is also methodological: we refine
the analysis of their partition-repair algorithm to identify the
configurations that prevent it from producing EFX, and then rule
out those configurations when only seven goods are present.

Once exact EFX can fail, the extent to which it can be approximated
becomes a separate existence question. Plaut and
Roughgarden~\cite{Plaut-Roughgarden} established that 
$1/2$-EFX allocations always exist for subadditive valuations.
Chan et al.~\cite{ChanChenLiWu2019} gave a polynomial-time
algorithm achieving this guarantee together with EF1. These
results apply to submodular valuations as well and provide the
positive baseline for our quantitative results. Under additivity,
Amanatidis, Markakis, and
Ntokos~\cite{AmanatidisMarkakisNtokos2020} improved the universal
guarantee to $(\sqrt5-1)/2\approx0.618$, and Amanatidis, Filos-Ratsikas, and Sgouritsa~\cite{AmanatidisFilosRatsikasSgouritsa2024} obtained $2/3$-EFX for up to seven agents.
Our impossibility bounds
concern broader valuation classes and therefore leave open the
possibility of stronger guarantees for additive valuations.

A different response to the difficulty of EFX is to consider
\emph{partial} allocations, in which some goods may
remain unallocated. Caragiannis, Gravin, and
Huang~\cite{CaragiannisGravinHuang2019} showed that, for additive
valuations, a partial EFX allocation can preserve at least half
the optimal Nash welfare, the geometric mean of the agents'
utilities. Chaudhury et
al.~\cite{Chaudhury-Kavitha-MehlhornSgouritsa-EFX} instead
controlled the number of unallocated goods: for arbitrary
monotone valuations, at most $n-1$ goods need remain, with no
agent preferring the entire unallocated set to her own bundle.
Mahara~\cite{Mahara-EFX} improved this bound to $n-2$ for
$n\ge2$. Thus, for three agents, exact EFX is always achievable
after withholding at most one good. This makes completeness
particularly consequential for our results: the failure of
complete EFX with eight goods coexists with a guarantee that
exact EFX can be retained by leaving at most one good unassigned.
For four agents with nice-cancelable valuations, which include
additive ones, a single unallocated good also
suffices~\cite{berger2021almost}. Allowing approximation as well, $(1-\varepsilon)$-EFX allocations exist
with only sublinearly many unallocated
goods~\cite{ChaudhuryGMMM2021rainbow,BerendsohnBK2022,AACGMM25,JahanSSS2023}.

EFX has also been studied for chores, where agents incur
costs instead of receiving value. There, nonexistence was established
earlier: Christoforidis and
Santorinaios~\cite{ChristoforidisSantorinaios2024} gave a three-agent
instance with superadditive costs that admits no EFX allocation, nor
any bounded approximation of it, and He and Tao~\cite{HeTao2026}
recently showed that EFX allocations need not exist for additive
chores with any number $n\ge4$ of agents. Our results concern goods,
for which no counterexample was previously known.

\section{Preliminaries}\label{sec:preliminaries}

Let $N$ be a set of $n$ agents and $M$ a set of $m$ indivisible goods. A bundle is a subset of $ M $.
Agent $i$ has a valuation $v_i:2^M\to\mathbb R_{\ge0}$. All valuations
are \emph{normalized}, meaning $v_i(\varnothing)=0$, and
\emph{monotone}, meaning $v_i(S)\le v_i(T)$ whenever $S\subseteq T$.
A valuation is \emph{strictly increasing under inclusion} if
$v_i(S)<v_i(T)$ whenever $S\subsetneq T$. This does not forbid ties
between incomparable bundles. When we need all bundles to
have pairwise distinct values, we instead call the valuation
\emph{nondegenerate}. Under monotonicity, nondegeneracy implies strict increase under
inclusion. Neither additional property is assumed unless stated. 

 For a valuation $v$, we also use the notation $ A \preceq_v B $ and $ A \prec_v B $ as a shorthand for $ v(A) \le v(B) $ and $ v(A) < v(B) $, respectively, and write $ \succeq_v $ and $ \succ_v $ for the reverse relations. For an agent $ i $, we abbreviate these relations by $ \preceq_i, \prec_i, \succeq_i, \succ_i $, respectively. When the underlying valuation is clear from the context, we omit the subscript.

An \textit{allocation} $X=(X_1,\ldots,X_n)$ is an $n$-tuple of pairwise disjoint bundles whose union is $M$, where $X_i$ is the bundle assigned to agent $i$. Empty bundles are allowed, but unallocated goods are not. A partition $\fancyP=\{P_1,\ldots,P_k\}$ of a nonempty set $ S\subseteq M$ is a family of pairwise disjoint, nonempty bundles whose union is $S$, where $k\ge1$.
EFX allows envy of another
agent's bundle, but requires that every single-good removal eliminates
that envy. Its approximate version allows a multiplicative loss.

\begin{definition}\label{def:efx}
For $0<\alpha\le1$, an allocation $X$ is \emph{$\alpha$-EFX} if
\begin{equation}\label{eq:efx}
 v_i(X_i)\ge\alpha v_i(X_j\setminus\{g\})
 \qquad\text{for all }i,j\in N\text{ and }g\in X_j.
\end{equation}
We write EFX for $1$-EFX. Every good is eligible for deletion, including
goods of zero marginal value to the envying agent.
\end{definition}

A valuation is \emph{subadditive} if
$v(S\cup T)\le v(S)+v(T)$ for all $S,T\subseteq M$. It is
\emph{submodular} if it has diminishing marginal returns:
\[
 v(S\cup\{g\})-v(S)\ge v(T\cup\{g\})-v(T)
 \qquad(S\subseteq T\subseteq M,\ g\notin T).
\]
Normalized monotone submodular valuations are subadditive. Additive
valuations are the special case $v(S)=\sum_{g\in S}v(\{g\})$.

In a \emph{weighted-coverage} valuation, there is a finite set of
features $F$, each feature $f$ has a nonnegative weight $w_f$, and each
good $g$ covers a set $F_g\subseteq F$. The value of a bundle is
\[
 v(S)=\sum_{f\in\bigcup_{g\in S}F_g}w_f.
\]
Covering a feature twice gives no additional value, so every
weighted-coverage valuation is normalized, monotone, and submodular.
For goods divided into types and a set $R$ of types, we write
\[
 \chi_R(S)=
 \begin{cases}
 1,&\text{if $S$ contains a good of a type in $R$},\\
 0,&\text{otherwise}.
 \end{cases}
\]
A nonnegative linear combination of these indicators is therefore a
weighted-coverage valuation.

For an agent $ i $, a partition $ \fancyP $ of goods and bundle $ X \in \fancyP $, we say that $ i $ accepts $ X $ under $ \fancyP $ whenever $ i $ is EFX-feasible when receiving $ X $ under $ \fancyP $, that is, if $ X \succeq_i Y \setminus \{g\} $ for all $ Y \in \fancyP \setminus \{X\} $ and $ g \in Y $. Let $ \Acc_i(\fancyP) $ denote the set of bundles $ X \in \fancyP $ that are accepted by agent $ i $ under $ \fancyP $. For an agent $ c $ and a partition $ \fancyP $, we say that $ \fancyP $ is $ c $-safe if $ \Acc_c(\fancyP) = \fancyP $, i.e., if $ c $ accepts every bundle in $ \fancyP $. For a nonempty collection $ \fancyX $ of bundles and a nondegenerate valuation $ \prec $, let $ \min_{\prec} \fancyX $ denote the unique $ \prec $-least bundle in $ \fancyX $. For an agent $ i $, we analogously write $ \min_i\fancyX $ when the underlying order is $ \prec_i $. We define $ \max_\prec\fancyX $ analogously, and write $ \max_i\fancyX $ when the underlying order is $ \prec_i $. We say that a partition $ \fancyP $ is EFX with respect to a
nondegenerate valuation $ \prec $ if $ X \succ Y \setminus \{g\} $ holds for all distinct $ X, Y \in \fancyP $ and every $ g \in Y $.

Clearly, $ \Acc_i(\fancyP) \neq \varnothing $ for every $ i $ and $ \fancyP $. Indeed, this is because agent $ i $ always accepts its favorite bundle in $ \fancyP $ (assuming monotonicity). %

\section{A Counterexample to EFX}\label{sec:counterexample}
\label{sec:construction}\label{sec:conceptual-proof}

We build the counterexample in two stages. Preferences over pairs of items first force
every EFX allocation into one shape. That shape then tells us which
triples to make more valuable. The difficulty is not merely to destroy
the forced allocation: the increase in triple values must not create
a new EFX allocation.

\subsection{A three-agent, eight-good counterexample}
\subsubsection{Instance description}

Let
\[
 M=\{A_1,A_2,B_1,B_2,C_1,C_2,x,y\},
 \qquad N=\{a,b,c\}.
\]
We call $A_1,A_2$ the goods of type $A$, and similarly for $B,C$.
Associate agents $a,b,c$ with types $A,B,C$, respectively.
The two goods of each type are interchangeable in every valuation.
We therefore write $AB$ for a pair containing one good of each type,
and $BCCx$ for a bundle containing one $B$-good, both $C$-goods, and $x$.
Each occurrence of a type within a bundle denotes a distinct good.

The construction is guided by bundle sizes. An EFX violation requires just one
deletion that leaves a bundle more valuable than the envying agent's
own bundle. Thus a valuable pair can witness a violation from a bundle
of at least three goods, and a valuable triple can do so from a bundle
of at least four. We choose the
pair values so that every agent must receive at least two goods in
any  EFX allocation. With three agents and eight goods, this
leaves only two possible bundle-size patterns: $(2,3,3)$ and
$(2,2,4)$, where each tuple lists the three
bundle sizes in no particular order.

Raising triple values cannot destroy an EFX allocation of
shape $(2,3,3)$: every deletion leaves at most a pair, while each own
value can only increase. Consequently the pairs must already rule out
that shape. They will also force a particular $(2,2,4)$ allocation,
which the exceptional triples will destroy.

Each agent has three preferred pair types: its own type with $y$, the
other two types together, and its own type with $x$:

\begin{center}
\begin{tikzpicture}[every node/.style={font=\small},
    pair/.style={draw=black!65,rounded corners=2pt,minimum width=14mm,
      minimum height=6mm,line width=.5pt}]
  \foreach \row/\leftside/\middle/\rightside/\agent in
    {0/Ay/BC/Ax/a,1/By/AC/Bx/b,2/Cy/AB/Cx/c} {
    \node[pair] (l\row) at (0,-0.8*\row) {$\leftside$};
    \node[pair] (m\row) at (1.9,-0.8*\row) {$\middle$};
    \node[pair] (r\row) at (3.8,-0.8*\row) {$\rightside$};
    \node at ($(l\row)!0.5!(m\row)$) {$\succ$};
    \node at ($(m\row)!0.5!(r\row)$) {$\succ$};
    \node[anchor=west] at (4.7,-0.8*\row) {agent $\agent$'s preferred row};
  }
\end{tikzpicture}
\end{center}

Each agent prefers every pair in its own row to every pair in the
other two rows, whose corresponding entries are tied. Same-type pairs and $xy$ have the lowest
value. For agent $i$, let $T$ be its associated type and $U,V$ the other
two types:

\begin{equation}\label{eq:pair-ranks}
\begin{gathered}
\begin{array}{c|cccccc}
\text{pair}&Ty&UV&Tx&Uy,Vy&TU,TV&TT,UU,VV,Ux,Vx,xy\\ \hline
u_i&6&5&4&3&2&1
\end{array}
\end{gathered}
\end{equation}

Set $u_i(\varnothing)=0$ and $u_i(\{g\})=1$ for every good $g$. For
larger bundles, let $u_i$ be the value of the best pair. These valuations
are monotone. Pair values are specified directly rather than obtained
by adding singleton values.

The common comparisons make it costly to separate certain pairs; the
agents' priorities force each to obtain a pair in its own row. Their
alternatives compete for the same goods.

The final valuations differ from $u$ only on bundles containing one
designated triple type. For an agent $i$ associated with type $T$, let
$U,V$ be the other two types and set
\begin{equation}\label{eq:rank-definition}
 v_i(S)=
 \begin{cases}
  7,&\text{if $S$ contains a triple of type $UVx$,}\\
  u_i(S),&\text{otherwise.}
 \end{cases}
\end{equation}
This defines the valuation profile $v=(v_a,v_b,v_c)$ of our
counterexample. We retain the auxiliary profile $u=(u_a,u_b,u_c)$
to analyze the effect of the exceptional triples.
Thus the exceptional triple types are $BCx$ for $a$, $ACx$ for $b$,
and $ABx$ for $c$. Both $u_i$ and $v_i$ are normalized and monotone:
enlarging a bundle cannot decrease its best pair value, and a bundle
containing an exceptional triple continues to contain it after enlargement.
Every nonempty bundle has a value in $\{1,\ldots,7\}$.

\par\medskip\noindent\textbf{Relabelling symmetry.}\quad
The pair values and exceptional triples depend only on an agent's own
type and the other two types. We may therefore relabel the agents in
any order, provided that we relabel their associated types in the same
way and leave $x,y$ fixed. This preserves both valuation profiles and
all EFX comparisons. Whenever we relabel agents below, we also relabel
their associated types.

\subsubsection{Proof of non-existence of EFX allocations}

We first classify EFX allocations under the pair valuations $u$.
The following six-good observation captures why the competing pair
preferences cannot all be kept in different bundles.

\begin{lemma}\label{lem:six}
Let $P,Q,R$ be distinct types among $A,B,C$. Divide $P,P,Q,Q,R,x$
between any two agents, giving each at least two goods. If this division
is EFX between them under $u$, some bundle of at least three goods
contains $Px$ or $QR$.
\end{lemma}
\begin{proof}
Suppose neither pair occurs in a bundle of at least three goods.
Join the goods forming these pairs, using subscripts to distinguish copies:
\begin{center}
\begin{tikzpicture}[every node/.style={circle,draw=black!65,minimum size=6.5mm,
    inner sep=1pt,font=\small},line width=.5pt]
  \node (x) at (0,0) {$x$};
  \node (p1) at (-.9,.75) {$P_1$};
  \node (p2) at (.9,.75) {$P_2$};
  \draw (x)--(p1) (x)--(p2);
  \node (r) at (4,0) {$R$};
  \node (q1) at (3.1,.75) {$Q_1$};
  \node (q2) at (4.9,.75) {$Q_2$};
  \draw (r)--(q1) (r)--(q2);
\end{tikzpicture}
\end{center}
Joined goods cannot form an entire two-good bundle either: the other
pictured group would then lie in the four-good bundle and contain a
forbidden pair. Thus every line joins different bundles. In particular,
$x$ is separated from both $P$-goods, and $R$ from both $Q$-goods.

The bundle containing $x$ has no $P$ and cannot contain both $Q$ and $R$.
It therefore consists of $x$ and goods of one type, say $U$, and has the
same value as $Ux$. The other bundle contains both $P$-goods and the
remaining type. Removing one $P$ leaves a pair of the two types other
than $U$. Every agent prefers this pair to $Ux$, by the pair rankings.
The holder of $x$ therefore still envies the other bundle.
\end{proof}

\begin{proposition}[The forced allocation]\label{prop:shape}
Up to simultaneously relabelling types and agents, the EFX allocations
for $u$ are exactly
\begin{equation}\label{eq:forced-shape}
\underbrace{Ay}_{\text{agent }a}
\quad\big|\quad
\underbrace{BCCx}_{\text{agent }b}
\quad\big|\quad
\underbrace{AB}_{\text{agent }c}.
\end{equation}
\end{proposition}
\begin{proof}
\par\medskip\noindent\textbf{The possible bundle sizes.}\quad
Every agent must receive at least two goods. Otherwise its value is
at most $1$, while another bundle has at least four goods. Such a bundle
contains either two distinct types from $A,B,C$, or only one such type,
in which case it is $TTxy$. Either way, a pair worth more than $1$
to every agent survives a removal, violating EFX. Only the size patterns
$(2,3,3)$ and $(2,2,4)$ remain.

\par\medskip\noindent\textbf{The bundle containing $y$.}\quad\label{sec:y-forces-pair}
Relabel so that agent $a$ receives $y$, and write $Y$ for its bundle.
If $|Y|\ge3$, it contains no
$B$- or $C$-good: the corresponding agent would face its favorite pair,
$By$ or $Cy$, of value $6$, after a removal, but could obtain at most $5$
without $y$. Thus $Y$ uses only $A,A,x,y$ and takes at least two of
$A,A,x$. It also contains $Ay$ after a removal, worth $3$ to each other
agent. Without $y$, matching this value requires agent $b$ to have $Bx$ or
$AC$, and agent $c$ to have $Cx$ or $AB$. Each needs a good from $A,A,x$,
but at most one remains. Hence $|Y|=2$.

If $Y=xy$, the six remaining goods force a mixed-type pair into a bundle
of at least three goods. This pair survives a removal and beats
$u_a(xy)=1$. If $Y=By$, the remaining goods are $A,A,C,C,B,x$.
Lemma~\ref{lem:six}, with $(P,Q,R)=(A,C,B)$, puts $Ax$ or $BC$ in a
bundle of at least three goods. Agent $a$ then faces value $4$ or $5$
after a removal, above its own $3$. Exchanging agents $b,c$ excludes
$Cy$ as well.
Therefore agent $a$ receives exactly $Ay$.

\par\medskip\noindent\textbf{The remaining goods.}\quad
The other agents divide $B,B,C,C,A,x$. If agent $b$ receives neither $Bx$
nor $AC$, its value is at most $2$. Neither pair is in its own bundle,
and EFX forbids either in the other bundle together with a third good.
This contradicts Lemma~\ref{lem:six}, with $(P,Q,R)=(B,C,A)$.
Thus agent $b$ must have $Bx$ or $AC$; similarly, agent $c$ must have $Cx$ or $AB$.

Both alternatives involving $x$ require the unique $x$. Both alternatives
without $x$ require the single remaining $A$-good. One agent must take
each kind. Exchanging agents $b,c$ if necessary, agent $b$ has $Bx$ and
agent $c$ has $AB$.
Only the two $C$-goods remain to be assigned:
\begin{center}
\begin{tikzpicture}[>={Stealth[length=1.7mm,width=1.3mm]},
    every node/.style={font=\small},
    bundle/.style={draw=black!65,rounded corners=2pt,minimum width=27mm,
      minimum height=9mm,line width=.5pt}]
  \node[bundle] (a) at (0,0) {agent $a$: $Ay$};
  \node[bundle] (b) at (4,0) {agent $b$: $Bx$};
  \node[bundle] (c) at (8,0) {agent $c$: $AB$};
  \node (spare) at (4,1.2) {unassigned: $C,\ C$};
  \draw[->,line width=.8pt] (spare) -- (b);
  \node[align=center,anchor=north] at (0,-.65) {$Ay$ must remain\\a two-good bundle};
  \node[align=center,anchor=north] at (4,-.65) {both $C$-goods\\must go here};
  \node[align=center,anchor=north] at (8,-.65) {adding $C$ leaves $AC$\\after removing $B$};
\end{tikzpicture}
\end{center}
Neither can join $Ay$, whose size is fixed. If one joins $AB$, removing
the $B$-good leaves an $AC$ pair, worth $5$ to agent $b$. Yet agent $b$'s own bundle,
consisting of $Bx$ and possibly a $C$, is worth only $4$.
Both $C$-goods must therefore go to agent $b$.

Conversely, $(Ay,BCCx,AB)$ is EFX under $u$. Agent $a$ has value $6$
and sees at most $5$ after a removal from $BCCx$; agent $c$ has value
$5$ and sees at most $4$ there. Every other comparison leaves a
singleton of value $1$. This proves the converse and completes the classification.
\end{proof}

\par\medskip\noindent\textbf{The forced allocation determines the exceptional triples.}\quad
Proposition~\ref{prop:shape} rules out $(2,3,3)$ and forces the displayed
$(2,2,4)$ allocation. In that allocation, deleting one $C$-good from
$BCCx$ leaves $BCx$. Promoting this triple to value $7$ makes it better
than $Ay$, whose value remains $6$. Applying the same rule to the other
two agents gives the exceptional triples $ACx$ for $b$ and $ABx$ for $c$.
It remains to show that these increases create no new EFX allocation.

\par\medskip\noindent\textbf{Exceptional triples cannot restore EFX allocations.}\quad    
Every exceptional triple contains the unique good $x$. If an agent
receives its exceptional triple, the other two agents must therefore
still be satisfied by their pair values.

\begin{lemma}\label{lem:triple}
In an EFX allocation for $v$, no agent receives a bundle containing its
exceptional triple. Consequently, the allocation is also EFX for $u$.
\end{lemma}
\begin{proof}
Suppose, after relabelling, that agent $a$ receives a bundle containing
its exceptional triple $BCx$.
After removing a $C$, that bundle still contains $Bx$, worth $4$ to
agent $b$; after removing a $B$, it still contains $Cx$, worth $4$ to
agent $c$. Neither other agent has $x$, so
neither can obtain its own $x$-pair or its exceptional triple. Hence
\[
\text{Agent }b\text{ must have }By\text{ or }AC,
\qquad
\text{agent }c\text{ must have }Cy\text{ or }AB.
\]
At most one $B$-good, one $C$-good, and one $y$ remain outside agent $a$'s bundle.
Giving $By$ to agent $b$ would exclude both alternatives for agent $c$: it takes
the $y$ needed for $Cy$ and the last $B$ needed for $AB$.
Thus agent $b$ must receive $AC$; symmetrically, agent $c$ must receive $AB$.
These pairs and agent $a$'s triple use every good except $y$:
\[
\underbrace{BCx}_{\text{agent }a}
\quad\big|\quad
\underbrace{AC}_{\text{agent }b}
\quad\big|\quad
\underbrace{AB}_{\text{agent }c}
\qquad\text{with }y\text{ still to place}.
\]
There is no possible destination for $y$. Each displayed bundle
contains the type associated with an agent among $b,c$ other than its recipient.
Adding $y$ therefore creates that other agent's favorite pair,
$By$ or $Cy$, which survives a removal. That agent retains its own
pair of value $5$ but faces a pair of value $6$, contradicting EFX.

Thus no agent's own value can increase from $u$ to $v$ in an EFX
allocation for $v$. Other bundles' values can only increase, so
\[
u_i(X_i)=v_i(X_i)\ge v_i(X_j\setminus\{g\})
                    \ge u_i(X_j\setminus\{g\}).
\]
The allocation is therefore also EFX for $u$.
\end{proof}

\begin{theorem}\label{thm:ordinal-counterexample}
The three monotone valuations $v_a,v_b,v_c$ admit no  EFX
allocation.
\end{theorem}
\begin{proof}
Lemma~\ref{lem:triple} and Proposition~\ref{prop:shape} force any such
allocation to be $(Ay,BCCx,AB)$, up to relabelling. Removing a $C$-good
from agent $b$'s bundle leaves $BCx$, and
\[
v_a(BCx)=7>6=v_a(Ay),
\]
contradicting EFX.
\end{proof}

This completes the proof that the instance admits no EFX allocation. The construction works by choosing pair preferences that leave a unique EFX allocation, up to relabelling, and then making selected triples more valuable to rule it out. The difficulty is ensuring that this change does not create a different EFX allocation. Here the goods \(x\) and \(y\) play complementary roles. Every exceptional triple contains \(x\), so at most one agent can receive its exceptional triple. If an agent does, satisfying the EFX constraints for the other two agents forces them to take particular pairs, after which every placement of \(y\) violates EFX. Thus the exceptional triples eliminate the allocation forced by the pair preferences without creating an alternative.

\subsubsection{Minimality of the eight-good counterexample}

We complement the eight-good negative result with a constructive proof that every fair division instance with three agents, at most seven goods, and monotone valuations admits an EFX allocation, establishing that eight is the minimum number of goods permitting a counterexample:

\threeseventhm*

Our proof begins with a novel structural reduction based on a refined analysis of the algorithm of \cite{AACGMM25} that constructs an EFX allocation whenever at least one agent has an MMS-feasible valuation. We retain its basic invariant-based framework but give a more detailed analysis of the combinatorial properties arising when the MMS-feasibility assumption is dropped. This yields a reduction that applies to any number of goods: it either constructs an EFX allocation or identifies a \textit{critical conflict} between every pair of agents.

Intuitively, agents $ p $ and $ q $ are in a conflict if, for some bundle $ B $, they each have a division of the remaining goods $ S := M \setminus B $ into two parts that is EFX according to their own valuation, yet their assessments of these divisions are sharply opposed. Each agent prefers either part of their own division to $ B $, but prefers $ B $, even after a suitable good is removed, to either part of the other agent's division. Thus, each proposes a division whose two parts they regard as valuable and fair, while the other regards both parts as too poor to receive alongside $ B $. The conflict is \textit{critical} when these incompatible divisions are additionally linked to the preferences of the third agent $ c $. Specifically, $ S $ admits another division into bundles $ A $ and $ D $ such that, under the partition $ \{A, B, D\} $, agent $ c $ can receive either $ A $ or $ B $ without strong envy, whereas $ p $ and $ q $ both regard $ D $ as their only EFX-feasible choice. Moreover, $ c $ prefers $ A $ to every part of either division proposed by $ p $ and $ q $.

The reduction exposes this structure by repeatedly attempting to repair a three-agent cut-and-choose allocation attempt. It maintains two bundles that are EFX-feasible for $ c $ (the ``cutter'' agent) and a third that is EFX-feasible for at least one other agent, while strictly increasing the value to $ c $ of the less valuable of the two bundles accepted by $ c $. Our analysis shows that either these repairs produce an EFX allocation or their failure yields the incompatible divisions and additional comparisons defining a critical conflict. Since $ c $ can be chosen arbitrarily, every counterexample must exhibit such a conflict between every pair of agents.

Our reduction also implies EFX existence for three agents and any number of goods whenever one agent has a leveled valuation.\footnote{A valuation $ v $ is \emph{leveled} if $ |A| < |B| $ implies $ v(A) < v(B) $ for all bundles $ A, B $.} This strengthens the known guarantee for three-agent profiles in which all agents are leveled \cite{christodoulou2025fairtruthfulallocationsleveled}: the other two agents may have arbitrary monotone valuations.

As a next step, we observe that a critical conflict cannot exist with fewer than seven goods, and prove that with exactly seven goods, the sets must have a very rigid combinatorial structure, most notably the intersection pattern of the two competing divisions of $ S $ must be unique up to relabeling of goods. In particular, we observe that $ |B| = 2 $. Using this structure, we derive a way of systematically constructing EFX allocations using the conflicting agents' proposed divisions. More precisely, we first try to construct an EFX allocation using $ B $ and the proposed division directly. If these allocations are not EFX, we instead give one agent a triple from their own fair division, and distribute the remaining four goods into two pairs, each containing one good of $ B $. We show that whenever the agents disagree about which good of $ B $ is more valuable, the triple recipient and the pair assignments can be chosen to yield an EFX allocation. Failure of these attempts forces all three agents to regard the same good $ b \in B $ as their uniquely most valuable good among all seven goods. Crucially, this uniqueness makes $ \{b\} $ a convenient threshold for comparing the bundles involved in a conflict and, even more importantly, suitable for linking the three conflicts together. By combining this with the rigid intersection structure, we show that the conflicts contradict each other, thus completing the proof of EFX existence for seven goods.

We present the reduction in Section~\ref{section:reduction_to_conflicts}, outline the proof of Theorem~\ref{thm:intro-seven} in Section~\ref{section:efx_existence_up_to_seven}, and present the details in Appendix~\ref{section:efx_existence_up_to_seven_appendix}.

\subsection{Extensions to subclasses and larger instances}
\label{sec:extensions}

The eight-good example is the smallest obstruction for three agents.
We now extend it in two directions: to more agents and goods, and to
weighted-coverage valuations. The first extension prevents extra agents
from dispersing the original obstruction. The second changes the numerical
values while preserving every strict comparison responsible for envy.

\subsubsection{More agents and more goods}
We extend the counterexample to more than three agents. More precisely,
we first extend it to $n \ge 4$ agents and $m=n+5$ goods, and then
to $m \ge n+5$ goods. The same construction also applies to approximate
EFX, provided the new goods are sufficiently valuable. We extend the counterexample to more agents and more goods. More
precisely, we show that any counterexample with $n_0$ agents and $m_0$
goods extends to $n\ge n_0$ agents and $m\ge m_0+n-n_0$ goods. The same
construction also applies to approximate EFX, provided the new goods are
sufficiently valuable.

\begin{lemma}[Adding agents and goods]\label{lem:padding}
Let $\alpha\in(0,1]$. Suppose an $n_0$-agent instance on $m_0 \geq n_0+1$ goods
has normalized monotone valuations, every singleton has positive value
for every agent, and no  $\alpha$-EFX allocation exists.
Then, for every $n\ge n_0$ and $m\ge m_0+n-n_0$, there is an $n$-agent,
$m$-good instance with no $\alpha$-EFX allocation.
If the original valuations are subadditive, submodular, weighted
coverage, or additive, the new valuations remain in the same class.
\end{lemma}

\begin{proof}
    Let $N_0$ and $G$ denote the agents and goods of the base instance,
    with valuations $(v_i)_{i\in N_0}$, and let $k=n-n_0$.

    We first consider $n$ agents and $m=m_0+k$ goods.
    Fix an agent $c\in N_0$. The agents are those in
    $N_0\setminus\{c\}$ together with $k+1$ copies
    $c_1,\ldots,c_{k+1}$ of agent $c$.
    The goods are the $m_0$ goods of $G$ and $k$ new goods
    $H=\{h_1,\ldots,h_k\}$.

    The valuations are as follows. We use $v_i$ for the
    valuations of the base instance and $\bar v_i$ for
    the valuations of the extended instance.
    Let $L>v_c(G)/\alpha$. Then, for every $S\subseteq G\cup H$,
    \begin{align*}
      \bar v_i(S)
        &=v_i(S\cap G)
        &&\text{for }i\in N_0\setminus\{c\},\\
      \bar v_{c_t}(S)
        &=L|S\cap H|+v_c(S\cap G)
        &&\text{for }t\in\{1,\ldots,k+1\}.
    \end{align*}

    Note that the goods in $H$ have no value to the agents in
    $N_0\setminus\{c\}$, and that agents $c_1,\ldots,c_{k+1}$ have
    identical valuations.
    Also, subsets of $G$ have the same value as in the base
    instance, and, for any copy of agent $c$, any subset
    containing an element of $H$ has value larger than any subset
    of $G$.

    Assume that an $\alpha$-EFX allocation $(X_j)_j$ exists.
    Since $H$ consists of exactly $k$ goods, there must be an
    agent $c_{t^*}$ that owns no good in $H$.
    Thus $X_{c_{t^*}}\subseteq G$ and hence
    \[
        \bar v_{c_{t^*}}(X_{c_{t^*}})
        \le v_c(G)<\alpha L.
    \]
    Consider now any $X_j$ with $X_j\cap H\neq\varnothing$.
    If $X_j$ has cardinality two or more, there is a good
    $g\in X_j$ such that $X_j\setminus\{g\}$ still contains
    an element of $H$. Consequently,
    \[
        \bar v_{c_{t^*}}(X_j\setminus\{g\})
        \ge L,
    \]
    and the allocation is not $\alpha$-EFX.
    Thus, there are $k$ agents each owning one of the items
    in $H$ and nothing else. The remaining $n_0$ agents own
    the items in $G$, and since $m_0\ge n_0+1$, one of them owns
    at least two items in $G$.

    Suppose now that there is an agent $i\in N_0\setminus\{c\}$ that
    owns one of the items in $H$ and nothing else.
    Then $X_i$ has value zero for $i$. Since every singleton
    in $G$ has positive value, agent $i$ strongly envies the
    agent owning two or more items in $G$, even after the
    removal of one of those items. This contradicts
    $\alpha$-EFX.

    We conclude that the agents in $N_0\setminus\{c\}$ together
    with $c_{t^*}$ share the goods in $G$ between them. Thus,
    their bundles would form an $\alpha$-EFX allocation for the
    base instance, a contradiction.

    The example readily extends to more goods by adding dummy
    items of value zero. Indeed, let $I_0=G\cup H$ be the set
    of the $m_0+k$ goods in the example above, and let
    $(\bar v_i)$ be the extended valuations. Let $I$ be any
    superset of $I_0$ and extend each valuation to $I$ by
    defining
    \[
        \bar v_i(S)=\bar v_i(S\cap I_0)
        \qquad\text{for every }S\subseteq I.
    \]
    Then there is no $\alpha$-EFX allocation of $I$, as it
    would induce an $\alpha$-EFX allocation of $I_0$.
    For $n=n_0$, the same argument simply adds dummy goods
    to the original instance.

    Finally, restricting a valuation to $G$ and adding a
    nonnegative additive term preserve subadditivity,
    submodularity, and additivity. For weighted-coverage
    valuations, each new good in $H$ can be assigned its own
    feature of weight $L$ for the copies of agent $c$, and no
    features for the other agents. Dummy goods cover no features.
    Thus the construction preserves all the stated
    valuation classes.
\end{proof}

Taking $\alpha=1$ applies this extension directly to our monotone
counterexample, whose singleton values are all one. In this case,
we can simply choose $L=v_c(G)+1$.

\begin{corollary}\label{cor:wc-general}
For every $n\ge3$ and every $m\ge n+5$, there is an instance with $n$
agents and $m$ goods, with monotone valuations, that admits no
EFX allocation.
\end{corollary}

The generalization to approximate EFX is immediate: the only
difference is that a new good must be worth more than
$v_c(G)/\alpha$, rather than merely $v_c(G)$. This ensures
that an agent receiving no new good cannot satisfy
$\alpha$-EFX if another agent bundles a new good with
anything else.

For later use, the extension can also preserve an entire
approximation barrier in a single instance. If a fixed base
instance rules out every $\alpha>\rho>0$, choose
$L>v_c(G)/\rho$ once. Then $\alpha L>v_c(G)$ for every
$\alpha>\rho$, so the same extended instance rules out all
these factors simultaneously.

\color{black}

\subsubsection{From monotone valuations to weighted coverage}

We first perturb the valuations so that they are strictly increasing under inclusions. We then replace the resulting
valuation by a weighted-coverage function with exactly the same bundle
ordering. The first step need only preserve strict comparisons; that is
enough because every failure of EFX is witnessed by a strict comparison.

We use a symmetry-preserving variant of the standard perturbation
argument; see \cite[Appendix~A, Lemma~5]{akrami2022efx}.
For a sufficiently small common $\varepsilon>0$, replace each
valuation $v_i$ by
\[
  v_i^\varepsilon(S)=v_i(S)+\varepsilon|S|.
\]
The perturbed valuations are normalized and strictly increasing
under inclusion, and preserve every original strict comparison.
Indeed, there are finitely many bundle comparisons, and the
perturbation changes each difference by at most $\varepsilon|M|$.
Consequently, every EFX allocation for the perturbed profile would
also be EFX for the original profile, so nonexistence is preserved.
Since the added term depends only on bundle size and is the same
for every agent, the perturbation also preserves every symmetry
of the original profile under simultaneous relabeling of agents
and goods.
Strict increase under inclusion does not require all distinct bundles
to have distinct values. In particular, it is compatible with our
symmetry between goods of the same type. The next representation
preserves ties as well as strict comparisons. It is an explicit finite
form of the classical ordinal coverage representation of
Kreps and Nehring~(\cite{Kreps,Nehring}); we give the proof to make both
the normalization and its application to EFX transparent.

\begin{lemma}[Ordinal weighted-coverage representation]
\label{prop:coverage-representation}
Let $v$ be strictly increasing under inclusion on a finite set $M$.
There is a normalized weighted-coverage valuation $w$ such that, for
all bundles $S,T$,
\[
 w(S)\le w(T)\quad\Longleftrightarrow\quad v(S)\le v(T).
\]
Relabelling goods before applying the construction gives the same
result as relabelling them afterwards.
\end{lemma}
\begin{proof}
Suppose $m:=|M|>0$; the empty-ground-set case is immediate. List the
distinct values of $v$ in increasing order as $t_0<\cdots<t_q$, and
define $k(S)$ by $v(S)=t_{k(S)}$. Thus $k(\varnothing)=0$ and
$k(M)=q$. Set $B=2^m$ and $F(S)=B^{q-k(S)}$.
For each proper subset $T\subsetneq M$, introduce a feature covered
by exactly the goods in $M\setminus T$, with weight
\[
 \omega_T=\sum_{R\supseteq T}(-1)^{|R\setminus T|}F(R).
\]
These weights are positive. Whenever $R\supsetneq T$, strict increase
under inclusion gives $k(R)\ge k(T)+1$, and hence
\[
 \omega_T\ge B^{q-k(T)}
 -(2^{m-|T|}-1)B^{q-k(T)-1}>0.
\]

For the inversion identity only, put $\omega_M=F(M)=1$; this is not a
feature. Inclusion--exclusion gives
\[
 F(S)=\sum_{T\supseteq S}\omega_T:
\]
the coefficient of $F(R)$ after expanding the right-hand side is
$\sum_{S\subseteq T\subseteq R}(-1)^{|R\setminus T|}$, equal to one
if $R=S$ and zero otherwise. The total feature weight is
$F(\varnothing)-1$. A bundle $S$ leaves feature $T$ uncovered exactly
when $S\subseteq T$, so its uncovered weight is $F(S)-1$. Its coverage
value is consequently
\begin{equation}\label{eq:ordinal-coverage}
 w(S)=B^q-B^{q-k(S)}.
\end{equation}
This vanishes at the empty set and is strictly increasing in $k(S)$,
proving the required equivalence. Relabeling goods permutes their
subsets and the associated features, leaving the construction unchanged
up to that permutation.
\end{proof}

The two steps transfer every monotone EFX counterexample to weighted
coverage on the same agents and goods. For the three-agent example,
strictification uses the same $\varepsilon$ for every agent, and the
weighted-coverage construction does not depend on the names of the goods.
Both steps therefore preserve the symmetry under simultaneous
relabelling of agents and their associated types. Additional null goods
may be left unchanged by these relabellings before applying the two steps.

Combining the extensions above gives the following general form of our nonexistence result.

\mainextended*

\begin{proof}
For the general statement, apply Lemma~\ref{lem:padding} with
$\alpha=1$ to Theorem~\ref{thm:ordinal-counterexample}, then apply
the strictification described above followed by
Lemma~\ref{prop:coverage-representation}.
\end{proof}

\section{Barriers to Approximate EFX}\label{sec:quantitative}

Exact nonexistence leaves a quantitative question: how close to EFX can an allocation always be? For monotone valuations, the lack of constraint between different levels of the valuation function means non-existence can be turned into arbitrarily bad inapproximability.\footnote{Indeed, every nonempty bundle in our counterexample has a value in
$\{1,\ldots,7\}$. For $L>1$, replace each such value $k$ by $L^{k}$;
this preserves monotonicity and every bundle comparison. Every allocation
still has a strict EFX violation, and its two values now differ by a
factor of at least $L$. Hence no allocation is $\alpha$-EFX for any
$\alpha>1/L$, and $L$ can be made arbitrarily large.} Nontrivial approximation guarantees therefore require additional structure on the valuations. For monotone subadditive valuations,
$\frac{1}{2}$-EFX allocations always exist, for any number
of agents and goods~(\cite{Plaut-Roughgarden,BS2026}). This guarantee
also applies to submodular valuations. 

We now investigate how far
the approximation factor can be improved within these two classes.
We first turn the eight-good counterexample
into a subadditive approximation barrier. We then give a separate
weighted-coverage construction that yields a barrier for submodular
valuations. Finally, a positive result for three agents shows that some
fairness can be recovered while keeping two agents exactly EFX. 

\subsection{The subadditive barrier}
\label{sec:subadditive-bound}

The counterexample tells us that every allocation leaves some agent
preferring another bundle even after a good is removed. To rule out
approximate EFX, we must also control \emph{how much} that agent prefers
it. We do so by transforming the valuations while preserving all
bundle comparisons.

We replace the seven positive values $1,\ldots,7$ by a geometric
sequence between $1/2$ and $1$. This ensures subadditivity: two nonempty
bundles have total value at least $1$, enough to pay for their union.
The six gaps between consecutive values distribute this factor of two
equally, giving the following construction.

Set
\[
 \eta=2^{-1/6}\approx 0.890899.
\]
For each agent $i$, let $v_i$ be its valuation in the eight-good
counterexample of Section~\ref{sec:counterexample}, and define
\begin{equation}\label{eq:geometric-realization}
 \widetilde v_i(\varnothing)=0,
 \qquad
 \widetilde v_i(S)=\eta^{\,7-v_i(S)}\quad\text{for }S\ne\varnothing.
\end{equation}
These valuations are monotone because $v_i$ is monotone and the
transformation is increasing.

\begin{theorem}\label{thm:subadditive-barrier}
For every $n\ge3$ and $m\ge n+5$, there are normalized, monotone
subadditive valuations for $n$ agents and $m$ goods admitting no 
$\alpha$-\EFX{} allocation for any $\alpha>\eta$.
\end{theorem}
\begin{proof}
Every nonempty bundle has value between
$\eta^6=1/2$ and $1$. If $S$ and $T$ are nonempty, then
\[
 \widetilde v_i(S\cup T)\le 1\le \widetilde v_i(S)+\widetilde v_i(T).
\]
If either set is empty, subadditivity holds with equality.

Consider any  allocation. Theorem~\ref{thm:ordinal-counterexample} gives an
agent $i$, another agent $j$, and a good $g\in X_j$ with
\(v_i(X_i)<v_i(X_j\setminus\{g\})\).
The right-hand bundle has positive value. If $X_i$ is empty, its
comparison ratio is zero. Otherwise the original values are integers
and differ by at least one, so
\[
 \frac{\widetilde v_i(X_i)}{\widetilde v_i(X_j\setminus\{g\})}
 =\eta^{\,v_i(X_j\setminus\{g\})-v_i(X_i)}
 \le\eta.
\]
This violates $\alpha$-\EFX{} whenever $\alpha>\eta$.
The eight-good profile has positive singleton values. Applying
the uniform choice of $L$ described after Lemma~\ref{lem:padding}
extends the obstruction simultaneously for all $\alpha>\eta$ to every
$n\ge3$ and $m\ge n+5$, without changing the factor.
\end{proof}

Thus the same combinatorial proof gives a quantitative barrier: every
strict EFX violation for $v$ becomes a comparison with ratio at most
$\eta$ for $\widetilde v$.

\subsection{The submodular barrier}
\label{sec:submodular-bound}

The preceding construction makes subadditivity easy by placing all
nonempty values between $1/2$ and $1$. That argument does not establish
submodularity: diminishing returns constrains how values change when
goods are added, not just their overall range. We therefore build the
next valuations directly as weighted coverage functions, so that
submodularity holds by construction.

The goal is again to leave a fixed amount of envy in every allocation.
We keep the three agents and eight goods $A,A,B,B,C,C,x,y$, but choose
new numerical values. The threshold moves closer to one, but the
obstruction now holds within the more restrictive class of submodular
valuations.

\begin{theorem}\label{thm:submodular-bound-sec}
For every $n\ge3$ and $m\ge n+5$, there are normalized weighted-coverage
valuations, and hence monotone submodular valuations, for which no 
allocation is $\alpha$-EFX for any $\alpha>\lambda_* \approx0.957455$, where $\lambda_*$ is the unique positive real root of $4z^4+4z^3-3z-4=0$. 
\end{theorem}

The design principle is to balance the competing ways an allocation
can fail. If one comparison is much less severe than the others, an
allocation may escape through it. We therefore arrange for the key
comparisons to have the same ratio. Write $w_a$ for the valuation of
the agent associated with type $A$. The values satisfy two chains:
\begin{align*}
 w_a(Bx)&=\lambda_*w_a(AB),&
 w_a(AB)&=\lambda_*w_a(By),&
 w_a(By)&=\lambda_*w_a(Ax),\\
 w_a(ABx)&=\lambda_*w_a(BC),&
 w_a(BC)&=\lambda_*w_a(Ay),&
 w_a(Ay)&=\lambda_*w_a(BCx).
\end{align*}
For example, agent $a$ holding $Ay$ fails any factor above $\lambda_*$
if another bundle still contains $BCx$ after a deletion. Each displayed
equality describes such a possible failure. The other agents have the
corresponding values after relabeling the capital types. The proof must
show that the available goods prevent all three agents from avoiding
these failures at once; the equalities alone are not enough.

The proof follows the same sequence as the qualitative construction.
We first rule out an agent receiving $x$ together with both types other
than its own: the other agents would then require particular pairs,
leaving nowhere to put $y$. Next, $y$ must form a pair with its
recipient's associated type. The six-good argument then forces the
remaining goods into $(Ay,BCCx,AB)$, up to relabelling, and the
comparison above rules out this allocation. Appendix~\ref{app:submodular-bound}
gives the coverage weights and verifies each step at the approximation
threshold $\lambda_*$.

The balancing also explains the otherwise unmotivated quartic. The coverage weights make
five of the six displayed equalities hold automatically. The remaining
equality,
$w_a(By)=\lambda_*w_a(Ax)$, requires
\[
 \frac{2+\lambda_*}{2(1+\lambda_*)}=2\lambda_*^3-1,
\]
which gives the stated quartic. Thus the constant comes from making the
competing ways to fail equally severe, subject to the chosen coverage
representation.

The obstruction does not rely on goods having zero marginal value.
For each fixed
$\alpha>\lambda_*$, a sufficiently small additive perturbation makes
every marginal strictly positive while preserving the counterexample.
The perturbation may depend on $\alpha$.

\par\medskip\noindent\textbf{Beyond the eight-good instances.}\quad
Both approximation barriers in this section extend to every $n\ge3$ and
$m\ge n+5$, without loss in their constants by Lemma~\ref{lem:padding}.

\subsection{A positive guarantee}\label{sec:positive}

The counterexample shows that we cannot always eliminate envy after every
one-good deletion. How much envy must we allow? For three agents, we can
confine the entire loss to one designated agent, provided that agent's
valuation is subadditive. The other two agents receive exact EFX and need
only have monotone valuations.

\mixedapprox*

In particular, three subadditive agents always admit a 
$2^{- \frac{1}{3}}$-EFX allocation, here $2^{- \frac{1}{3}}\approx 0.7937$. The theorem is stronger than this uniform
guarantee: the agent allowed an approximation loss can be designated in
advance, while the other two agents receive exact EFX even when their
valuations are not subadditive. There is no restriction on the number of
goods. The argument uses finite extremal searches; we claim no
polynomial-time bound.

\par\medskip\noindent\textbf{The proof idea.}\quad
We run the reduction of
Section~\ref{section:reduction_to_conflicts}, our refined analysis of the
algorithm of~\cite{AACGMM25}, with exact EFX for $p$ replaced by
$\alpha$-EFX: agent $c$ keeps two bundles $A$ and $B$ that it accepts,
and we try to assign the third bundle $D$ to $p$ or $q$.
This approach
uses the identical-valuation rebalancing principle underlying
Plaut and Roughgarden's EFX existence result~\cite[Theorem~4.2]{Plaut-Roughgarden}.
Protected bundles, a potential measuring their smaller value, trimming,
and a second fair cut are existing ingredients. The earlier exact
theorem permits one MMS-feasible agent and two arbitrary monotone
agents~\cite[Theorem~2 and Section~4 of the 2022 version]{akrami2022efx}.
For an MMS-feasible valuation, the minimum value in one two-way partition
of any set never exceeds the maximum value in another. This condition
already includes additivity, but does not hold for all subadditive
valuations. Our additional step extracts quantitative constraints from the conflict at which the reduction fails
and uses subadditivity to allow loss for just one agent.
The analysis applies to any number of goods. The full proof is in
Appendix~\ref{app:positive}.

The three agents play the roles they have in a
critical conflict (Definition~\ref{definition:critical_conflict}).
Agent $c$ is the third agent: it accepts both $A$ and $B$, so assigning
one of them to somebody else still leaves $c$ satisfied. Agents $p$ and
$q$ are the conflicting agents: both accept only $D$, and $q$ provides
its fair view of the conflict region $A\cup D$.
Agent $p$, the only
agent allowed an approximation loss, supplies the subadditivity inequality
that makes the argument work.

Fix a proposed approximation factor $\alpha$, and suppose no
allocation meets the required guarantees.
Then the reduction stops at a conflict between
$p$ and $q$ around the pivot bundle $B$
(Lemma~\ref{lem:approx-reduction}), and $p$ pays a factor $1/\alpha$
twice. Since $p$ and $q$ compete for $D$, agent $p$ values $D$ minus
some good $g$ above $v_p(B)/\alpha$. Moving $g$ to $A$ does not resolve
this competition, so $p$ values $A\cup\{g\}$ above $v_p(B)/\alpha^2$.
On the other side of the conflict, $p$ values each block of $q$'s fair
view below $\alpha\,v_p(B)$; otherwise giving one block to $p$, the
other to $q$, and the pivot $B$ to $c$ would let the reduction continue.
But these two cheap blocks together contain
the valuable bundle $A\cup\{g\}$. Subadditivity would require
$v_p(B)/\alpha^2<2\alpha\,v_p(B)$, which is impossible at
$\alpha=2^{-\frac{1}{3}}$.

This comparison explains both the constant and an important choice in the
proof.
Unlike in a conflict, we use only $q$'s fair view
and do not replace $A\cup\{g\}$ by a fair view of $p$: rebalancing
$A\cup\{g\}$ would lose the stronger lower bound on its value.
The
approximation guarantee comes from retaining this quantitative information
about the conflict.

\section{Reducing EFX to conflicts between agents}
\label{section:reduction_to_conflicts}

We start by recalling the following well-known results.%

\begin{proposition}[\cite{Plaut-Roughgarden}]
    \label{prop:PR}
    Let $ \prec $ be a nondegenerate monotone valuation and let $ \fancyP = \{X_1, \dots, X_k\} $ be a partition of a finite set of goods into $ k \ge 2 $ bundles. Then the
    following procedure terminates at a partition $ \fancyQ = \{Y_1, \dots, Y_k\} $ that is EFX with respect to $ \prec $.
    
    As long as the current partition is not EFX with respect to $ \prec $, let $ L $ be the least valuable bundle. Choose bundles $ X,Y $ and a good $ g\in Y $ witnessing an EFX violation in that $ X\prec Y\setminus\{g\} $ and transfer $ g $ from $ Y $ to $ L $.
    
    Moreover, $ \min_{\prec}\fancyQ \succeq \min_{\prec}\fancyP $, where strict inequality holds unless $ \fancyP $ was already EFX with respect to $ \prec $.
\end{proposition}

\begin{restatable}[Non-degenerate perturbation, see \cite{Chaudhury-Garg-Mehlhorn:EFX, AACGMM25}]{proposition}{nondegperturb}
    \label{prop:nondegenerate_perturbation}
	For every monotone valuation profile $ (v_i)_{i \in N} $ there exists a nondegenerate monotone valuation profile $ (\widehat v_i)_{i\in N} $ on the same agents and goods such that \[
		v_i(A) < v_i(B) \quad\Longrightarrow\quad \widehat v_i(A) < \widehat v_i(B)
	\] holds for every agent $ i $ and bundles $ A, B $. It follows that every strong envy for $ v_i $ holds also for $ \widehat v_i $. In particular, it follows that $ \widehat v_i $ is a counterexample to EFX existence if $ v_i $ is. In other words, every EFX allocation for $ \widehat v_i $ remains EFX under $ v_i $. 
\end{restatable}

\subsection{Conflicts and critical conflicts}
	
We introduce the notion of a \textit{conflict}. Intuitively, the situation we want to capture is as follows. For a fixed bundle $ B $ of goods, each of the two agents $ p $ and $ q $ can divide the goods outside $ B $ into two parts that are EFX according to their own preferences, and each would prefer either part of their own division to $ B $. However, they evaluate each other’s divisions almost oppositely: each would rather receive $ B $ than receive either part of the other agent’s division.

\begin{definition}[Conflict]
    \label{definition:conflict}
    Fix a nondegenerate monotone valuation profile and let $ p $ and $ q $ be distinct agents. We say that $ p $ and $ q $ are \textit{in a conflict} around a bundle $ B \subseteq M $ if there exist binary partitions $ \fancyP = \{P_1, P_2\} $ and $ \fancyQ = \{Q_1, Q_2\} $ of $ S := M \setminus B $ such that: \begin{enumerate}[label=(\alph*)]
        \item\label{definition:conflict:a} $ \fancyP $ is EFX with respect to $ \prec_p $ and $ \fancyQ $ is EFX with respect to $ \prec_q $;
        \item\label{definition:conflict:b} There exist goods $ b_p, b_q \in B $ such that
        \begin{align*}
            \max_p\fancyQ &\prec_p B \setminus \{b_p\}, &B &\prec_p \min_p\fancyP \\
            \max_q\fancyP &\prec_q B \setminus \{b_q\}, &B &\prec_q \min_q\fancyQ
        \end{align*}
    \end{enumerate} A \textit{conflict} is a tuple $ (p, q, \fancyP, \fancyQ, B) $ satisfying the above conditions. We call $ B $ the \textit{pivot bundle}, $ S $ the \textit{conflict region}, $ \fancyP $ the \textit{$ p $-fair view} of the conflict region, and $ \fancyQ $ the \textit{$ q $-fair view}. We refer to $ \fancyP $ and $ \fancyQ $ collectively as \textit{fair views}. If the pivot bundle and fair views are clear from the context, we may omit them in the tuple and simply write $ (p, q) $ to denote the conflict.
\end{definition}

We now show that conflicts generalize MMS-feasibility and leveled valuations\footnote{Recall that a valuation $ v $ is \emph{leveled} if $ |A| < |B| $ implies $ v(A) < v(B) $ for all bundles $ A, B $.}, alongside other basic properties:

\begin{lemma}[Properties of conflicts]
    \label{lemma:conflict_properties}
    Let agents $ p $ and $ q $ be in a conflict around a bundle $ B $ in some fixed nondegenerate monotone valuation profile. Denote conflict region by $ S $ and let the $ p $-fair view and $ q $-fair view be $ \{P_1, P_2\} $ and $ \{Q_1, Q_2\} $, respectively. Then: \begin{enumerate}[label=(\alph*)]
        \item\label{lemma:conflict_properties:a} Neither $ p $ nor $ q $ has an MMS-feasible valuation;
        \item\label{lemma:conflict_properties:b} Neither $ p $ nor $ q $ has a leveled valuation;
        \item\label{lemma:conflict_properties:c} $ |B| \ge 2 $;
        \item\label{lemma:conflict_properties:d} For every $ i, j \in \{1, 2\} $, \[
        P_i \cap Q_j \neq \varnothing
        \] In particular, it follows that $ |P_i|, |Q_i| \ge 2 $ for all $ i \in \{1, 2\} $;
    \end{enumerate}
\end{lemma}
\begin{proof}
    We start with Item~\ref{lemma:conflict_properties:a}. Definition~\ref{definition:conflict}~\ref{definition:conflict:b} implies $ \max_p\fancyQ \prec_p \min_p\fancyP $ and $ \max_q\fancyP \prec_q \min_q\fancyQ $, which immediately gives MMS-feasibility violations for both $ p $ and $ q $.

    We prove Item~\ref{lemma:conflict_properties:b}. Without loss of generality, suppose that the valuation of $ p $ was leveled. Then from $ Q_1, Q_2 \prec_p B \setminus \{b_p\} $ it follows that $ |Q_1|, |Q_2| \le |B| - 1 $ and hence $ |S| \le 2(|B| - 1) $. Similarly, $ B \prec_p P_1, P_2 $ yields $ |B| \le |P_1|, |P_2| $ and hence $ |S| \ge 2|B| $. This is a contradiction.
    
    We show Item~\ref{lemma:conflict_properties:c}. Fix elements $ b_p, b_q \in B $ of Definition~\ref{definition:conflict}~\ref{definition:conflict:b}. Clearly, $ B \neq \varnothing $ because $ b_p \in B $. The case $ |B| = 1 $ is impossible because otherwise from $ B = \{b_p\} $ and Definition~\ref{definition:conflict}~\ref{definition:conflict:b} it would follow that $ Q_1, Q_2 \prec_p B \setminus \{b_p\} = \varnothing $, which contradicts monotonicity. Hence, $ |B| \ge 2 $.
    
    We prove Item~\ref{lemma:conflict_properties:d}. Towards a contradiction, suppose that $ P_i \cap Q_j = \varnothing $ for some $ i, j \in \{1, 2\} $. Since $ \fancyQ $ is a binary partition of $ S $ and $ P_i \subseteq S $, it follows that $ P_i \subseteq Q_{3-j} $. Therefore, by monotonicity, $ P_i \preceq_p Q_{3-j} $. Combining this with Definition~\ref{definition:conflict}~\ref{definition:conflict:b} gives
    \[
    B \prec_p \min_p\fancyP \preceq_p P_i \preceq_p Q_{3-j} \preceq_p \max_p\fancyQ \prec_p B \setminus \{b_p\} \prec_p B
    \] which is a contradiction.
\end{proof}

We now define what it means for a conflict to be \textit{critical} in the setting of three agents. To understand the intuition, suppose agents $ p $ and $ q $ are in a conflict and start with the following natural way of trying to construct an allocation. Keep the pivot bundle $ B $ fixed and suppose that the remaining agent $ c $ can divide the conflict region into two bundles $ A $ and $ D $ such that $ c $ accepts both $ A $ and $ B $ under the partition $ \{A, B, D\} $; such a partition need not exist in general, but the reduction in the proof of Theorem~\ref{theorem:critical_conflict_determine_efx} maintains one. If the three agents can be assigned distinct bundles that they each accept, then the resulting allocation is EFX. If one of $ p $ and $ q $ can receive $ D $ while the other can receive one of $ A $ and $ B $, and both agents accept their respective bundles, then the remaining bundle among $ A $ and $ B $ can be assigned to $ c $, yielding an EFX allocation. Thus, once both $ p $ and $ q $ accept $ D $, it would suffice that at least one of them also accepted $ A $ or $ B $.

A critical conflict captures the most rigid way in which this allocation attempt can fail. Both $ p $ and $ q $ regard $ D $ as their only acceptable bundle: even after removing a suitable good from $ D $, each of them still prefers what remains to both $ A $ and $ B $. Thus, although $ c $ is flexible between two safe bundles, this flexibility is useless because $ p $ and $ q $ compete for the same bundle $ D $, and at most one of them can receive it. Moreover, this is not merely an accidental failure of the particular partition $ \{A, B, D\} $. Agent $ c $ prefers $ A $ to every block occurring in either the $ p $-fair view or the $ q $-fair view. Consequently, $ A $ cannot lie entirely inside one block of either fair view: it must contain goods from both blocks of both divisions. Thus, the bundle that is protected for $ c $ is intertwined with both of the incompatible fair divisions proposed by $ p $ and $ q $. This ties the two obstructions together: the fair views of $ p $ and $ q $ cannot be reconciled, while the alternative division $ \{A, D\} $ leaves both of them competing for $ D $.

\begin{definition}[Critical conflict]
    \label{definition:critical_conflict}
    Let $ (p, q, \fancyP, \fancyQ, B) $ be a conflict between $ p $ and $ q $ in some fixed nondegenerate monotone valuation profile for three agents, where $ \fancyP $ and $ \fancyQ $ denote the $ p $-fair view and $ q $-fair view, respectively, and $ B $ is the pivot bundle. Denote the conflict region by $ S $ and let $ c $ be the remaining agent. We call this conflict \textit{critical} if there exists a binary partition $ \{A, D\} $ of $ S $ such that: \begin{enumerate}[label=(\roman*)]
        \item\label{definition:critical_conflict:i} $ A \succ_c X $ for every $ X \in \fancyP \cup \fancyQ $;
        \item\label{definition:critical_conflict:ii} For every $ i \in \{p, q\} $ there exists $ d_i \in D $ such that $ A, B \prec_i D \setminus \{d_i\} $;
        \item\label{definition:critical_conflict:iii} $ A, B \in \Acc_c(\{A, B, D\}) $;
    \end{enumerate} A \textit{critical conflict} is a tuple $ (p, q, \fancyP, \fancyQ, B, A, D) $ satisfying the above conditions. We call the ordered pair $ (A, D) $ the \textit{critical cut}. As in Definition~\ref{definition:conflict}, if the pivot bundle, fair views and critical cut are clear from the context, we may omit them in the tuple and simply write $ (p, q) $ to denote the critical conflict.
\end{definition}

Note that Definitions~\ref{definition:conflict} and \ref{definition:critical_conflict} are symmetric in $ p $ and $ q $. In other words, $ p $ and $ q $ are in a (resp., critical) conflict if and only if $ q $ and $ p $ are in a (resp., critical) conflict with the fair views swapped and everything else left unchanged.

We now establish some basic properties of critical conflicts. In particular, we make precise the above intuition about why $ A $ must be intertwined with both fair views:

\begin{lemma}[Properties of critical conflicts]
    \label{lemma:critical_conflict_properties}
    Let agents $ p $ and $ q $ be in a critical conflict around a bundle $ B $ in some fixed nondegenerate monotone valuation profile for three agents. Denote conflict region by $ S $, let the $ p $-fair view and $ q $-fair view be $ \{P_1, P_2\} $ and $ \{Q_1, Q_2\} $, respectively, and let $ (A, D) $ be the critical cut. Then:
    \begin{enumerate}[label=(\alph*)]
        \item\label{lemma:critical_conflict_properties:a} For each $ i \in \{p, q\} $, fix any $ d_i $ satisfying Definition~\ref{definition:critical_conflict}~\ref{definition:critical_conflict:ii}. Then, $ D \setminus \{d_p\} $ intersects both $ Q_1 $ and $ Q_2 $. Similarly, $ D \setminus \{d_q\} $ intersects both $ P_1 $ and $ P_2 $. In particular, it follows that $ |D| \ge 3 $;
        \item\label{lemma:critical_conflict_properties:b} $ A $ intersects $ P_1 $, $ P_2 $, $ Q_1 $ and $ Q_2 $. In particular, it follows that $ |A| \ge 2 $;
        \item\label{lemma:critical_conflict_properties:c} $ |S| \ge 5 $;
    \end{enumerate}
\end{lemma}
\begin{proof}
    We start with Item~\ref{lemma:critical_conflict_properties:a}. For every $ i \in \{p, q\} $, fix a good $ d_i \in D $ as in Definition~\ref{definition:critical_conflict}~\ref{definition:critical_conflict:ii}. Note that by Definitions~\ref{definition:conflict}~\ref{definition:conflict:b} and \ref{definition:critical_conflict}~\ref{definition:critical_conflict:ii}, \begin{equation}
        \label{eq:d_w_dp_succp_b_succp_q1q2}
        D \setminus \{d_p\} \succ_p B \succ_p Q_1, Q_2
    \end{equation} This means that $ D \setminus \{d_p\} $ cannot be a subset of $ Q_i $ for any $ i $, because otherwise we would have $ D \setminus \{d_p\} \preceq_p Q_i $ by monotonicity, contradicting \eqref{eq:d_w_dp_succp_b_succp_q1q2}. Since $ \fancyQ $ is a partition of $ S \supseteq D \setminus \{d_p\} $, it follows that $ D \setminus \{d_p\} $ must intersect both its blocks. The proof for $ D \setminus \{d_q\} $ and $ \fancyP $ is identical.
    
    Turning to Item~\ref{lemma:critical_conflict_properties:b}, suppose that $ A $ did not intersect $ P_i $ for some $ i \in \{1, 2\} $. Then $ A \subseteq P_{3-i} $ because there are only two blocks. Hence, by monotonicity read in conjunction with Definition~\ref{definition:critical_conflict}~\ref{definition:critical_conflict:i}, \[
    P_{3-i} \prec_c A \preceq_c P_{3-i}
    \] which is a contradiction. The proof for $ Q_1 $ and $ Q_2 $ is identical. This shows Item~\ref{lemma:critical_conflict_properties:b}.
    
    Since $ |S| = |A| + |D| $, Item~\ref{lemma:critical_conflict_properties:c} follows immediately from Items \ref{lemma:critical_conflict_properties:a} and \ref{lemma:critical_conflict_properties:b}.
\end{proof}

\subsection{Reduction}

\begin{figure}[t]
\centering
\begin{tikzpicture}[
    box/.style={
        draw,
        rounded corners=2pt,
        align=center,
        text width=3.9cm,
        minimum height=1.15cm,
        inner sep=5pt,
        font=\small
    },
    implication/.style={
        -{Latex[length=2mm]},
        semithick
    },
    reason/.style={
        font=\footnotesize,
        align=center,
        fill=white,
        inner sep=2pt
    }
    ]
    
    \node[box] (leveled) at (0,0) {One leveled agent};
    \node[box] (mms) at (-4.6,0) {One MMS-feasible\\agent};
    \node[box] (agreement) at (4.6,0) {Two agents agree on\\comparisons between disjoint bundles};
    
    \node[box] (ordinary) at (0,-3.0) {A pair of agents\\without a conflict};
    \node[box] (critical) at (0,-5.2) {A pair of agents\\without a critical conflict};
    \node[box] (efx) at (0,-7.4) {EFX exists};
    
    \draw[implication] (mms.south) -- node[reason] {Lemma~\ref{lemma:conflict_properties}~\ref{lemma:conflict_properties:a}} (ordinary.north west);
    \draw[implication] (leveled.south) -- node[reason] {Lemma~\ref{lemma:conflict_properties}~\ref{lemma:conflict_properties:b}} (ordinary.north);
    \draw[implication] (agreement.south) -- node[reason] {Definition~\ref{definition:conflict}~\ref{definition:conflict:b}} (ordinary.north east);
    
    \draw[implication] (ordinary.south) -- node[reason] {Trivial} (critical.north);
    \draw[implication] (critical.south) -- node[reason] {Theorem~\ref{theorem:critical_conflict_determine_efx}} (efx.north);
    
    \draw[implication,rounded corners=8pt] ($ (mms.south) + (-.5,0) $) -- node[reason,text width=4cm] {\cite[Theorem~10]{akrami2022efx}}(mms.south |- efx.west) -- (efx.west);
    
    \draw[implication,rounded corners=8pt] ($ (agreement.south) + (.5,0) $) -- node[reason,text width=4cm] {\cite{Plaut-Roughgarden}} (agreement.south |- efx.east) -- (efx.east);
\end{tikzpicture}
\caption{Implications for three agents with nondegenerate monotone valuations. In other words, arrows denote inclusion between the specified classes of valuation profiles.}
\label{fig:conflict_classes_vs_known_classes}
\end{figure}

It turns out that for any class of three-agent valuation profiles, the ability to find EFX allocations in situations of critical conflicts between all pairs of agents fully determines EFX existence for all profiles. In other words, proving EFX existence reduces to constructing EFX allocations for valuations where every pair of agents is in a critical conflict.
	
\begin{theorem}
    \label{theorem:critical_conflict_determine_efx}
    Fix a nondegenerate monotone valuation profile for three agents. Then at least one of the following statements is true: \begin{itemize}
        \item There exists an EFX allocation;
        \item Between every pair of agents there is a critical conflict;
    \end{itemize}
\end{theorem}

We prove Theorem~\ref{theorem:critical_conflict_determine_efx} in Appendix~\ref{section:critical_conflict_determine_efx}. In essence, the proof is an improved application and analysis of the known invariant-based technique for constructing EFX allocations in the presence of an MMS-feasible agent \cite[Theorem~10]{akrami2022efx}. In particular, Theorem~\ref{theorem:critical_conflict_determine_efx} subsumes this result for nondegenerate profiles, because if there is an MMS-feasible agent, then that agent cannot be in a conflict with any other agent by Lemma~\ref{lemma:conflict_properties}~\ref{lemma:conflict_properties:a}, meaning that by Theorem~\ref{theorem:critical_conflict_determine_efx}, an EFX allocation exists. In the same spirit, we obtain EFX existence in the presence of one leveled agent:

\begin{corollary}
    \label{corollary:efx_for_three_agents_one_leveled}
    Every monotone valuation profile for three agents where at least one agent has a leveled valuation admits an EFX allocation.
\end{corollary}
\begin{proof}
    Note that the transformation of Proposition~\ref{prop:nondegenerate_perturbation} preserves leveled valuations. Therefore, Proposition~\ref{prop:nondegenerate_perturbation} reduces the proof to the case of a nondegenerate valuation profile. Let $ p $ denote the agent with a leveled valuation. By Lemma~\ref{lemma:conflict_properties}~\ref{lemma:conflict_properties:b}, $ p $ is not in a conflict with any other agent. Hence, an EFX allocation exists by Theorem~\ref{theorem:critical_conflict_determine_efx}.
\end{proof}

It is known that an EFX allocation exists if all agents are leveled \cite[Theorem~4]{christodoulou2025fairtruthfulallocationsleveled}. Corollary~\ref{corollary:efx_for_three_agents_one_leveled} substantially strengthens this result in the setting of three agents by showing that only one leveled agent is in fact sufficient.

In Figure~\ref{fig:conflict_classes_vs_known_classes}, we illustrate the relationship between the notion of (critical) conflicts we introduced and important known valuation classes such as MMS-feasible and leveled ones. In particular, the figure demonstrates how Theorem~\ref{theorem:critical_conflict_determine_efx} generalizes the results of \cite{Plaut-Roughgarden, AACGMM25}. If two agents agree on all comparisons between disjoint bundles, an EFX allocation exists: a partition that is EFX for one of them by Proposition~\ref{prop:PR} is EFX for both, and the third agent chooses first. Such agents are clearly not in a conflict, so Theorem~\ref{theorem:critical_conflict_determine_efx} covers this case as well.

\section{EFX existence up to seven goods}
\label{section:efx_existence_up_to_seven}

In this section, we investigate the case of three agents, at most seven goods and outline our proof of EFX existence for any such monotone valuation profile, deferring the full proofs to Appendix~\ref{section:efx_existence_up_to_seven_appendix}. 

Interestingly, the mere presence of a critical conflict in the setting of at most seven goods imposes the following rigid structure on the pivot bundle, conflict region, the components of a critical cut and the fair view blocks.

\begin{restatable}{lemma}{threesevenbasicstructurelemma}
    \label{lemma:three_seven_basic_structure}
    Fix a nondegenerate monotone valuation profile for three agents and at most seven goods. Assume that agents $ p $ and $ q $ are in a critical conflict around bundle $ B $, with $ p $-fair view $ \fancyP $, $ q $-fair view $ \fancyQ $, conflict region $ S $ and critical cut $ (A, D) $. Then: \begin{enumerate}[label=(\alph*)]
        \item\label{lemma:three_seven_basic_structure:a} The number of goods must be exactly seven and $ |B| = 2 $, $ |S| = 5 $;
        \item\label{lemma:three_seven_basic_structure:b} Every fair view consists of a size-two block and a size-three block. Thus, we can write $ \fancyP = \{P^{(2)}, P^{(3)}\} $ and $ \fancyQ = \{Q^{(2)}, Q^{(3)}\} $, where the superscripts indicate cardinalities. Moreover, the intersections must be of the form \[
        \begin{array}{c|cc}
            \cap & Q^{(2)} & Q^{(3)} \\ \hline
            P^{(2)} & \{x\} & \{y\} \\
            P^{(3)} & \{z\} & \{u, v\}
        \end{array}
        \] for distinct goods $ x,y,z,u,v $. In other words, \[
        \begin{aligned}
            P^{(2)}&=\{x,y\}, &P^{(3)}&=\{z,u,v\},\\
            Q^{(2)}&=\{x,z\}, &Q^{(3)}&=\{y,u,v\}
        \end{aligned}
        \]
        \item\label{lemma:three_seven_basic_structure:c} It holds that $ |A| = 2 $ and $ |D| = 3 $. Moreover, in the notation of Item~\ref{lemma:three_seven_basic_structure:b}, there are only the following possibilities for $ A $ and $ D $:\begin{center}
            \begin{tabular}{c|c}
                $ A $ & $ D $ \\\hline
                $ \{x, u\} $ & $ \{y, z, v\} $ \\
                $ \{x, v\} $ & $ \{y, z, u\} $ \\
                $ \{y, z\} $ & $ \{x, u, v\} $
            \end{tabular}
        \end{center}
    \end{enumerate}
\end{restatable}

Lemma~\ref{lemma:three_seven_basic_structure} shows that, when there are at most seven goods, every critical conflict has the same rigid combinatorial shape. In particular, the sizes of the bundles are fixed, and the relationship between the fair view blocks is unique up to relabeling of goods.

We now sketch the proof of Theorem~\ref{thm:intro-seven}. By Proposition~\ref{prop:nondegenerate_perturbation}, it suffices to consider nondegenerate valuations. Let $ p $, $ q $ and $ c $ denote the three agents. Towards a contradiction, suppose that no EFX allocation exists. By Theorem~\ref{theorem:critical_conflict_determine_efx}, there is a critical conflict between every pair of agents. We denote the components of these three critical conflicts using the following notation: \begin{center}
    \begin{tabular}{c|c|c|c|c}
        Critical conflict & Other agent & Pivot bundle & Conflict region & Fair views\\\hline
        $ p, q $ & $ c $ & $ B_c $ & $ S_c $ & $ \begin{aligned}
            \fancyP_c &= \{P_c^{(2)}, P_c^{(3)}\}\\
            \fancyQ_c &= \{Q_c^{(2)}, Q_c^{(3)}\}
        \end{aligned} $ \\\hline
        $ p, c $ & $ q $ & $ B_q $ & $ S_q $ & $ \begin{aligned}
            \fancyP_q &= \{P_q^{(2)}, P_q^{(3)}\}\\
            \fancyC_q &= \{C_q^{(2)}, C_q^{(3)}\}
        \end{aligned} $ \\\hline
        $ q, c $ & $ p $ & $ B_p $ & $ S_p $ & $ \begin{aligned}
            \fancyQ_p &= \{Q_p^{(2)}, Q_p^{(3)}\}\\
            \fancyC_p &= \{C_p^{(2)}, C_p^{(3)}\}
        \end{aligned} $ \\\hline
    \end{tabular}
\end{center} By Lemma~\ref{lemma:three_seven_basic_structure}, there are exactly seven goods, each pivot bundle has size two, and each conflict region has size five. Each fair view has one block of size two and one of size three, as indicated by the superscripts. We let $ x, y, z, u, v $ denote the five goods of $ S_c $ such that \[
\begin{array}{c|cc}
    \cap & Q_c^{(2)} & Q_c^{(3)} \\ \hline
    P_c^{(2)} & \{x\} & \{y\} \\
    P_c^{(3)} & \{z\} & \{u, v\}
\end{array}
\]

First, the assumed absence of EFX allocations imposes rigidity on the agents' preferences as follows. If $ c $ accepted a block of $ \fancyP_c $ under $ \{B_c, P_c^{(2)}, P_c^{(3)}\} $, we could give that block to $ c $, the other block to $ p $ and the pivot bundle $ B_c $ to $ q $, obtaining an EFX allocation. The same argument applies to $ \fancyQ_c $. Hence, $ c $ must reject every fair-view block. Furthermore, the rejection of the preferred block of either view must be witnessed by a singleton remaining after deleting one good from $ B_c $. Hence, $ c $ prefers its favorite pivot element to all four fair-view blocks.

Next, if the agents preferred different pivot goods, we could separate the two goods into two pairs, give each pair holder their preferred pivot good, and give a suitable fair-view triple to its owner. The conflict inequalities and the EFX property of the triple's fair view make this allocation EFX. Thus, all three agents must agree on the more valuable good $ b $ of the pivot bundle. Moreover, each conflict participant ranks their own fair-view blocks above the pivot and the other participant's blocks below $ b $; the remaining agent places all four blocks below $ b $. For each agent, these blocks ranked below $ \{b\} $ cover the conflict region. Hence, $ b $ is the most valuable good from the viewpoint of every agent. It is therefore independent of the chosen critical conflict and we can write \[
    B_c = \{b, t_c\}, \qquad B_q = \{b, t_q\}, \qquad B_p = \{b, t_p\}
\] We may imagine these pivots as a star centered at $ b $, initially allowing coincident spokes.

Separating the pivot goods also shows that certain pairs containing $ t_c $ must be ranked below $ \{b\} $. If such a pair exceeded $ \{b\} $, we could give it to the corresponding agent, a suitable fair-view triple to its owner, and the remaining pair containing $ b $ to the third agent. The triple holder prefers their triple to either singleton from the complementary fair-view pair by its EFX property, and to either pivot singleton because the triple is preferred to the pivot. Each pair holder ranks the triple below $ \{b\} $ and receives a pair above $ \{b\} $. Since $ b $ is the most valuable good, neither pair strongly envies the other. Since we assumed that no EFX allocation exists, this argument gives, for the conflict between $ p $ and $ q $, the following comparisons: \[
    \begin{aligned}
        \{b\} &\succ_c \{x, t_c\}, \{y, t_c\}, \{z, t_c\} \\
        \{b\} &\succ_p \{x, t_c\}, \{z, t_c\} \\
        \{b\} &\succ_q \{x, t_c\}, \{y, t_c\}
    \end{aligned}
\] We apply the same argument also to other conflicts.

We next compare two spokes sharing an agent, say the $ (p,q) $ and $ (p,c) $ conflicts. The good $ b $, viewed as a singleton, supplies a common benchmark: a bundle above it for an agent cannot be contained in a bundle below it for that same agent. In particular, $ c $'s new pair $ C_q^{(2)} $ lies above $ \{b\} $, so it cannot lie inside any original fair-view block, all of which lie below $ \{b\} $ for $ c $. The additional pair comparisons leave few possible locations. Once this pair is located, its required one-good intersection with $ P_q^{(2)} $, together with $ p $'s comparisons, sharply restricts the new $ p $-fair view as well.

This explains why the second pivot cannot take its other good arbitrarily. If the pivots coincide, the required intersections and critical cuts on the same region become incompatible. If the second pivot instead takes a good from $ P_c^{(3)} $, the remaining arrangements either contradict the comparisons with $ \{b\} $, impose opposite comparisons through the EFX properties of the fair views, or leave every deletion remainder of a critical triple below the pivot when one is required to exceed it. These cases rule out $ t_q \in \{t_c, z, u, v\} $, leaving only $ t_q \in P_c^{(2)} = \{x, y\} $. It remains to exclude $ t_q = y $.

The choice $ t_q = y $ is also impossible. Applied to the second conflict, the pair comparisons imply that $ \{x, y\} \prec_p \{b\} $ whenever $ x \in C_q^{(2)} $. Since $ \{x, y\} = P_c^{(2)} \succ_p \{b\} $, we must therefore have $ x \notin C_q^{(2)} $. Now $ S_q = \{t_c, x, z, u, v\} $, and $ c $ ranks $ C_q^{(2)} $ above $ \{b\} $, but both $ \{z, u, v\} $ and $ \{t_c, z\} $ below it. Consequently, $ C_q^{(2)} = \{t_c, r\} $ for some $ r \in \{u, v\} $. Its complementary triple $ C_q^{(3)} $ then contains $ Q_c^{(2)} = \{x, z\} $. However, $ q $ ranks its own fair-view pair $ Q_c^{(2)} $ above $ \{b\} $, while, as the remaining agent of the second conflict, it ranks $ C_q^{(3)} $ below $ \{b\} $. This contradicts monotonicity.

Thus, the compatibility argument forces the second pivot to contain the unique good shared by the two original fair-view pairs. More precisely, Lemma~\ref{lemma:three_seven_pivots} gives $ t_q \in P_c^{(2)} \cap Q_c^{(2)} = \{x\} $. Applying the same lemma with the roles of $ p $ and $ q $ exchanged yields $ t_p = x $. Hence, the two other pivot bundles must coincide: \[
    B_q = B_p = \{b, x\}
\]

But the same lemma also rules out this coincidence. Applied to the conflicts $ (c, p) $ and $ (c, q) $, it requires \[
t_p \in C_q^{(2)} \cap P_q^{(2)}
\subseteq S_q = M \setminus B_q
\] This contradicts $ t_p = x \in B_q $. Consequently, the three critical conflicts cannot satisfy the additional requirements forced by the absence of EFX, and an EFX allocation must exist. For the full proof, see Appendix~\ref{section:efx_existence_up_to_seven_appendix}.

\section{Conclusion}\label{sec:conclusion}

We resolved the EFX existence problem for monotone valuations in the
negative. Our counterexample has three agents and eight goods. It can
be realized with weighted-coverage valuations, so EFX allocations need
not exist even under submodular valuations
(Theorems~\ref{thm:intro-counterexample}
and~\ref{thm:general-nonexistence}). The construction extends to $n$
agents and $m$ goods whenever $n\ge3$ and $m\ge n+5$.

\par\medskip\noindent\textbf{The remaining exact case.}\quad
EFX allocations exist under arbitrary monotone valuations for two
agents~\cite{Plaut-Roughgarden} and whenever
$m\le n+3$~\cite{Mahara-EFX}. Together with our counterexamples, this
leaves only the case of $n$ agents and $m=n+4$ goods open. For $n=3$,
this case is settled by Theorem~\ref{thm:intro-seven}: every instance
with three agents and seven goods admits an EFX allocation. The question
is therefore open exactly for $n\ge4$ and $m=n+4$. By
Lemma~\ref{lem:padding}, a single counterexample with $n_0$ agents and
$n_0+4$ goods would already give counterexamples with $n$ agents and
$n+4$ goods for every $n\ge n_0$. Conversely, a proof of existence for
$m=n+4$ would mark exactly where EFX existence ends. %

\par\medskip\noindent\textbf{Approximate EFX.}\quad
For general monotone valuations, no positive approximation factor can
be guaranteed. For subadditive valuations, $\frac12$-EFX allocations
always exist~\cite{Plaut-Roughgarden,BS2026}. We showed that
$\alpha$-EFX allocations need not exist for any $\alpha>2^{-1/6}\approx
0.891$ (Theorem~\ref{thm:subadditive-barrier}). For submodular, and in
fact weighted-coverage, valuations they need not exist for any
$\alpha>\lambda_*\approx0.957$
(Theorem~\ref{thm:submodular-bound-sec}). Both barriers hold for every
$n\ge3$ and $m\ge n+5$. The optimal universal approximation factor
therefore lies in $[\frac12,2^{-1/6}]$ for subadditive valuations and in
$[\frac12,\lambda_*]$ for submodular valuations.

For three agents, we raised the lower bound from $\frac12$ to
$2^{-1/3}\approx0.794$. We showed this in a stronger form: if one agent
has a subadditive valuation and the other two have arbitrary monotone
valuations, some allocation is EFX for the two monotone agents and
$2^{-1/3}$-EFX for the subadditive agent
(Theorem~\ref{thm:mixed-approx}). For three agents, the optimal factor
thus lies in $[2^{-1/3},2^{-1/6}]$ for subadditive valuations and in
$[2^{-1/3},\lambda_*]$ for submodular valuations. Closing these gaps,
and improving the guarantee of $\frac12$ for more than three agents,
remain open problems.

\par\medskip\noindent\textbf{Additive valuations.}\quad
Our counterexamples depend on non-additive valuations. Whether EFX
allocations always exist for additive valuations beyond three agents
remains open.

\section*{Acknowledgements} 
For the prior SAT solving work, we thank Henrik B\"{o}ving of the LEAN FRO for discussions on Lean library internals and integration with CadiCal. Additionally, we acknowledge that Shreyas Srinivas is a doctoral student at the Saarbrücken Graduate School for Computer Science.  Simon Mackenzie is supported by the NSF-CSIRO grant on ``Fair Sequential
Collective Decision-Making'' (RG230833). Mashbat Suzuki is supported by the ARC Laureate
Project FL200100204 on ``Trustworthy AI''.

\section*{Declaration of generative AI use}

For the original SAT-based approach (Appendix~\ref{sec:old-sat-approach}), no AI was used to generate the SAT formula or to verify the unsatisfiability proof and counterexample. The Lean verification of the SAT formula was largely handwritten, with GPT 5.2 used to prove some lemmas stated by the authors. Subsequent updates and extensions to the Lean formalization used incremental assistance from GPT 6.1 Sol. The authors carefully checked all Lean theorem statements and proofs and edited them where needed.

The EFX counterexample presented in Section~\ref{sec:counterexample} was obtained through extensive interaction with GPT 5.5 and Opus 4.8, starting from the original counterexample in Appendix~\ref{sec:old-sat-approach}. The authors guided the search for simplifications, for example by suggesting that the models look for relabeling symmetries. The models were used to generate and iteratively refined Z3 formulas, ultimately yielding the simpler counterexample.

GPT 5.6 and Astra also assisted in developing the submodular inapproximability construction in Appendix~\ref{app:submodular-bound} and establishing Theorem~\ref{thm:mixed-approx}. This assistance included mathematical exploration, proof refinement, and exposition. 

The material in Sections~\ref{section:reduction_to_conflicts} and~\ref{section:efx_existence_up_to_seven} and Appendices~\ref{section:critical_conflict_determine_efx}, \ref{section:efx_existence_up_to_seven_appendix}, and~\ref{section:three_seven_pivots} is mostly human-generated. The authors developed the mathematical ideas and directed the exploration underlying these results. AI assistance for these parts was limited to implementing and running SAT/SMT experiments specified by the authors. These experiments helped develop intuition about the combinatorial structure of the valuation classes considered.

The authors take full responsibility for the content of this paper, including all mathematical claims and proofs.

\clearpage
\bibliographystyle{alpha}
\bibliography{references}
\clearpage
\appendix

\section{Proof of the submodular approximation bound}
\label{app:submodular-bound}

We prove Theorem~\ref{thm:submodular-bound-sec} by following the argument
for the qualitative counterexample. The useful first step is to rule
out an agent receiving $x$ together with both types other than its own.
We can then force the bundle containing $y$ to be a pair, and force the
remaining goods into the same allocation as before. In each contradiction,
an agent values its own bundle at most $\lambda_*$ times a bundle
remaining elsewhere after one good is removed.

\subsection{The valuations}

Write $\lambda=\lambda_*$, so
\[
 4\lambda^4+4\lambda^3-3\lambda-4=0,
 \qquad \frac{19}{20}<\lambda<1.
\]
The bounds follow by evaluating the defining polynomial at $19/20$ and $1$.
The set $M$ consists of two distinct goods of each type $A,B,C$, together with $x,y$.
Agents $a,b,c$ are associated with $A,B,C$, respectively. As in the
main construction, repeated letters in a bundle denote distinct goods.

For a set $R$ of types, let $\chi_R(S)$ be $1$ if $S$ contains a good
whose type belongs to $R$, and $0$ otherwise. We write, for example,
$\chi_{AB}$ for $\chi_{\{A,B\}}$. Define
\begin{align}
 \theta_1&=\frac{\lambda^2}{2(1+\lambda)},&
 \theta_2&=\lambda^{-2}-1,&
 \theta_3&=\frac{1-\lambda}{2},\nonumber\\
 \theta_4&=\lambda^{-1}-1,&
 \theta_5&=\lambda^{-2}-\lambda,&
 \theta_6&=\lambda+\frac{\lambda}{2(1+\lambda)}-\lambda^{-2}.
 \label{eq:sm-weights}
\end{align}
Agent $a$ has valuation
\begin{align}
 w_a(S)={}&\theta_1\bigl(\chi_{AB}(S)+\chi_{AC}(S)\bigr)
              +\theta_2\chi_x(S)\nonumber\\
          &+\theta_3\bigl(\chi_{ABx}(S)+\chi_{ACx}(S)\bigr)
              +\theta_4\chi_y(S)\nonumber\\
          &+\theta_5\bigl(\chi_{By}(S)+\chi_{Cy}(S)\bigr)
              +\theta_6\bigl(\chi_{Bxy}(S)+\chi_{Cxy}(S)\bigr).
 \label{eq:sm-coverage}
\end{align}
The valuation of $b$ is given by the same formula with $A,B$ exchanged,
and that of $c$ by the same formula with $A,C$ exchanged. Each agent
therefore uses the same rule for its own type and treats the other two
types interchangeably. We may relabel agents and their associated
types together, keeping $x,y$ fixed. Whenever we relabel agents below,
we also relabel their associated types.

All six weights are positive. This is immediate for the first five;
for the last,
\[
 \theta_6>\frac{19}{16}-\frac{400}{361}>0.
\]
Thus the valuations are normalized weighted-coverage valuations, and
in particular are monotone and submodular. Repeating a type adds no
value: for example, $w_a(AAB)=w_a(AB)$.

Set
\begin{equation}\label{eq:sm-abbreviations}
 t=\frac{2+\lambda}{2(1+\lambda)}=2\lambda^3-1,
\end{equation}
where the second equality follows from the equation for $\lambda$.
Substitution in \eqref{eq:sm-coverage} gives the values used in the proof:
\begin{align}
 w_a(Bx)&=\lambda t,& w_a(AB)&=t,&
 w_a(By)&=t/\lambda,& w_a(Ax)&=t/\lambda^2,\nonumber\\
 w_a(ABx)&=\lambda,& w_a(BC)=w_a(ABC)&=1,&
 w_a(Ay)&=\lambda^{-1},& w_a(BCx)&=\lambda^{-2}.
 \label{eq:sm-balanced}
\end{align}
The same equalities hold with $B,C$ exchanged. We also have
\begin{equation}\label{eq:sm-strict}
 \max_{g\in M}w_a(\{g\})<\lambda t,
 \qquad w_a(xy)<\lambda t,
 \qquad t<\lambda^3.
\end{equation}
For completeness, the last inequality follows from
$t=2\lambda^3-1$ and $\lambda<1$. The first two follow from
$w_a(A)=1/(1+\lambda)$, $w_a(B)=w_a(C)=1/2$, and
\begin{align*}
 \lambda t-w_a(A)
 &=\frac{\lambda^2+2\lambda-2}{2(1+\lambda)}>0,\\
 \lambda t-w_a(xy)
 &=\frac{(1-\lambda)(\lambda^3+\lambda^2+\lambda+2)}
         {2\lambda^2(1+\lambda)}>0.
\end{align*}
Here we used the equation for $\lambda$ in the second identity;
the singleton values of $x,y$ are at most $w_a(xy)$.

Three simple consequences will be useful. A bundle without $x$ that
contains neither $Ay$ nor $BC$ is worth at most $t/\lambda$ to $a$:
its types are contained in one of $AB,AC,By,Cy$. A bundle without $y$
that contains neither $Ax$ nor $BC$ is worth at most $t$, since its
types are contained in one of $AB,AC,Bx,Cx$. Finally, a bundle without
$y$ that does not contain $BCx$ is worth at most $1$, since its types
are contained in one of $ABC,ABx,ACx$. The corresponding statements
hold for the other agents by relabelling.

\subsection{The six-good observation}

We use the following version of the observation from the main proof.
\begin{lemma}\label{lem:sm-six}
Let $P,Q,R$ be distinct types among $A,B,C$. Divide $P,P,Q,Q,R,x$
between two agents, giving each at least two goods. If the division is
$\alpha$-EFX between them for some $\alpha>\lambda$, then some bundle
of at least three goods contains $Px$ or $QR$.
\end{lemma}
\begin{proof}
Suppose neither pair occurs in a bundle of at least three goods.
A pair $Px$ cannot form an entire two-good bundle either, since the
other bundle would contain $QQR$; similarly, a two-good bundle $QR$
would leave $PPx$ together. Thus $x$ is separated from both $P$-goods,
and $R$ from both $Q$-goods.

The bundle containing $x$ therefore consists of $x$ and goods of just
one type, say $U$, and has the same value as $Ux$. The other bundle
contains both $P$-goods; removing one leaves a bundle containing a pair
of the two types other than $U$. Every agent values $Ux$ at most $\lambda$ times this
pair. The agent associated with $U$ compares $t/\lambda^2$ with $1$,
and $t/\lambda^2<\lambda$; either other agent compares $\lambda t$
with $t$. The holder of $x$ therefore violates $\alpha$-EFX.
\end{proof}

\subsection{Forcing the allocation}

Fix $\alpha>\lambda$ and suppose an $\alpha$-EFX allocation
$X$ exists. We will repeatedly contradict this by finding a nonempty
bundle remaining after a deletion for which
\begin{equation}\label{eq:sm-witness}
 w_i(X_i)\leq\lambda w_i(X_j\setminus\{g\}).
\end{equation}
All nonempty bundles have positive value, so such a comparison also
rules out every factor strictly greater than $\lambda$.

\par\medskip\noindent\textbf{Every agent receives at least two goods.}\quad
An empty bundle immediately gives a contradiction. If an agent receives
one good, another bundle has at least four goods. Every four-good
bundle contains either two different types from $A,B,C$, or both copies
of one such type together with $x,y$. In either case, it contains a
pair worth at least $t$ to every agent that survives a deletion.
This contradicts the singleton bound in
\eqref{eq:sm-strict}.

\par\medskip\noindent\textbf{No agent receives $x$ with both other types.}\quad
Suppose, after relabelling, that $a$ receives a bundle containing
$BCx$. Removing a $C$ leaves a bundle containing $Bx$, worth
$t/\lambda^2$ to $b$.
Agent $b$ must therefore have value greater than $t/\lambda$.
Since it has no $x$, its bundle must contain $By$ or $AC$.
Similarly, $c$ must have $Cy$ or $AB$.

Giving $By$ to $b$ excludes both alternatives for $c$: it takes the
only $y$ and the remaining $B$-good. Hence $b$ must receive $AC$,
and similarly $c$ must receive $AB$. These pairs and $a$'s triple use
every good except $y$:
\[
 \underbrace{BCx}_{\text{agent }a}
 \quad\big|\quad
 \underbrace{AC}_{\text{agent }b}
 \quad\big|\quad
 \underbrace{AB}_{\text{agent }c}.
\]
There is nowhere to put $y$. Adding it to $a$'s or $c$'s bundle makes
$By$ survive a deletion, while $b$ keeps $AC$. Adding it to $b$'s bundle
makes $Cy$ survive a deletion, while $c$ keeps $AB$. In each case an
agent compares its own value $1$ with $\lambda^{-1}$, giving
\eqref{eq:sm-witness}.

\par\medskip\noindent\textbf{The good $y$ forces a pair.}\quad
Relabel so that $a$ receives $y$, and write $Y$ for its bundle.
If $|Y|\geq3$ and $Y$ contains $B$, then $By$ survives a deletion
and is worth $\lambda^{-1}$ to $b$. But $b$ has no $y$, and the
preceding argument excludes $ACx$ from its bundle, so its value is at
most $1$. This gives \eqref{eq:sm-witness}. The same argument
excludes $C$ from $Y$.

Thus $Y$ uses only $A,A,x,y$ and takes at least two of $A,A,x$.
It also contains $Ay$ after a deletion, worth $t/\lambda$ to each
other agent. Their own values must therefore exceed $t$. Without $y$,
this requires $b$ to have $Bx$ or $AC$, and $c$ to have $Cx$ or $AB$.
Each needs a good from $A,A,x$, but at most one remains. Hence $|Y|=2$.

If $Y=xy$, the other six goods force a mixed-type pair into a bundle
of at least three goods. That pair is worth at least $t$ to $a$ and
survives a deletion, contradicting $w_a(xy)<\lambda t$.
If $Y=By$, the other agents divide $A,A,C,C,B,x$.
Lemma~\ref{lem:sm-six} puts $Ax$ or $BC$ in a bundle of at least
three goods. Either pair gives \eqref{eq:sm-witness}, since
\[
 w_a(By)=\lambda w_a(Ax),
 \qquad w_a(By)=t/\lambda<\lambda w_a(BC).
\]
Exchanging $b,c$ excludes $Cy$. Thus $a$ receives exactly $Ay$.

\par\medskip\noindent\textbf{The remaining goods are forced.}\quad
The other agents divide $B,B,C,C,A,x$. If $b$ receives neither $Bx$
nor $AC$, its value is at most $t$. Lemma~\ref{lem:sm-six}, with
$(P,Q,R)=(B,C,A)$, nevertheless puts one of these pairs into a
bundle of at least three goods. It must be $c$'s bundle, and either
pair is worth more than $t/\lambda$ to $b$, a contradiction.
Thus $b$ must have $Bx$ or $AC$; similarly, $c$ must have $Cx$ or $AB$.

Both alternatives involving $x$ require the unique $x$, and both
alternatives without $x$ require the remaining $A$-good. Exchanging
$b,c$ if necessary, $b$ has $Bx$ and $c$ has $AB$.
Only the two $C$-goods remain. Neither can join $Ay$, whose size is
fixed. If one joins $c$'s bundle, removing the $B$-good leaves a bundle
containing $AC$, worth $1$ to $b$. Agent $b$ then has $Bx$ and at most one $C$-good,
so its own value is at most $w_b(BCx)=\lambda$. This again gives
\eqref{eq:sm-witness}.

Both $C$-goods must therefore go to $b$, forcing
\[
 \underbrace{Ay}_{\text{agent }a}
 \quad\big|\quad
 \underbrace{BCCx}_{\text{agent }b}
 \quad\big|\quad
 \underbrace{AB}_{\text{agent }c}.
\]
Removing a $C$-good from $b$'s bundle leaves $BCx$, and
\[
 w_a(Ay)=\lambda^{-1}=\lambda w_a(BCx).
\]
This is the final contradiction. Hence no  allocation is
$\alpha$-EFX for any $\alpha>\lambda$.

\subsection{Strictly positive marginals}

\begin{corollary}\label{cor:sm-positive-marginals}
For every fixed $\alpha>\lambda_*$, there are three normalized,
monotone submodular valuations on eight goods with strictly positive
marginals that admit no  $\alpha$-EFX allocation.
\end{corollary}
\begin{proof}
Replace each valuation by
\[
 w_i^\varepsilon(S)=w_i(S)+\varepsilon|S|.
\]
For every allocation, the preceding proof gives a strict violation of
$\alpha$-EFX. There are only finitely many allocations, so a sufficiently
small $\varepsilon>0$ preserves one such strict violation for each of
them. The added term is additive, so submodularity is preserved, and
every marginal is now at least $\varepsilon$. The choice of
$\varepsilon$ may depend on $\alpha$.
\end{proof}

\section{Proof of the approximation guarantee with one subadditive agent}\label{app:positive}

Throughout this section, $p$ has a subadditive valuation, $q$ and $c$
have arbitrary monotone valuations, and $\alpha_0=2^{-1/3}$. We show that
there is an allocation that is $\alpha_0$-EFX for $p$ and EFX for $q$
and $c$. We run the reduction of Section~\ref{section:reduction_to_conflicts},
a refined analysis of the algorithm of~\cite{AACGMM25}, while requiring only $\alpha$-EFX for $p$. If
it fails to produce an allocation, $p$ and $q$ compete for the same
bundle $D$, and we extract three inequalities from this conflict that
are incompatible with subadditivity of $v_p$.

\subsection{What the reduction gives}

Fix $0<\alpha\le1$, and set $\beta_p=\alpha$ and $\beta_q=\beta_c=1$.
For an agent $i$ and a partition $\fancyP$, we say that $i$
\emph{$\alpha$-accepts} $X\in\fancyP$ if
$v_i(X)\ge\beta_i\,v_i(Y\setminus\{g\})$ for all
$Y\in\fancyP\setminus\{X\}$ and all $g\in Y$. An allocation in which
every agent receives a bundle it $\alpha$-accepts is $\alpha$-EFX for
$p$ and EFX for $q$ and $c$. Every agent $\alpha$-accepts a favorite
bundle, and $\alpha$-acceptance is preserved when the agent's own bundle
grows and the other bundles shrink.

\begin{lemma}\label{lem:approx-reduction}
Fix a nondegenerate monotone valuation profile for three agents $p,q,c$
and at least four goods, and let $0<\alpha\le1$. If no 
allocation is $\alpha$-EFX for $p$ and EFX for $q$ and $c$, then there
are a partition $\{A,B,D\}$ of $M$, goods $g\in D$ and $b\in B$, and a
partition $\{Q_1,Q_2\}$ of $A\cup D$ such that
\begin{align}
  v_p(B)&<\alpha\,v_p(D\setminus\{g\}),\label{eq:red-compete}\\
  v_p(D\setminus\{g\})&<\alpha\,v_p(A\cup\{g\}),\label{eq:red-transfer}\\
  v_p(Q_j)&<\alpha\,v_p(B\setminus\{b\})\qquad(j=1,2).\label{eq:red-cheap}
\end{align}
\end{lemma}

\begin{proof}
Run the proof of Theorem~\ref{theorem:critical_conflict_determine_efx}
with $c$ as the fixed agent and with acceptance by $p$ replaced by
$\alpha$-acceptance. Every step concerning only $c$ is unchanged, and
every branch either produces a required allocation or strictly increases
the potential $\min_c\{A,B\}$. We may therefore assume that we have
reached a partition $\{A,B,D\}$ satisfying~\eqref{eq:invariant}, with
$A\prec_c B$, in which $p$ and $q$ both accept only $D$,
as in~\eqref{eq:p_and_q_compete_over_D}. For $i\in\{p,q\}$, let
$g_i\in D$ maximize $v_i(D\setminus\{g_i\})$.

\par\medskip\noindent\textbf{The conflict over $D$.}\quad
Agent $i$ rejects its favorite $X$ among $A$ and $B$. The witness cannot
be a deletion of the other bundle in $\{A,B\}$, which is worth at most
$v_i(X)$ to $i$. Hence
$\max\{v_i(A),v_i(B)\}<\beta_i\,v_i(D\setminus\{g_i\})$. For $i=p$ and
$g:=g_p$, this is~\eqref{eq:red-compete}.

\par\medskip\noindent\textbf{The failed transfer.}\quad
Move $g_i$ from $D$ to $A$. If $i$ $\alpha$-accepted
$D\setminus\{g_i\}$ in $\{A\cup\{g_i\},B,D\setminus\{g_i\}\}$, the
transfer step of the original proof would increase the potential. That
step uses only that $i$ accepts $D'\supseteq D\setminus\{g_i\}$ against
$A'\subseteq A\cup\{g_i\}$ and $B$. So $i$ rejects it. The witness is
not a deletion of $B$, because
$\beta_i\,v_i(B)<v_i(D\setminus\{g_i\})$ by the previous paragraph.
Hence it is a deletion of $A\cup\{g_i\}$, and by monotonicity
$v_i(D\setminus\{g_i\})<\beta_i\,v_i(A\cup\{g_i\})$. For $i=p$, this
is~\eqref{eq:red-transfer}. This is the hard case~\eqref{eq:hard_case},
now with factors.

\par\medskip\noindent\textbf{$q$'s fair view is cheap for $p$.}\quad
Apply Proposition~\ref{prop:PR} under $v_q$ to
$\{A\cup\{g_q\},D\setminus\{g_q\}\}$, obtaining $\{Q_1,Q_2\}$, which is
EFX with respect to $v_q$. By the two previous paragraphs with $i=q$,
both parts of the input have $q$-value greater than $v_q(B)$, and
therefore so do $Q_1$ and $Q_2$. Let $Q_2$ be $p$'s favorite of the two.
If $p$ $\alpha$-accepted $Q_2$ in $\{B,Q_1,Q_2\}$, the
allocation $c\leftarrow B$, $p\leftarrow Q_2$, $q\leftarrow Q_1$ would
leave neither $p$ nor $q$ with strong envy. As in the proof of
Theorem~\ref{theorem:critical_conflict_determine_efx}
(Claims~\ref{claim:agents_p_q_do_not_strongly_envy_anyone}
and~\ref{claim:r_accepts_fprime_under_pprime}, with $p$ and $q$
exchanged), this either is a required allocation or increases the
potential. Hence $p$ rejects $Q_2$. The witness is not a deletion of
$Q_1$, which is worth at most $v_p(Q_2)$ to $p$. So
$v_p(Q_1)\le v_p(Q_2)<\alpha\,v_p(B\setminus\{b\})$ for some $b\in B$,
which is~\eqref{eq:red-cheap}.
\end{proof}

In the terminology of Section~\ref{section:reduction_to_conflicts},
inequality~\eqref{eq:red-cheap} is the $p$-side of
Definition~\ref{definition:conflict}~\ref{definition:conflict:b} for the
$q$-fair view. Inequalities~\eqref{eq:red-compete}
and~\eqref{eq:red-transfer} each cost $p$ a factor of $1/\alpha$: the
first because $p$ and $q$ compete for $D$, and the second because moving
a good out of $D$ does not resolve that competition.

\subsection{Proof of Theorem~\ref{thm:mixed-approx}}

\mixedapprox*

\begin{proof}[Proof of Theorem~\ref{thm:mixed-approx}]
With at most three goods, give each good to a different agent. All
deletions of other agents' bundles are empty, so this is exact EFX.
Suppose henceforth that there are at least four goods, and first assume
nondegeneracy.

Set $\alpha=\alpha_0$, and suppose that no  allocation is
$\alpha$-EFX for $p$ and EFX for $q$ and $c$. Take $A,B,D,g,b,Q_1,Q_2$
as in Lemma~\ref{lem:approx-reduction}. Since
$A\cup\{g\}\subseteq A\cup D=Q_1\mathbin{\dot\cup}Q_2$, inequalities
\eqref{eq:red-compete}, \eqref{eq:red-transfer}, monotonicity,
subadditivity, and~\eqref{eq:red-cheap} give
\begin{equation}\label{eq:positive-cubic}
  \frac{v_p(B)}{\alpha^2}
  <\frac{v_p(D\setminus\{g\})}{\alpha}
  <v_p(A\cup\{g\})
  \le v_p(Q_1)+v_p(Q_2)
  <2\alpha\,v_p(B).
\end{equation}
Since $\alpha^3=\tfrac12$, we have $2\alpha\,v_p(B)=v_p(B)/\alpha^2$, so
\eqref{eq:positive-cubic} is a contradiction. (The same argument works
for every $\alpha\le\alpha_0$.) We did not divide by $v_p(B)$; if
$v_p(B)=0$, the displayed inequalities are already impossible.

\par\medskip\noindent\textbf{Removing ties and zero marginal values.}\quad
We use the perturbation argument of~\cite[Appendix~A,
Lemma~5]{akrami2022efx}, which extends the reduction to nondegenerate
instances of~\cite{Chaudhury-Garg-Mehlhorn:EFX} to general monotone valuations.
Enumerate the goods and let
$b(S)=\sum_{g\in S}2^{\mathrm{index}(g)}$, a positive additive valuation
with distinct values on distinct subsets. Let $\delta>0$ be the smallest
nonzero difference $|v_i(S)-v_i(T)|$ over all agents $i$ and all sets
$S,T\subseteq M$ (if there is none, let $\delta=1$), and let
$0<\varepsilon<\delta/b(M)$. Replace each $v_i$ by $v_i+\varepsilon b$.
Since $\varepsilon|b(S)-b(T)|\le\varepsilon b(M)<\delta$, every
originally strict comparison is preserved, while all original ties are
broken by $b$. The new valuations are normalized, monotone, and
nondegenerate, and $p$'s valuation remains subadditive.

Apply the proved case along a sequence of such perturbations with
$\varepsilon\to0$. There are finitely many allocations, so one
allocation occurs along an infinite subsequence. Pass each of its
required EFX inequalities to the limit. This gives factors
$(\alpha_0,1,1)$ for the original valuations, with no exception for
zero-valued goods or zero-marginal deletions.
\end{proof}

\section{Proof of Theorem~\ref{theorem:critical_conflict_determine_efx}}
\label{section:critical_conflict_determine_efx}

Before we start, we establish the following consequence of Proposition~\ref{prop:PR} that we will need:

\begin{proposition}
    \label{prop:every_agent_has_safe_partition}
    Let $ (\prec_1, \dots, \prec_n) $ be a nondegenerate monotone valuation profile for $ n $ agents and $ m \ge n $ goods. For every agent $ c $, there exists a $ c $-safe partition of the $ m $ goods into $ n $ bundles.
\end{proposition}
\begin{proof}
    Proposition~\ref{prop:PR} applied with the valuation $ \prec_c $ yields the desired partition.
\end{proof}

We prove Theorem~\ref{theorem:critical_conflict_determine_efx}. The case when there are at most three goods trivially implies EFX existence, so in the following we can assume $ |M| > 3 $.
	
Fix an arbitrary agent $ c $ and let $ p $ and $ q $ denote the remaining agents. From now on, we maintain a partition $ \fancyP = \{A, B, D\} $ satisfying the following invariant: \begin{equation}
    \label{eq:invariant}
    A,B\in\Acc_c(\fancyP) \qquad\text{and}\qquad D\in\Acc_p(\fancyP)\cup\Acc_q(\fancyP)
\end{equation}

To initialize the invariant \eqref{eq:invariant}, we set $ \fancyP $ to be an arbitrary but fixed $ c $-safe partition, which exists by Proposition~\ref{prop:every_agent_has_safe_partition}. Let $ D $ be agent $ p $'s favorite bundle in $ \fancyP $ and call the other two bundles $ A $ and $ B $. Since $ \fancyP $ is $ c $-safe, agent $ c $ accepts both $ A $ and $ B $. Moreover, since $ D $ is the favorite bundle of $ p $, agent $ p $ accepts $ D $.

Suppose that \eqref{eq:invariant} holds, and, without loss of generality, that $ A \prec_c B $. We can further assume that \begin{equation}
    \label{eq:p_and_q_compete_over_D}
    \Acc_p(\fancyP) = \Acc_q(\fancyP) = \{D\}
\end{equation} Indeed, suppose \eqref{eq:p_and_q_compete_over_D} does not hold. Since the acceptance sets are nonempty, it is possible to choose $ \{a, b\} = \{p, q\} $ such that $ a $ accepts $ D $ and $ b $ accepts $ A $ or $ B $. If $ b $ accepts $ A $, then assigning $ D $, $ A $ and $ B $ to agents $ a $, $ b $ and $ c $, respectively, yields an EFX allocation. Similarly, $ b $ accepting $ B $ leads to an EFX allocation via assigning $ D $, $ B $ and $ A $ to agents $ a $, $ b $ and $ c $, respectively. Hence, it remains to consider the case when \eqref{eq:p_and_q_compete_over_D} holds.

Let $ g_p \in D $ be the good that maximizes $ D \setminus \{g_p\} $ under the $ \prec_p $ order. Similarly, choose $ g_q \in D $ such that $ D \setminus \{g_q\} $ is maximized under $ \prec_q $.

\begin{claim}
    \label{claim:d_without_ga_succa_A_B}
    For every $ a \in \{p, q\} $, $ D \setminus \{g_a\} \succ_a A, B $.
\end{claim}
\begin{proof}
    Fix $ a \in \{p, q\} $ and let $ X \in \{A, B\} $. Write $ Y $ for the other bundle in $ \{A, B\} $, and suppose towards a contradiction that $ D\setminus\{g_a\}\prec_a X $. Since $ X \notin \Acc_a(\fancyP) $ by \eqref{eq:p_and_q_compete_over_D}, some one-good deletion of another bundle must be preferred to $ X $ from the perspective of agent $ a $. This witness cannot be of the form $ X \prec_a D \setminus \{g\} $ for some $ g \in D $, because by the choice of $ g_a $, we have \begin{equation}
        \label{claim:d_without_ga_succa_A_B:eq01}
        X \succ_a D \setminus \{g_a\} \succeq_a D \setminus \{g\}
    \end{equation} for all $ g \in D $. Hence, there is some $ y \in Y $ such that $ X \prec_a Y \setminus \{y\} $. In particular, $ Y \succ_a X $ and consequently $ Y \succ_a X \setminus \{x\} $ for all $ x \in X $. Combining this with monotonicity and \eqref{claim:d_without_ga_succa_A_B:eq01} yields \[
    Y \succ_a Y \setminus \{y\} \succ_a X \succ_a D \setminus \{g\}
    \] for every $ g \in D $. Therefore, $ Y \in \Acc_a(\fancyP) $, which is a contradiction.
\end{proof}

Now suppose that $ D \setminus \{g_a\} \succ_a A \cup \{g_a\} $ for some $ a \in \{p, q\} $. Fix a minimal (with respect to inclusion) set $ A' \subseteq A \cup \{g_a\} $ such that $ A' \succ_c A $. Indeed, monotonicity and nondegeneracy ensure that such $ A' $ exists. Consider the partition $ \fancyP' := \{A', B, D'\} $ where \[
D' := (D \setminus \{g_a\}) \cup (A \cup \{g_a\} \setminus A')
\] By Claim~\ref{claim:d_without_ga_succa_A_B} read in conjunction with monotonicity, we have \begin{align*}
    D' &\succeq_a D \setminus \{g_a\} \succ_a B \\
    D' &\succeq_a D \setminus \{g_a\} \succ_a A \cup \{g_a\} \succeq_a A'
\end{align*} which means that $ D' $ is agent $ a $'s favorite bundle in $ \fancyP' $ and hence $ a $ accepts $ D' $ under $ \fancyP' $.

\begin{claim}
    \label{claim:no_strong_envy_of_c_between_aprime_and_b}
    There is no strong envy of agent $ c $ between bundles $ A' $ and $ B $.
\end{claim}
\begin{proof}
    Since $ c $ accepts $ A $ under $ \fancyP $ by \eqref{eq:invariant}, we have that $ A \succ_c B \setminus \{g\} $ for all $ g \in B $. Combining this with $ A' \succ_c A $ rules out strong envy of agent $ c $ directed from $ A' $ to $ B $. Regarding the other direction, we have to show that $ B \succ_c A' \setminus \{g\} $ for every $ g $ in $ A' $. Indeed, $ A' \setminus \{g\} \preceq_c A $ holds by the inclusion-minimality of $ A' $ with respect to $ \prec_c $, and together with $ A \prec_c B $ we obtain the desired absence of strong envy.
\end{proof}

We now distinguish whether the two bundles $ A' $ and $ B $ are both accepted by $ c $ under $ \fancyP' $. First suppose that this was the case. We have already shown that $ a $ accepts $ D' $ under $ \fancyP' $. Since $ A \prec_c A' $ and $ A \prec_c B $, we have \[
\min_c\{A', B\} \succ_c A = \min_c\{A, B\}
\] Hence, we have obtained a partition $ \fancyP' $ satisfying the same invariant \eqref{eq:invariant} as before, with a strictly improved least valuable bundle among the two bundles accepted by $ c $. For this reason, we can safely proceed by repeating the above argumentation for $ \fancyP' $, without falling into an infinite loop of reasoning.

Now suppose that at least one of $ A' $ and $ B $ is not accepted by $ c $ under $ \fancyP' $. By Claim~\ref{claim:no_strong_envy_of_c_between_aprime_and_b}, this failure cannot be caused by strong envy between $ A' $ and $ B $. Therefore, for some $ g \in D' $ we have $ D' \setminus \{g\} \succ_c A' $ in the case when $ c $ does not accept $ A' $, and $ D' \setminus \{g\} \succ_c B $ if $ c $ does not accept $ B $. In both cases, it follows that \[
D' \succ_c \min_c\{A', B\}
\] Hence, \[
A \prec_c \min_c\{A', B\} = \min_c{\fancyP'}
\] Since at least one of $ A' $ and $ B $ is not accepted by $ c $, the partition $ \fancyP' $ is not EFX with respect to $ \prec_c $. Let $ \fancyQ $ be the partition provided by Proposition~\ref{prop:PR} applied to $ \fancyP' $ and $ \prec_c $. Note that $ \fancyQ $ is $ c $-safe and $ A \prec_c \min_c{\fancyP'} \prec_c \min_c\fancyQ $.

Let $ D'' $ be the favorite bundle of agent $ a $ in $ \fancyQ $, and denote the remaining two bundles $ A'' $ and $ B'' $. Clearly, this ensures that $ a $ accepts $ D'' $ under $ \fancyQ $. Since $ \fancyQ $ is $ c $-safe, $ c $ accepts both $ A'' $ and $ B'' $. Overall, we have \[
\min_c\{A'',B''\} \succeq_c \min_c\fancyQ \succeq_c \min_c{\fancyP'} \succ_c A
\] meaning that in this case we again get a partition $ \fancyP'' := \{A'', B'', D''\} $ satisfying \eqref{eq:invariant} that has a strictly improved least valuable bundle among the two bundles accepted by $ c $, and can safely repeat the above argumentation to obtain an EFX allocation or a critical conflict.

We have reduced the proof to the case in which \begin{equation}
    \label{eq:hard_case}
    A\cup\{g_a\} \succ_a D\setminus\{g_a\} \succ_a A,B
\end{equation} holds for all $ a \in \{p, q\} $. Apply Proposition~\ref{prop:PR} to the partition $ \{A \cup \{g_p\}, D \setminus \{g_p\}\} $ using the valuation $ \prec_p $, and let $ \{P_1, P_2\} $ be the resulting binary partition. By \eqref{eq:hard_case}, we have \begin{equation}
    \label{eq:p_partition_high_for_p}
    B \overset{\eqref{eq:hard_case}}{\prec_p} D \setminus \{g_p\} \overset{\eqref{eq:hard_case}}{=} \min_p\{D \setminus \{g_p\}, A \cup \{g_p\}\} \overset{\ref{prop:PR}}{\preceq_p} P_1, P_2
\end{equation} Moreover, $ \{P_1, P_2\} $ is EFX with respect to $ \prec_p $. Without loss of generality we assume $ P_1 \prec_q P_2 $.

Suppose $ P_2 \in \Acc_q(\{B, P_1, P_2\}) $ and consider the allocation \begin{equation}
    \label{eq:tentative_allocation}
    c \leftarrow B, \qquad p \leftarrow P_1, \qquad q \leftarrow P_2
\end{equation}

\begin{claim}
    \label{claim:agents_p_q_do_not_strongly_envy_anyone}
    Under the allocation \eqref{eq:tentative_allocation}, agents $ p $ and $ q $ do not strongly envy anyone.
\end{claim}
\begin{proof}
    Since $ \{P_1, P_2\} $ is EFX with respect to $ \prec_p $, $ p $ does not strongly envy $ q $ under \eqref{eq:tentative_allocation}. Envy of $ p $ towards $ c $ is ruled out by \eqref{eq:p_partition_high_for_p}. Thus, agent $ p $ does not strongly envy either of the other two agents. The assumption $ P_2 \in \Acc_q(\{B, P_1, P_2\}) $ rules out strong envy of $ q $ towards $ p $ or $ c $.
\end{proof}

We can assume that allocation \eqref{eq:tentative_allocation} is not EFX, because otherwise we would be done. By Claim~\ref{claim:agents_p_q_do_not_strongly_envy_anyone}, this means that agent $ c $ must strongly envy $ p $ or $ q $. Let $ E $ be a bundle strongly envied by $ c $, and let $ F $ denote the other bundle such that $ \{E, F\} = \{P_1, P_2\} $. In other words, there is a good $ g \in E $ such that $ B \prec_c E \setminus \{g\} $. Fix an inclusion-minimal subset $ E' \subseteq E $ such that $ E' \succ_c B $. This ensures that \begin{equation}
    \label{eq:eprime_minimality}
    B \succ_c E' \setminus \{h\}
\end{equation} holds for all $ h \in E' $. By moving the goods removed from $ E $ to $ F $, we obtain the partition $ \fancyP' := \{B, E', F'\} $ where \[
F' := F \cup (E \setminus E')
\] Let $ r \in \{p, q\} $ be the agent who received $ F $ in the allocation \eqref{eq:tentative_allocation}. In other words, $ r := p $ if $ F = P_1 $ and $ r := q $ if $ F = P_2 $. We have already shown (Claim~\ref{claim:agents_p_q_do_not_strongly_envy_anyone}) that $ r $ accepts $ F $ under the partition $ \{B, E, F\} $ and now show, in essence, that this remains true after our modification:

\begin{claim}
    \label{claim:r_accepts_fprime_under_pprime}
    Agent $ r $ accepts $ F' $ under $ \fancyP' $.
\end{claim}
\begin{proof}
    By Claim~\ref{claim:agents_p_q_do_not_strongly_envy_anyone} taken together with monotonicity, \[
    F' \succeq_r F \overset{\ref{claim:agents_p_q_do_not_strongly_envy_anyone}}{\succ_r} B \setminus \{h\}
    \] holds for all $ h \in B $. Thus, there is no strong envy of agent $ r $ directed from $ F' $ towards $ B $. It remains to rule out such strong envy directed towards $ E' $. Indeed, the acceptance of $ F $ by agent $ r $ under $ \{B, E, F\} $ by Claim~\ref{claim:agents_p_q_do_not_strongly_envy_anyone} combined with monotonicity yields \[
    F' \succeq_r F \succ_r E \setminus \{h\} \succeq_r E' \setminus \{h\}
    \] for all $ h \in E' $, as desired.
\end{proof}

We next observe that there is no strong envy of $ c $ between $ B $ and $ E' $. The envy direction from $ E' $ to $ B $ is excluded by the fact that $ E' \succ_c B $, and the converse direction is ruled out by \eqref{eq:eprime_minimality}. We make a case distinction.

\textbf{Case~1}: Suppose that both $ B $ and $ E' $ are accepted by $ c $ under $ \fancyP' $. Read together with Claim~\ref{claim:r_accepts_fprime_under_pprime}, this shows that $ \fancyP' $ satisfies invariant~\eqref{eq:invariant}. Moreover, \[
\min_c\{B, E'\} = B \succ_c A
\] meaning that we can safely repeat the above argumentation with the new partition $ \fancyP' $, because the proof is guaranteed to strictly increase the value of the least valuable bundle $ A $ among the two bundles accepted by agent $ c $, unless the existence of EFX has already been established.

\textbf{Case~2}: Suppose that at least one of $ B $ and $ E' $ is not accepted by $ c $ under $ \fancyP' $. Since there is no strong envy of $ c $ between $ B $ and $ E' $, this failure must be witnessed by strong envy of the form \[
X \prec_c F' \setminus \{h\}
\] for some $ h \in F' $ and $ X \in \{B, E'\} $. It follows that \[
F' \succ_c \min_c\{B, E'\} = B
\] where the equality holds since $ B \prec_c E' $. Consequently, \[
\min_c{\fancyP'} = \min_c\{B, E'\} = B.
\] The partition $ \fancyP' $ is not EFX with respect to $ \prec_c $, because at least one of $ B $ and $ E' $ is not accepted by $ c $. Apply Proposition~\ref{prop:PR} to $ \fancyP' $ using $ \prec_c $, and let $ \fancyQ $ be the resulting partition. Then $ \fancyQ $ is $ c $-safe and \[
\min_c\fancyQ \succ_c \min_c\fancyP' = B \succ_c A
\] Let $ D'' $ be the favorite bundle of agent $ p $ in $ \fancyQ $, and denote the remaining two bundles $ A'' $ and $ B'' $. This clearly ensures that $ p $ accepts $ D'' $ under $ \fancyQ $. Due to $ \min_c\{A'', B''\} \succeq_c \min_c\fancyQ \succ_c A $, $ \fancyQ $ satisfies invariant \eqref{eq:invariant}, and repeating the above proof cannot lead to an infinite cycle.

It follows that hereinafter we can assume that $ P_2 \notin \Acc_q(\{B, P_1, P_2\}) $. Since $ P_1 \prec_q P_2 $, $ q $ does not envy $ p $, let alone strongly envy, so there exists a good $ b_q \in B $ such that \begin{equation}
    \label{eq:p_1_prec_q_p_2_prec_q_b_without_g}
    P_1 \prec_q P_2 \prec_q B \setminus \{b_q\}
\end{equation} We repeat the above construction with the roles of $ p $ and $ q $ swapped. More precisely, let $ \{Q_1, Q_2\} $ be the partition provided by Proposition~\ref{prop:PR} applied to the valuation $ \prec_q $ and input partition $ \{A \cup \{g_q\}, D \setminus \{g_q\}\} $. Thus, $ \{Q_1, Q_2\} $ is EFX with respect to $ \prec_q $. From \eqref{eq:hard_case} it follows that \begin{equation}
    \label{eq:q_partition_high_for_q}
    B \overset{\eqref{eq:hard_case}}{\prec_q} D \setminus \{g_q\} \overset{\eqref{eq:hard_case}}{=} \min_q\{D \setminus \{g_q\}, A \cup \{g_q\}\} \overset{\ref{prop:PR}}{\preceq_q} Q_1, Q_2
\end{equation} Without loss of generality we assume $ Q_1 \prec_p Q_2 $. If $ Q_2 \in \Acc_p(\{B, Q_1, Q_2\}) $, then, by the above proof with the roles of $ p $ and $ q $ swapped, we obtain an EFX allocation or a new partition satisfying \eqref{eq:invariant} in a way that ensures the absence of infinite cycles (namely, the new partition is guaranteed to increase the value of the least valuable bundle among the two bundles accepted by $ c $). Hence, we can assume that $ Q_2 \notin \Acc_p(\{B, Q_1, Q_2\}) $. Since $ Q_2 \succ_p Q_1 $, this means that \begin{equation}
    \label{eq:q_1_prec_p_q_2_prec_p_b_without_h}
    Q_1 \prec_p Q_2 \prec_p B \setminus \{b_p\}
\end{equation} holds for some $ b_p \in B $. Overall, we have obtained two partitions $ \{P_1, P_2\} $ and $ \{Q_1, Q_2\} $ of $ A \cup D = M \setminus B $ such that: \begin{itemize}
    \item $ \{P_1, P_2\} $ is EFX with respect to $ \prec_p $
    \item $ \{Q_1, Q_2\} $ is EFX with respect to $ \prec_q $
    \item There exist goods $ b_p, b_q \in B $ such that \begin{align}
        \label{eq:q_order}
        P_1 \overset{\eqref{eq:p_1_prec_q_p_2_prec_q_b_without_g}}{\prec_q} P_2 \overset{\eqref{eq:p_1_prec_q_p_2_prec_q_b_without_g}}{\prec_q} B \setminus \{b_q\} \prec_q &B \overset{\eqref{eq:q_partition_high_for_q}}{\prec_q} Q_1, Q_2 \\
        \label{eq:p_order}
        Q_1 \overset{\eqref{eq:q_1_prec_p_q_2_prec_p_b_without_h}}{\prec_p} Q_2 \overset{\eqref{eq:q_1_prec_p_q_2_prec_p_b_without_h}}{\prec_p} B \setminus \{b_p\} \prec_p &B \overset{\eqref{eq:p_partition_high_for_p}}{\prec_p} P_1, P_2
    \end{align}
\end{itemize}

From $ \{P_1, P_2\} $ being EFX with respect to $ \prec_p $ and \eqref{eq:p_order} it follows that \begin{equation}
    \label{eq:p1_p2_accepted_by_p_under_b_p1_p2}
    P_1, P_2 \in \Acc_p(\{B, P_1, P_2\})
\end{equation} Similarly, $ \{Q_1, Q_2\} $ being EFX with respect to $ \prec_q $ together with \eqref{eq:q_order} imply that \begin{equation}
    \label{eq:q1_q2_accepted_by_q_under_b_q1_q2}
    Q_1, Q_2 \in \Acc_q(\{B, Q_1, Q_2\})
\end{equation} Moreover, from \eqref{eq:q_order} and \eqref{eq:p_order} it follows that \begin{equation}
    \label{eq:b_accepted_by_q_underbp1p2_by_p_underbq1q2}
    B \in \Acc_q(\{B, P_1, P_2\}) \cap \Acc_p(\{B, Q_1, Q_2\})
\end{equation} Observe that if $ c $ accepts $ P_i $ under $ \{B, P_1, P_2\} $ for some $ i \in \{1, 2\} $, then the allocation \[
c \leftarrow P_i, \qquad p \leftarrow P_{3-i}, \qquad q \leftarrow B
\] is EFX by \eqref{eq:p1_p2_accepted_by_p_under_b_p1_p2} and \eqref{eq:b_accepted_by_q_underbp1p2_by_p_underbq1q2}. Similarly, $ c $ accepting $ Q_i $ under $ \{B, Q_1, Q_2\} $ for some $ i \in \{1, 2\} $ leads to an EFX allocation \[
c \leftarrow Q_i, \qquad p \leftarrow B, \qquad q \leftarrow Q_{3-i}
\] by \eqref{eq:q1_q2_accepted_by_q_under_b_q1_q2} and \eqref{eq:b_accepted_by_q_underbp1p2_by_p_underbq1q2}. For these reasons we can hereinafter assume that \begin{equation}
    \label{eq:c_accepts_exactly_b_under_bp1p2_bq1q2}
    \Acc_c(\{B, P_1, P_2\}) = \Acc_c(\{B, Q_1, Q_2\}) = \{B\}
\end{equation}

\begin{claim}
    \label{claim:a_succc_p1p2q1q2}
    $ A \succ_c P_1, P_2, Q_1, Q_2 $
\end{claim}
\begin{proof}
    Choose $ \{i, j\} = \{1, 2\} $ so that $ P_j \prec_c P_i $. Since $ c $ does not accept $ P_i $ under $ \{B, P_1, P_2\} $ by \eqref{eq:c_accepts_exactly_b_under_bp1p2_bq1q2}, some one-good deletion of another bundle must be preferred by $ c $ to $ P_i $. This bundle cannot be $ P_j $ because $ P_j \prec_c P_i $, so it follows that \begin{equation}
        \label{claim:a_succc_p1p2q1q2:eq1}
        B \setminus \{b_P\} \succ_c P_i \succ_c P_j
    \end{equation} for some $ b_P \in B $. By the same argument applied to $ \{B, Q_1, Q_2\} $, there exists $ b_Q \in B $ such that \begin{equation}
        \label{claim:a_succc_p1p2q1q2:eq2}
        B \setminus \{b_Q\} \succ_c Q_1, Q_2
    \end{equation} Since $ A \in \Acc_c(\fancyP) $ by \eqref{eq:invariant}, we have in particular $ A \succ_c B \setminus \{b_P\} $ and $ A \succ_c B \setminus \{b_Q\} $, which together with \eqref{claim:a_succc_p1p2q1q2:eq1} and \eqref{claim:a_succc_p1p2q1q2:eq2} yields the claim.
\end{proof}

Having done the above analysis, we are now ready to conclude that there exists a critical conflict between $ p $ and $ q $ around the bundle $ B $. Set $ S := A \cup D $ to be the conflict region and let $ \{P_1, P_2\} $ and $ \{Q_1, Q_2\} $ be the $ p $-fair and $ q $-fair views, respectively. Indeed, we have already shown that $ \{P_1, P_2\} $ is EFX with respect to $ \prec_p $ and $ \{Q_1, Q_2\} $ is EFX with respect to $ \prec_q $, as required by Definition~\ref{definition:conflict}~\ref{definition:conflict:a}. The conditions of Definition~\ref{definition:conflict}~\ref{definition:conflict:b} follow from \eqref{eq:q_order} and \eqref{eq:p_order}. This shows that we have a conflict. By Claim \ref{claim:a_succc_p1p2q1q2}, $ c $ strictly prefers $ A $ over any bundle in $ \{P_1, P_2, Q_1, Q_2\} $ as required by Definition~\ref{definition:critical_conflict}~\ref{definition:critical_conflict:i}. Definition~\ref{definition:critical_conflict}~\ref{definition:critical_conflict:ii} follows from Claim \ref{claim:d_without_ga_succa_A_B} by setting $ d_p := g_p $ and $ d_q := g_q $. Finally, \eqref{eq:invariant} ensures that $ A, B \in \Acc_c(\{A, B, D\}) $ as required by Definition~\ref{definition:critical_conflict}~\ref{definition:critical_conflict:iii}. Hence, the conflict is in fact critical. Since $ c $ was chosen arbitrarily, repeating the proof yields that a critical conflict exists between every pair of agents unless an EFX allocation exists. This completes the proof of Theorem~\ref{theorem:critical_conflict_determine_efx}.

\section{EFX existence up to seven goods}
\label{section:efx_existence_up_to_seven_appendix}

We investigate the case of three agents and at most seven goods with the goal of rigorously proving the existence of EFX for any such monotone valuation profile, i.e., Theorem~\ref{thm:intro-seven}.

\subsection{Combinatorial rigidity}
\label{section:efx_existence_up_to_seven_appendix:rigidity}

First, we recall and prove Lemma~\ref{lemma:three_seven_basic_structure}.

\threesevenbasicstructurelemma*

\begin{proof}
By Lemma~\ref{lemma:conflict_properties}~\ref{lemma:conflict_properties:c} and Lemma~\ref{lemma:critical_conflict_properties}~\ref{lemma:critical_conflict_properties:c}, we have $ |B| \ge 2 $ and $ |S| \ge 5 $, respectively. Hence, $ |M| = |B| + |S| \ge 7 $. Since $ |M| \le 7 $, these inequalities must be equalities: \[
|M| = 7, \qquad |B| = 2, \qquad |S| = 5
\] This shows Item~\ref{lemma:three_seven_basic_structure:a}.
Write $ \fancyP = \{P_1, P_2\} $ and $ \fancyQ = \{Q_1, Q_2\} $. Since $ |P_1|, |P_2|, |Q_1|, |Q_2| \ge 2 $ by Lemma~\ref{lemma:conflict_properties}~\ref{lemma:conflict_properties:d}, from \[
5 = |S| = |P_1| + |P_2| = |Q_1| + |Q_2|
\] it follows that $ \{|P_1|, |P_2|\} = \{|Q_1|, |Q_2|\} = \{2, 3\} $. We can therefore write $ \fancyP = \{P^{(2)}, P^{(3)}\} $ and $ \fancyQ = \{Q^{(2)}, Q^{(3)}\} $, where the superscripts indicate cardinalities. The intersections $ P^{(2)} \cap Q^{(2)} $ and $ P^{(2)} \cap Q^{(3)} $ are nonempty by Lemma~\ref{lemma:conflict_properties}~\ref{lemma:conflict_properties:d} and have total size two. Hence, \[
|P^{(2)} \cap Q^{(2)}| = |P^{(2)} \cap Q^{(3)}| = 1
\] By the same argument, since $ Q^{(2)} $ has size two, from Lemma~\ref{lemma:conflict_properties}~\ref{lemma:conflict_properties:d} it follows that \[
|P^{(3)} \cap Q^{(2)}| = 1
\] The remaining two goods of $ P^{(3)} $ must therefore belong to $ P^{(3)} \cap Q^{(3)} $: \[
|P^{(3)} \cap Q^{(3)}| = 2
\] Since $ \{P^{(i)} \cap Q^{(j)} \mid i, j \in \{2, 3\}\} $ is a partition of $ S $ and we have determined the sizes of all blocks, we can name the goods in these four intersections as stated in Item~\ref{lemma:three_seven_basic_structure:b}.

It remains to prove Item~\ref{lemma:three_seven_basic_structure:c}. By Lemma~\ref{lemma:critical_conflict_properties}~\ref{lemma:critical_conflict_properties:a} and \ref{lemma:critical_conflict_properties:b} we have $ |A| \ge 2 $ and $ |D| \ge 3 $. Since $ 5 = |S| = |A| + |D| $, these inequalities turn into equalities $ |A| = 2 $ and $ |D| = 3 $. By Lemma~\ref{lemma:critical_conflict_properties}~\ref{lemma:critical_conflict_properties:b}, $ A $ must intersect all four fair-view blocks $ P^{(2)}, P^{(3)}, Q^{(2)}, Q^{(3)} $. Hence, the two elements of $ A $ must be located in diagonal cells of the intersection table of Item~\ref{lemma:three_seven_basic_structure:b}. There are only three such diagonal pairs, namely $ \{x, u\} $, $\{x, v\} $ and $ \{y, z\} $. Passing to the complements inside $ S $ yields that the respective values of $ D $ are $ \{y, z, v\} $, $ \{y, z, u\} $ and $ \{x, u, v\} $.
\end{proof}

Lemma~\ref{lemma:three_seven_basic_structure} shows that, when there are at most seven goods, every critical conflict has the same rigid combinatorial shape. The next lemma establishes a corresponding rigidity in the agents' preferences, unless an EFX allocation already exists. The key lies in the two-good pivot bundle $ B $. Intuitively, if the agents preferred different singleton goods of $ B $, then the two pivot goods could be separated and used to protect two different agents from strong envy, while the third agent receives a suitable three-good block from one of the fair views. Similarly, if $ c $ accepted any fair-view block, the remaining bundles could be assigned to $ p $ and $ q $ using the defining inequalities of the conflict. Consequently, we either obtain an EFX allocation, or all three agents must regard the same good $ b \in B $ as the more valuable pivot good. This good then becomes a common value barrier: each of $ p $ and $ q $ ranks the blocks of his own fair view above the entire pivot bundle, but the blocks of the other fair view below the singleton $ \{b\} $, while $ c $ also ranks $ \{b\} $ above every fair-view block. The lemma further tests whether pairing the other pivot good with certain goods from the conflict region can compensate for the absence of $ b $. Whenever it can, a natural pair--pair--triple allocation is EFX; in other words, if no EFX allocation exists, even these pairs must lie below $ \{b\} $. In this way, the lemma reduces the remaining complexity of the agents' preferences to a small collection of comparisons centered on one unanimously valuable good. %

\begin{lemma}
    \label{lemma:three_seven_efx_or_b_unanimity}
    Fix a nondegenerate monotone valuation profile for three agents and at most seven goods. Assume that agents $ p $ and $ q $ are in a critical conflict around bundle $ B $, with $ p $-fair view $ \fancyP $, $ q $-fair view $ \fancyQ $ and conflict region $ S $. Let $ c $ denote the remaining agent. By Lemma~\ref{lemma:three_seven_basic_structure}~\ref{lemma:three_seven_basic_structure:b}, it is well-defined to write $ \fancyP = \{P^{(2)}, P^{(3)}\} $ and $ \fancyQ = \{Q^{(2)}, Q^{(3)}\} $, where the superscripts indicate cardinalities. Assume furthermore that there is no EFX allocation. Then, there exists $ b \in B $ such that: \begin{enumerate}[label=(\alph*)]
        \item\label{lemma:three_seven_efx_or_b_unanimity:a} The following concrete version of Definition~\ref{definition:conflict}~\ref{definition:conflict:b} holds: \begin{align*}
            Q^{(2)}, Q^{(3)} \prec_p \{b\} &\prec_p B \prec_p P^{(2)}, P^{(3)} \\
            P^{(2)}, P^{(3)} \prec_q \{b\} &\prec_q B \prec_q Q^{(2)}, Q^{(3)}
        \end{align*}
        \item\label{lemma:three_seven_efx_or_b_unanimity:b} It holds that \[
        \{b\} \succ_c P^{(2)}, P^{(3)}, Q^{(2)}, Q^{(3)}
        \]
        \item\label{lemma:three_seven_efx_or_b_unanimity:c} All agents regard $ b $ as their uniquely most valuable good. In particular, it follows that $ b $ is uniquely determined by the valuation profile and is independent of the critical conflict and the labeling of the agents;
        \item\label{lemma:three_seven_efx_or_b_unanimity:d} Let $ e $ denote the remaining good in $ B = \{b, e\} $ (recall that $ |B| = 2 $ by Lemma~\ref{lemma:three_seven_basic_structure}~\ref{lemma:three_seven_basic_structure:a}). Let furthermore $ x, y, z, u, v $ denote the goods of $ S $ such that \[
        \begin{array}{c|cc}
            \cap & Q^{(2)} & Q^{(3)} \\ \hline
            P^{(2)} & \{x\} & \{y\} \\
            P^{(3)} & \{z\} & \{u, v\}
        \end{array}
        \] as provided for by Lemma~\ref{lemma:three_seven_basic_structure}~\ref{lemma:three_seven_basic_structure:b}. Then: \begin{align*}
            \{b\} &\succ_c \{x, e\}, \{y, e\}, \{z, e\} \\
            \{b\} &\succ_p \{x, e\}, \{z, e\} \\
            \{b\} &\succ_q \{x, e\}, \{y, e\}
        \end{align*}
    \end{enumerate}
\end{lemma}
\begin{proof}
For each agent $ i \in \{c, p, q\} $, let $ \beta_i \in B $ be the most valuable good in $ B $ for $ i $. Thus, from Definition~\ref{definition:conflict}~\ref{definition:conflict:b} it follows that \begin{equation}
    \label{eq:betap_betaq_domination}
    Q^{(2)}, Q^{(3)} \prec_p \{\beta_p\} \qquad\text{and}\qquad P^{(2)}, P^{(3)} \prec_q \{\beta_q\}
\end{equation} Observe that if $ P^{(s)} $ is accepted by $ c $ under $ \{B, P^{(2)}, P^{(3)}\} $ for some $ s \in \{2, 3\} $ and $ P^{(l)} $ denotes the other bundle in $ \fancyP $, then from Definition~\ref{definition:conflict}~\ref{definition:conflict:b} it immediately follows that the allocation \[
c \leftarrow P^{(s)}, \qquad p \leftarrow P^{(l)}, \qquad q \leftarrow B
\] is EFX, which is a contradiction. By the same argument, assuming $ Q^{(s)} \in \Acc_c(\{B, Q^{(2)}, Q^{(3)}\}) $ and writing $ Q^{(l)} $ for the other bundle in $ \fancyQ $ yields that \[
c \leftarrow Q^{(s)}, \qquad p \leftarrow B, \qquad q \leftarrow Q^{(l)}
\] is an EFX allocation, which is a contradiction. Hence, \begin{equation}
    \label{eq:acc_cbp2p3_acc_cbq2q3_b}
    \Acc_c(\{B, P^{(2)}, P^{(3)}\}) = \Acc_c(\{B, Q^{(2)}, Q^{(3)}\}) = \{B\}
\end{equation} Write $ \{P^{(j)}, P^{(k)}\} = \{P^{(2)}, P^{(3)}\} $ so that $ P^{(j)} \succ_c P^{(k)} $. Since $ c $ does not accept $ P^{(j)} $ under $ \{B, P^{(2)}, P^{(3)}\} $ by \eqref{eq:acc_cbp2p3_acc_cbq2q3_b}, $ P^{(j)} \succ_c P^{(k)} $ ensures that strong envy must be of the form $ P^{(j)} \prec_c B \setminus \{g\} $ for some $ g \in B $. Since $ \beta_c $ is most valuable to $ c $ inside $ B $, it follows that $ P^{(k)} \prec_c P^{(j)} \prec_c \{\beta_c\} $. By the same argument, we have $ Q^{(2)}, Q^{(3)} \prec_c \{\beta_c\} $. We conclude that \begin{equation}
    \label{eq:betac_succc_p23q23}
    \{\beta_c\} \succ_c P^{(2)}, P^{(3)}, Q^{(2)}, Q^{(3)}
\end{equation} Now suppose that $ \beta_c, \beta_p, \beta_q $ were not all equal. Let $ \{b, e\} = B $ denote the elements of $ B $. Since the above reasoning is invariant under swapping $ b $ and $ e $, this without loss of generality means that one agent must find $ e $ most valuable inside $ B $, while the remaining two agents should prefer $ b $ over $ e $. We therefore have the following three cases: \[
(\beta_c, \beta_p, \beta_q) \in \{(e, b, b), (b, e, b), (b, b, e)\}
\] For each of them, we construct an EFX allocation as follows: \begin{center}
    \begin{tabular}{ccc|ccc}
        \multicolumn{3}{c|}{Case} & \multicolumn{3}{c}{EFX allocation}\\\hline
        $ \beta_c $ & $ \beta_p $ & $ \beta_q $ & Bundle of $ c $ & Bundle of $ p $ & Bundle of $ q $\\\hline
        $ e $ & $ b $ & $ b $ & $ \{x, e\} $ & $ \{z, b\} $ & $ Q^{(3)} = \{y, u, v\} $\\
        $ b $ & $ e $ & $ b $ & $ \{x, b\} $ & $ \{z, e\} $ & $ Q^{(3)} = \{y, u, v\} $\\
        $ b $ & $ b $ & $ e $ & $ \{x, b\} $ & $ P^{(3)} = \{z, u, v\} $ & $ \{y, e\} $
    \end{tabular}
\end{center} We prove that these three allocations are EFX. We start by showing that in each case, $ c $ does not strongly envy $ p $ or $ q $. Note that the bundle we assign to $ c $ is always $ \{x, \beta_c\} $, so by \eqref{eq:betac_succc_p23q23} we have \[
\{x, \beta_c\} \succ_c \{\beta_c\} \overset{\eqref{eq:betac_succc_p23q23}}{\succ_c} Q^{(3)}, P^{(3)}, \{y\}, \{z\}
\] Taken together with the fact that $ \beta_c $ is more valuable for $ c $ than the remaining good in $ B $, it follows that every one-good deletion remainder of any bundle assigned to $ p $ or $ q $ has less value to $ c $ compared to the bundle given to $ c $, as desired.

We proceed by proving that the holder of the bundle of size three does not strongly envy anyone. In the first two cases this holder is agent $ q $, who receives $ Q^{(3)} $. Since $ \fancyQ $ is EFX with respect to $ \prec_q $, we have $ Q^{(3)} \succ_q \{x\}, \{z\} $. It also holds that $ Q^{(3)} \succ_q B \succ_q \{b\}, \{e\} $. Thus, we have bounded the valuation for all one-good deletion remainders of the bundles assigned to $ c $ and $ p $. Hence, in the first two cases, agent $ q $ accepts $ Q^{(3)} $. Turning to the third case, the reasoning is analogous. Since $ \fancyP $ is EFX with respect to $ \prec_p $, we have $ P^{(3)} \succ_p \{x\}, \{y\} $. Combined with $ P^{(3)} \succ_p B \succ_p \{b\}, \{e\} $, this means that $ p $ prefers $ P^{(3)} $ to all one-good deletion remainders of the bundles assigned to $ c $ and $ q $. Hence, $ p $ accepts $ P^{(3)} $ under the allocation of the third case.

It remains to show that the pair holder among $ \{p, q\} $ does not strongly envy anyone. If this holder is agent $ p $, as it is in the first two cases, then we note that $ p $ receives bundle $ \{z, \beta_p\} $. Hence, from \eqref{eq:betap_betaq_domination} it follows that \[
\{z, \beta_p\} \succ_p \{\beta_p\} \overset{\eqref{eq:betap_betaq_domination}}{\succ_p} Q^{(2)} \succ_p \{x\}
\] Read together with the choice of $ \beta_p $ as the most valuable good in $ B $ from the perspective of $ p $, this means that agent $ p $ does not strongly envy agent $ c $ in the first two cases. By \eqref{eq:betap_betaq_domination} we also have \[
\{z, \beta_p\} \succ_p \{\beta_p\} \overset{\eqref{eq:betap_betaq_domination}}{\succ_p} Q^{(3)}
\] meaning that in the first two cases, agent $ p $ does not envy $ q $, let alone strongly envy. Similar reasoning applies in the third case, where the pair holder is $ q $. Indeed, then it is clear that $ \{y, e\} \succ_q \{e\} = \{\beta_q\} \succ_q \{b\} $ by the choice of $ \beta_q $ and from \eqref{eq:betap_betaq_domination} it moreover follows that \[
\{y, e\} \succ_q \{e\} = \{\beta_q\} \overset{\eqref{eq:betap_betaq_domination}}{\succ_q} P^{(2)} \succ_q \{x\}
\] Hence, agent $ q $ does not strongly envy agent $ c $. Similarly, since by \eqref{eq:betap_betaq_domination} we have \[
\{y, e\} \succ_q \{e\} = \{\beta_q\} \overset{\eqref{eq:betap_betaq_domination}}{\succ_q} P^{(3)}
\] there is no envy from $ q $ towards $ p $, let alone strong envy.

We conclude that in all three cases, the stated allocations are EFX, which is a contradiction. Hence, the three goods $ \beta_c, \beta_p, \beta_q $ coincide and we can define \begin{equation}
    \label{eq:betacpq_all_equal_b}
    b := \beta_c = \beta_p = \beta_q
\end{equation} By combining \eqref{eq:betacpq_all_equal_b} with \eqref{eq:betap_betaq_domination} and Definition~\ref{definition:conflict}~\ref{definition:conflict:b}, we obtain Item~\ref{lemma:three_seven_efx_or_b_unanimity:a}. Furthermore, substituting \eqref{eq:betacpq_all_equal_b} into \eqref{eq:betac_succc_p23q23} yields Item~\ref{lemma:three_seven_efx_or_b_unanimity:b}.

We show Item~\ref{lemma:three_seven_efx_or_b_unanimity:c}. By \eqref{eq:betacpq_all_equal_b}, $ b $ is the most valuable good inside $ B $ for all agents, so it remains to prove that $ b $ is more valuable than any good in $ S $ from the perspective of all agents. Indeed, from \eqref{eq:betacpq_all_equal_b} and \eqref{eq:betap_betaq_domination} it follows that agents $ p $ and $ q $ strictly prefer $ b $ over any other good in $ S $, and \eqref{eq:betac_succc_p23q23} read in conjunction with \eqref{eq:betacpq_all_equal_b} yields the same conclusion for agent $ c $. Hence, Item~\ref{lemma:three_seven_efx_or_b_unanimity:c} follows.

\begin{figure}[t]
    \centering
    \begin{tabular}{c|ccc}
        Condition & \multicolumn{3}{c}{EFX allocation}\\\hline
        & Bundle of $ c $ & Bundle of $ p $ & Bundle of $ q $\\\hline
        $ \{x, e\} \succ_c \{b\} $ & $ \{x, e\} $ & $ \{z, b\} $ & $ Q^{(3)} = \{y, u, v\} $ \\
        $ \{y, e\} \succ_c \{b\} $ & $ \{y, e\} $ & $ P^{(3)} = \{z, u, v\} $ & $ \{x, b\} $ \\
        $ \{z, e\} \succ_c \{b\} $ & $ \{z, e\} $ & $ \{x, b\} $ & $ Q^{(3)} = \{y, u, v\} $ \\\hline
        $ \{x, e\} \succ_p \{b\} $ & $ \{z, b\} $ & $ \{x, e\} $ & $ Q^{(3)} = \{y, u, v\} $ \\
        $ \{z, e\} \succ_p \{b\} $ & $ \{x, b\} $ & $ \{z, e\} $ & $ Q^{(3)} = \{y, u, v\} $ \\\hline
        $ \{x, e\} \succ_q \{b\} $ & $ \{y, b\} $ & $ P^{(3)} = \{z, u, v\} $ & $ \{x, e\} $ \\
        $ \{y, e\} \succ_q \{b\} $ & $ \{x, b\} $ & $ P^{(3)} = \{z, u, v\} $ & $ \{y, e\} $ \\\hline
    \end{tabular}
    \caption{List of allocations that are EFX if the corresponding condition holds.}
    \label{fig:conditional_efx_allocations}
\end{figure} We prove Item~\ref{lemma:three_seven_efx_or_b_unanimity:d}. Suppose that at least one inequality listed in Item~\ref{lemma:three_seven_efx_or_b_unanimity:d} were false. Since the valuations are nondegenerate, the corresponding strict reverse inequality appearing in Figure~\ref{fig:conditional_efx_allocations} would then hold. We show that the allocation in the same row is EFX, which yields a contradiction. In other words, the proof reduces to showing that the allocations listed in Figure~\ref{fig:conditional_efx_allocations} are indeed EFX provided that the corresponding additional condition holds.

All EFX allocations in Figure~\ref{fig:conditional_efx_allocations} follow the same idea. For this reason we now focus on explaining the intuition behind the construction, using which the correctness of Figure~\ref{fig:conditional_efx_allocations} can be easily verified, rather than repetitively apply the same reasoning to every concrete case.

We first choose one of the size-three fair-view blocks and assign it to the owner of that fair view: $ P^{(3)} $ is assigned to $ p $, while $ Q^{(3)} $ is given to $ q $. If we choose $ Q^{(3)} = \{y, u, v\} $, then the remaining goods of $ S $ are $ Q^{(2)} = \{x, z\} $. Similarly, if we choose $ P^{(3)} = \{z, u, v\} $, then the remaining goods of $ S $ are $ P^{(2)} = \{x, y\} $. Together with $ B = \{b, e\} $, we thus have four goods to divide, namely $ \{x, z, b, e\} $ in the former case and $ \{x, y, b, e\} $ in the latter one. We divide these four remaining goods into two pairs, where, crucially, we place $ b $ and $ e $ in different pairs. We now explain why this construction ensures absence of strong envy.

\paragraph{The triple receiver does not strongly envy.} Consider the holder of the size-three bundle. Since this bundle belongs to the agent's own EFX fair view, the agent prefers it to each of the two singletons obtained by deleting one good from the complementary size-two block. Moreover, by Item~\ref{lemma:three_seven_efx_or_b_unanimity:a}, the agent prefers the size-three bundle to $ B $, and therefore also to both $ \{b\} $ and $ \{e\} $. These four singletons are exactly the one-good deletion remainders of the two pairs. Hence, the triple holder does not strongly envy either pair holder.

\paragraph{The receiver of the pair containing $ b $ does not strongly envy.} The agent holding the pair containing $ b $ is chosen so that this agent regards the selected size-three bundle as less valuable than $ \{b\} $. If the pair holder is $ p $ or $ q $, this follows from Item~\ref{lemma:three_seven_efx_or_b_unanimity:a}, because the triple comes from the other agent's fair view. If the pair holder is $c$, it follows from Item~\ref{lemma:three_seven_efx_or_b_unanimity:b}. Thus, the $ b $-pair holder does not envy the triple holder, let alone strongly envy. Moreover, by Item~\ref{lemma:three_seven_efx_or_b_unanimity:c} we have that $ b $ is the uniquely most valuable good, so the $ b $-pair holder also prefers its bundle to both one-good deletion remainders of the pair containing $ e $. We conclude that the $ b $-pair holder does not strongly envy anyone.

\paragraph{The receiver of the pair containing $ e $ does not strongly envy.} Now consider the holder of the pair containing $ e $. The condition in the corresponding row of Figure~\ref{fig:conditional_efx_allocations} states precisely that this pair is preferred to $ \{b\} $. The selected size-three bundle is again chosen so that this agent values it below $ \{b\} $: for agent $ p $ we use $ Q^{(3)} $, for agent $ q $ we use $ P^{(3)} $, and for agent $ c $ both fair-view triples are valued below $ \{b\} $, so we choose one that is disjoint from the prescribed pair containing $ e $, as specified in Figure~\ref{fig:conditional_efx_allocations}. It follows that the $ e $-pair holder does not envy the triple holder, let alone strongly envy. The one-good deletion remainders of the $ b $-pair are $ \{b\} $ and another singleton, and the latter is less valuable than $ \{b\} $ by Item~\ref{lemma:three_seven_efx_or_b_unanimity:c}. Since the $ e $-pair holder prefers its bundle to $ \{b\} $ by the condition, it therefore does not strongly envy the $ b $-pair holder either.
\end{proof}

The next technical lemma studies in detail how two critical conflicts sharing a single agent interact: it expresses the pivot bundle of one conflict in terms of the fair views belonging to the other. Suppose that $ p $ participates both in a critical conflict with $ q $ and in a critical conflict with $ c $. Under the assumption that no EFX allocation exists, Lemma~\ref{lemma:three_seven_efx_or_b_unanimity} already forces the two pivot bundles to contain the same unanimously most valuable good $ b $. It therefore remains to understand where the other good of the second pivot can be located. The shared agent $ p $ provides the link between the two conflicts: in each conflict, $ p $ ranks both blocks of his own fair view above the corresponding pivot bundle, but ranks both blocks of the other agent's fair view below $ \{b\} $. Consequently, the two pairs of fair views cannot be arranged independently. The rigid five-good intersection pattern of Lemma~\ref{lemma:three_seven_basic_structure} leaves only a small number of possible ways in which the second conflict can intersect the first one, and a careful analysis of them using Lemma~\ref{lemma:three_seven_efx_or_b_unanimity} shows that the other good in the second pivot must be the good in $ P_c^{(2)} \cap Q_c^{(2)} $.

\begin{lemma}
    \label{lemma:three_seven_pivots}
    Fix a nondegenerate monotone valuation profile for three agents and at most seven goods. Denote the three agents $ p $, $ q $ and $ c $. Suppose that \begin{itemize}
        \item there is a critical conflict between $ p $ and $ q $, with $ p $-fair view $ \fancyP_c $ and $ q $-fair view $ \fancyQ_c $. By Lemma~\ref{lemma:three_seven_basic_structure}~\ref{lemma:three_seven_basic_structure:b}, it is well-defined to write $ \fancyP_c = \{P_c^{(2)}, P_c^{(3)}\} $ and $ \fancyQ_c = \{Q_c^{(2)}, Q_c^{(3)}\} $, where the superscripts indicate cardinalities;
        \item there is a critical conflict between $ p $ and $ c $ with pivot bundle $ B_q $;
        \item there is no EFX allocation;
    \end{itemize} Then, $ B_q = \{b, t\} $, where $ b $ is the unique most-valuable good of Lemma~\ref{lemma:three_seven_efx_or_b_unanimity} and $ t $ is the good such that $ \{t\} = P_c^{(2)} \cap Q_c^{(2)} $ (recall that $ |P_c^{(2)} \cap Q_c^{(2)}| = 1 $ by Lemma~\ref{lemma:three_seven_basic_structure}~\ref{lemma:three_seven_basic_structure:b}).
\end{lemma}
\begin{proof}
    See Appendix~\ref{section:three_seven_pivots}.
\end{proof}

\subsection{Proof of Theorem~\ref{thm:intro-seven}}

We are now ready to rigorously prove Theorem~\ref{thm:intro-seven}. By Proposition~\ref{prop:nondegenerate_perturbation}, it suffices to consider nondegenerate valuations. Let $ p $, $ q $ and $ c $ denote the three agents. Towards a contradiction, suppose that no EFX allocation exists. By Theorem~\ref{theorem:critical_conflict_determine_efx}, there is a critical conflict between every pair of agents. We denote the components of these three critical conflicts using the following notation: \begin{center}
    \begin{tabular}{c|c|c|c|c}
        Critical conflict & Other agent & Pivot bundle & Conflict region & Fair views\\\hline
        $ p, q $ & $ c $ & $ B_c $ & $ S_c $ & $ \begin{aligned}
            \fancyP_c &= \{P_c^{(2)}, P_c^{(3)}\}\\
            \fancyQ_c &= \{Q_c^{(2)}, Q_c^{(3)}\}
        \end{aligned} $ \\\hline
        $ p, c $ & $ q $ & $ B_q $ & $ S_q $ & $ \begin{aligned}
            \fancyP_q &= \{P_q^{(2)}, P_q^{(3)}\}\\
            \fancyC_q &= \{C_q^{(2)}, C_q^{(3)}\}
        \end{aligned} $ \\\hline
        $ q, c $ & $ p $ & $ B_p $ & $ S_p $ & $ \begin{aligned}
            \fancyQ_p &= \{Q_p^{(2)}, Q_p^{(3)}\}\\
            \fancyC_p &= \{C_p^{(2)}, C_p^{(3)}\}
        \end{aligned} $ \\\hline
    \end{tabular}
\end{center} More precisely, the superscripts in the names of fair-view blocks indicate their respective cardinality. Indeed, we recall that by Lemma~\ref{lemma:three_seven_basic_structure}~\ref{lemma:three_seven_basic_structure:b}, every fair view has one block of size two and one block of size three.

By Lemma~\ref{lemma:three_seven_basic_structure}~\ref{lemma:three_seven_basic_structure:b}, the intersection of the size-two blocks of the two fair views of the critical conflict between $ p $ and $ q $ consists of exactly one good. Denote this good by $ x $, so that $ P_c^{(2)} \cap Q_c^{(2)} = \{x\} $. By Lemma~\ref{lemma:three_seven_efx_or_b_unanimity}~\ref{lemma:three_seven_efx_or_b_unanimity:c}, there exists a unique good $ b \in B_c \cap B_q \cap B_p $ that all agents find to be most valuable among single goods. Note that $ x \neq b $, because $ x \in S_c = M \setminus B_c $ whereas $ b \in B_c $.

Apply Lemma~\ref{lemma:three_seven_pivots} to the critical conflicts $ (p, q) $ and $ (p, c) $. Since $ P_c^{(2)} \cap Q_c^{(2)} = \{x\} $, we obtain $ B_q = \{b, x\} $. Applying the same lemma to the critical conflicts $ (q, p) $ and $ (q, c) $, with the roles of $ p $ and $ q $ exchanged, similarly gives $ B_p = \{b, x\} $. Hence, \begin{equation}
    \label{eq:final_three_seven_equal_pivots}
    B_q = B_p = \{b, x\}
\end{equation}

We now apply Lemma~\ref{lemma:three_seven_pivots} to the critical conflicts $ (c, p) $ and $ (c, q) $. The fair views of the first conflict are $ \fancyC_q $ and $ \fancyP_q $, while the pivot bundle of the second conflict is $ B_p $. Thus, the lemma yields $ B_p \setminus \{b\} \subseteq C_q^{(2)} \cap P_q^{(2)} $. Combining this with \eqref{eq:final_three_seven_equal_pivots} and $ x \neq b $ gives \[
\{x\} = B_p \setminus \{b\}
\subseteq C_q^{(2)} \cap P_q^{(2)}
\subseteq S_q = M \setminus B_q
\] which contradicts $ x \in B_q $ from \eqref{eq:final_three_seven_equal_pivots}. Consequently, the assumption that no EFX allocation exists must be false. This completes the proof of Theorem~\ref{thm:intro-seven}.

\section{Proof of Lemma~\ref{lemma:three_seven_pivots}}
\label{section:three_seven_pivots}

We prove Lemma~\ref{lemma:three_seven_pivots}. We start by fixing the notation and framework.

\subsection{Preliminaries}

We use the following notation for the components of the two given critical conflicts: \begin{center}
    \begin{tabular}{c|c|c|c|c|c}
        Conflict & Other agent & Pivot bundle & Conflict region & Fair views & Critical cut\\\hline
        $ p, q $ & $ c $ & $ B_c $ & $ S_c $ & $ \fancyP_c, \fancyQ_c $ & $ (A_c, D_c) $\\
        $ p, c $ & $ q $ & $ B_q $ & $ S_q $ & $ \fancyP_q, \fancyC_q $ & $ (A_q, D_q) $
    \end{tabular}
\end{center} By Lemma~\ref{lemma:three_seven_basic_structure}~\ref{lemma:three_seven_basic_structure:b}, it is well-defined to write $ \fancyP_c = \{P_c^{(2)}, P_c^{(3)}\} $ and $ \fancyQ_c = \{Q_c^{(2)}, Q_c^{(3)}\} $, where the superscripts indicate cardinalities, and let $ x, y, z, u, v $ be distinct goods such that \[
\begin{array}{c|cc}
    \cap & Q_c^{(2)} & Q_c^{(3)} \\ \hline
    P_c^{(2)} & \{x\} & \{y\} \\
    P_c^{(3)} & \{z\} & \{u, v\}
\end{array}
\] In this notation, our goal is to show that $ B_q = \{b, x\} $ where $ b $ is the unique most-valuable good of Lemma~\ref{lemma:three_seven_efx_or_b_unanimity}.

By Lemma~\ref{lemma:three_seven_basic_structure}~\ref{lemma:three_seven_basic_structure:a} it is clear that $ |B_q| = 2 $. Let $ b $ be the unique good of Lemma~\ref{lemma:three_seven_efx_or_b_unanimity}. Note that by Lemma~\ref{lemma:three_seven_efx_or_b_unanimity}~\ref{lemma:three_seven_efx_or_b_unanimity:c}, $ b $ does not depend on the choice of the critical conflict. Hence, $ b $ is well-defined and its properties listed in Lemma~\ref{lemma:three_seven_efx_or_b_unanimity} hold for both critical conflicts. Since $ b \in B_q $, the proof reduces to determining the remaining element of $ B_q $. Let $ t $ denote that remaining element of $ B_q = \{b, t\} $. Write $ \fancyP_q = \{P_q^{(2)}, P_q^{(3)}\} $ and $ \fancyC_q = \{C_q^{(2)}, C_q^{(3)}\} $, where the superscripts indicate cardinalities. Indeed, this is well-defined by Lemma~\ref{lemma:three_seven_basic_structure}~\ref{lemma:three_seven_basic_structure:b}. By Items \ref{lemma:three_seven_efx_or_b_unanimity:a} and \ref{lemma:three_seven_efx_or_b_unanimity:b} of Lemma~\ref{lemma:three_seven_efx_or_b_unanimity} applied to the critical conflict between $ p $ and $ c $, we have \begin{equation}
    \label{eq:ts_cq_pq_relationships}
    \begin{aligned}
        C_q^{(2)}, C_q^{(3)} \prec_p \{b\} &\prec_p B_q = \{b, t\} \prec_p P_q^{(2)}, P_q^{(3)} \\
        P_q^{(2)}, P_q^{(3)} \prec_c \{b\} &\prec_c B_q = \{b, t\} \prec_c C_q^{(2)}, C_q^{(3)}
    \end{aligned}
\end{equation} and \begin{equation}
    \label{eq:b_succq_Pq2_Pq3_Cq2_Cq3}
    \{b\} \succ_q P_q^{(2)}, P_q^{(3)}, C_q^{(2)}, C_q^{(3)}
\end{equation} respectively. By Lemma~\ref{lemma:three_seven_basic_structure}~\ref{lemma:three_seven_basic_structure:b}, there exist distinct goods $ \alpha, \beta, \gamma, \theta, \omega $ such that \[
\begin{array}{c|cc}
    \cap & C_q^{(2)} & C_q^{(3)} \\ \hline
    P_q^{(2)} & \{\alpha\} & \{\beta\} \\
    P_q^{(3)} & \{\gamma\} & \{\theta, \omega\}
\end{array}
\] Hence, by Lemma~\ref{lemma:three_seven_efx_or_b_unanimity}~\ref{lemma:three_seven_efx_or_b_unanimity:d}, \begin{equation}
    \label{eq:b_succq_alphabetagammat}
    \begin{aligned}
        \{b\} &\succ_q \{\alpha, t\}, \{\beta, t\}, \{\gamma, t\} \\
        \{b\} &\succ_p \{\alpha, t\}, \{\gamma, t\} \\
        \{b\} &\succ_c \{\alpha, t\}, \{\beta, t\}
    \end{aligned}
\end{equation}
\begin{figure}[t]
    \centering
    \begin{tabular}{c|c|c|c|c}
        Conflict & Pivot bundle & Conflict region & Fair views & Intersections\\\hline
        $ p, q $ & $ B_c = \{b, e\} $ & $ S_c = \{x, y, z, u, v\} $ & $ \begin{aligned}
            \fancyP_c &= \{P_c^{(2)}, P_c^{(3)}\}\\
            \fancyQ_c &= \{Q_c^{(2)}, Q_c^{(3)}\}
        \end{aligned} $ & $ \begin{array}{c|cc}
            \cap & Q_c^{(2)} & Q_c^{(3)} \\ \hline
            P_c^{(2)} & \{x\} & \{y\} \\
            P_c^{(3)} & \{z\} & \{u, v\}
        \end{array} $\\\hline
        $ p, c $ & $ B_q = \{b, t\} $ & $ S_q = \{\alpha, \beta, \gamma, \theta, \omega\} $ & $ \begin{aligned}
            \fancyP_q &= \{P_q^{(2)}, P_q^{(3)}\}\\
            \fancyC_q &= \{C_q^{(2)}, C_q^{(3)}\}
        \end{aligned} $ & $ \begin{array}{c|cc}
            \cap & C_q^{(2)} & C_q^{(3)} \\ \hline
            P_q^{(2)} & \{\alpha\} & \{\beta\} \\
            P_q^{(3)} & \{\gamma\} & \{\theta, \omega\}
        \end{array} $\\\hline
    \end{tabular}
    \caption{Structure of the critical conflicts in the proof of Lemma~\ref{lemma:three_seven_pivots}.}
    \label{fig:three_seven_pivots}
\end{figure}

\noindent Let $ e $ denote the good such that $ B_c = \{b, e\} $. For the reader's convenience, we summarize the main objects and our notation in Figure~\ref{fig:three_seven_pivots}. Since our goal is to prove that $ t = x $, we have to show $ t \notin \{e, y, z, u, v\} $. We rule out these five values of $ t $ one-by-one in Sections~\ref{section:three_seven_pivots:t_e}, \ref{section:three_seven_pivots:t_y}, \ref{section:three_seven_pivots:t_z}, \ref{section:three_seven_pivots:t_u} and \ref{section:three_seven_pivots:t_v} below. To help the reader better understand the proof, we start by explaining the most important intuitions and ideas behind the rigorous argumentation below.

\subsection{Intuition}

The main connection between the two critical conflicts is the good $ b $, for which Lemma~\ref{lemma:three_seven_efx_or_b_unanimity} gives strong preference comparisons with the other relevant bundles. More precisely, Lemma~\ref{lemma:three_seven_efx_or_b_unanimity}~\ref{lemma:three_seven_efx_or_b_unanimity:a} places, from the perspective of each conflicting agent, the blocks of that agent's own fair view above the pivot bundle and the blocks of the other agent's fair view below the singleton $ \{b\} $. Lemma~\ref{lemma:three_seven_efx_or_b_unanimity}~\ref{lemma:three_seven_efx_or_b_unanimity:b} places all four fair-view blocks below $ \{b\} $ from the perspective of the remaining agent. Since Lemma~\ref{lemma:three_seven_efx_or_b_unanimity}~\ref{lemma:three_seven_efx_or_b_unanimity:c} ensures that the same good $ b $ belongs to the pivot bundles of both critical conflicts, we can use $ \{b\} $ as a common benchmark against which the components of the two conflicts can be compared. Equations~\eqref{eq:ts_cq_pq_relationships} and \eqref{eq:b_succq_Pq2_Pq3_Cq2_Cq3} record these consequences for the second critical conflict, but throughout the proof we also use the corresponding consequences of the same lemma for the original conflict. Lemma~\ref{lemma:three_seven_efx_or_b_unanimity}~\ref{lemma:three_seven_efx_or_b_unanimity:d}, specialized to the second conflict in \eqref{eq:b_succq_alphabetagammat}, provides analogous comparisons for certain pairs containing the other pivot good.

We convert these comparisons with $ \{b\} $ into strong constraints on how the fair views of the two conflicts may intersect. The basic principle is simple: if $ X \succ_i \{b\} \succ_i Y$, then monotonicity rules out $ X \subseteq Y $. When $ Y $ is one block of a binary partition, showing that $ X $ cannot be contained in either block forces $ X $ to intersect both of them. Together with the rigid structure of the fair views established by Lemma~\ref{lemma:three_seven_basic_structure}, a small number of such restrictions often determines the relevant fair-view blocks completely.

In the more involved cases, we first determine $ C_q^{(2)} $, the size-two block in $ c $'s fair view of the critical conflict between $ p $ and $ c $. This is the natural starting point because the two applications of Lemma~\ref{lemma:three_seven_efx_or_b_unanimity} impose particularly strong restrictions on the possible location of this bundle. Applied to the conflict between $ p $ and $ c $, Lemma~\ref{lemma:three_seven_efx_or_b_unanimity}~\ref{lemma:three_seven_efx_or_b_unanimity:a} places $ C_q^{(2)} $ above $ \{b\} $ from $ c $'s perspective. Applied to the original conflict between $ p $ and $ q $, Lemma~\ref{lemma:three_seven_efx_or_b_unanimity}~\ref{lemma:three_seven_efx_or_b_unanimity:b} places every block of $ \fancyP_c $ and $ \fancyQ_c $ below $ \{b\} $ from the same perspective, while Lemma~\ref{lemma:three_seven_efx_or_b_unanimity}~\ref{lemma:three_seven_efx_or_b_unanimity:d} places several additional pairs containing $ e $ below $ \{b\} $. Consequently, $ C_q^{(2)} $ cannot be contained in any of these low bundles, which leaves only a few possible pairs.

Once $ C_q^{(2)} $ has been fixed, we determine $ P_q^{(2)} $. The rigid intersection pattern of Lemma~\ref{lemma:three_seven_basic_structure}~\ref{lemma:three_seven_basic_structure:b} requires $ P_q^{(2)} $ to contain exactly one good of $ C_q^{(2)} $, so only a short list of candidates remains. On the one hand, applied to the conflict between $ p $ and $ c $, Lemma~\ref{lemma:three_seven_efx_or_b_unanimity}~\ref{lemma:three_seven_efx_or_b_unanimity:a} places $ P_q^{(2)} $ above $ \{b\} $ from $ p $'s perspective. On the other hand, applied to the original conflict between $ p $ and $ q $, the same item places the blocks of $ \fancyQ_c $ below $ \{b\} $ from $ p $'s perspective, while Lemma~\ref{lemma:three_seven_efx_or_b_unanimity}~\ref{lemma:three_seven_efx_or_b_unanimity:d} provides further pairs lying below $ \{b\} $. These comparisons eliminate most candidates for $ P_q^{(2)} $. The EFX conditions on $ p $'s two fair views eliminate further candidates by forcing opposite comparisons between the same two bundles. Once $ C_q^{(2)} $ and $ P_q^{(2)} $ are determined, their size-three complements and the intersection goods $ \alpha, \beta, \gamma $ become fixed as well. This then allows us to apply \eqref{eq:b_succq_alphabetagammat} to concrete pairs and analyze the small number of possible critical cuts.

Once the possible intersections have been narrowed down, we repeatedly use three kinds of contradictions: \begin{itemize}
    \item A proposed block may be forced above $ \{b\} $ by one conflict while being contained in a bundle lying below $ \{b\} $ according to the other conflict.
    \item The EFX properties of two fair views belonging to the same agent may impose opposite comparisons on the same two bundles: one view forces $ X \succ_i Y $, while the other forces $ Y\succ_i X $.
    \item The critical cuts may become incompatible. More precisely, Definition~\ref{definition:critical_conflict}~\ref{definition:critical_conflict:ii} requires an appropriate one-good deletion remainder of $ D_c $ to lie above $ B_c $ for each of $ p $ and $ q $, whereas the combined comparisons from the two conflicts sometimes force every such remainder below $ \{b\} $, and hence below $ B_c $. When the two conflict regions coincide, the intersection constraints may instead force the two critical cuts to coincide, after which Definition~\ref{definition:critical_conflict}~\ref{definition:critical_conflict:ii} for one conflict contradicts the acceptance condition (Definition~\ref{definition:critical_conflict}~\ref{definition:critical_conflict:iii}) for the other.
\end{itemize}

Thus, the proof strategy in the more involved branches is to first locate the pair $C_q^{(2)}$, which lies above $\{b\}$ for $c$, then locate the pair $P_q^{(2)}$, which lies above $\{b\}$ for $p$, and finally compare the resulting critical-cut remainders against the common benchmark $\{b\}$.

\subsection{Ruling out the case $ t = e $}
\label{section:three_seven_pivots:t_e}

Suppose first that $ t = e $. Then $ B_c = B_q = \{b, e\} $ and the two conflict regions must coincide: \[
S_c = S_q = \{x, y, z, u, v\}
\] Let \[
\mathcal{T} := \bigl\{\{x,u\},\{x,v\},\{y,z\}\bigr\}
\] be the set of good pairs that lie in diagonally opposite cells of the intersection table given in Lemma~\ref{lemma:three_seven_basic_structure}~\ref{lemma:three_seven_basic_structure:b}. In other words, every pair in $ S_c $ but not in $ \mathcal{T} $ is contained in at least one of $ P_c^{(2)}, P_c^{(3)}, Q_c^{(2)}, Q_c^{(3)} $. In particular, this means that \begin{equation}
    \label{eq:t_e_Cq2_transversal}
    C_q^{(2)} \in \mathcal{T}
\end{equation} because by \eqref{eq:ts_cq_pq_relationships} and Lemma~\ref{lemma:three_seven_efx_or_b_unanimity}~\ref{lemma:three_seven_efx_or_b_unanimity:b} we have \[
C_q^{(2)} \overset{\eqref{eq:ts_cq_pq_relationships}}{\succ_c} \{b\} \overset{\ref{lemma:three_seven_efx_or_b_unanimity}\ref{lemma:three_seven_efx_or_b_unanimity:b}}{\succ_c} P_c^{(2)}, P_c^{(3)}, Q_c^{(2)}, Q_c^{(3)}
\] which is a contradiction if $ C_q^{(2)} $ is contained in one of $ P_c^{(2)}, P_c^{(3)}, Q_c^{(2)}, Q_c^{(3)} $.

We now analyze the relationship between $ A_c $ and the fair view blocks of the conflict between $ p $ and $ c $. We first claim that \begin{equation}
    \label{eq:t_e_Ac_intersects_Cq2}
    A_c \cap C_q^{(2)} \neq \varnothing
\end{equation} Towards a contradiction, assume that $ A_c \cap C_q^{(2)} = \varnothing $. Then $ C_q^{(2)} \subseteq D_c $. Since $ |D_c| = 3 $ by Lemma~\ref{lemma:three_seven_basic_structure}~\ref{lemma:three_seven_basic_structure:c}, it follows that $ C_q^{(2)} $ can be obtained from $ D_c $ by removing one good. Hence, using that $ B_c \in \Acc_c(\{A_c, B_c, D_c\}) $ by Definition~\ref{definition:critical_conflict}~\ref{definition:critical_conflict:iii} gives $ B_c \succ_c C_q^{(2)} $, which is a contradiction to \eqref{eq:ts_cq_pq_relationships}. 

Next, we claim that \begin{equation}
    \label{eq:t_e_Ac_not_Cq2}
    A_c \neq C_q^{(2)}
\end{equation} Towards a contradiction, suppose that $ A_c = C_q^{(2)} $. Passing to the complements inside the coinciding conflict regions gives $ D_c = C_q^{(3)} $. Combining this with \eqref{eq:ts_cq_pq_relationships} and Definition~\ref{definition:critical_conflict}~\ref{definition:critical_conflict:ii} yields \[
\{b\} \overset{\eqref{eq:ts_cq_pq_relationships}}{\succ_p} C_q^{(3)} = D_c \overset{\ref{definition:critical_conflict}\ref{definition:critical_conflict:ii}}{\succ_p} B_c = \{b, e\}
\] which is a contradiction to monotonicity.

As a next step, we show that \begin{equation}
    \label{eq:t_e_Ac_crosses_Pq}
    A_c \cap P_q^{(2)} \neq \varnothing \qquad\text{and}\qquad A_c \cap P_q^{(3)} \neq \varnothing
\end{equation} To see this, we first note that, since $ c $ accepts $ A_c $ under $ \{A_c, B_c, D_c\} $ by Definition~\ref{definition:critical_conflict}~\ref{definition:critical_conflict:iii}, we have \[
A_c \succ_c B_c \setminus \{e\} = \{b\}
\] Combining this with \eqref{eq:ts_cq_pq_relationships} yields \begin{equation}
    \label{eq:t_e_Ac_succc_Pq_blocks}
    A_c \succ_c \{b\} \overset{\eqref{eq:ts_cq_pq_relationships}}{\succ_c} P_q^{(2)}, P_q^{(3)}
\end{equation} Thus, if one of the intersections in \eqref{eq:t_e_Ac_crosses_Pq} were empty, then $ A_c $ would be a subset of the other block of $ \fancyP_q $, which contradicts \eqref{eq:t_e_Ac_succc_Pq_blocks}. This proves \eqref{eq:t_e_Ac_crosses_Pq}.

We will also need the analogous crossing property for $ P_q^{(2)} $ with respect to the original $ q $-fair view: \begin{equation}
    \label{eq:t_e_Pq2_crosses_Qc}
    P_q^{(2)} \cap Q_c^{(2)} \neq \varnothing \qquad\text{and}\qquad P_q^{(2)} \cap Q_c^{(3)}\neq\varnothing
\end{equation}
Indeed, by \eqref{eq:ts_cq_pq_relationships} and Lemma~\ref{lemma:three_seven_efx_or_b_unanimity}~\ref{lemma:three_seven_efx_or_b_unanimity:a}, \begin{equation}
    \label{eq:t_e_pq2_succp_b_succp_qc2_qc3}
    P_q^{(2)} \overset{\eqref{eq:ts_cq_pq_relationships}}{\succ_p} \{b\} \overset{\ref{lemma:three_seven_efx_or_b_unanimity}\ref{lemma:three_seven_efx_or_b_unanimity:a}}{\succ_p} Q_c^{(2)},Q_c^{(3)}
\end{equation} Thus, if $ P_q^{(2)} $ failed to intersect one of the two blocks of $ \fancyQ_c $, then it would be contained in the other block, contradicting \eqref{eq:t_e_pq2_succp_b_succp_qc2_qc3}. This shows \eqref{eq:t_e_Pq2_crosses_Qc}.

Before distinguishing the three possibilities for $ C_q^{(2)} $, we establish two further properties of $ A_q $. First, we show that \begin{equation}
    \label{eq:t_e_Aq_intersects_Qc2}
    A_q \cap Q_c^{(2)} \neq \varnothing
\end{equation} Towards a contradiction, assume $ A_q \cap Q_c^{(2)} = \varnothing $. Since $ Q_c^{(2)} = \{x, z\} $ has size two and $ D_q = S_q \setminus A_q $ has size three (by Lemma~\ref{lemma:three_seven_basic_structure}~\ref{lemma:three_seven_basic_structure:c}), there exists $ g \in D_q $ such that $ D_q \setminus \{g\} = Q_c^{(2)} $. Since $ q $ accepts $ B_q $ under $ \{A_q, B_q, D_q\} $ by Definition~\ref{definition:critical_conflict}~\ref{definition:critical_conflict:iii}, this means that $ B_q \succ_q Q_c^{(2)} $. However, Lemma~\ref{lemma:three_seven_efx_or_b_unanimity}~\ref{lemma:three_seven_efx_or_b_unanimity:a}, applied to the original critical conflict, gives the opposite conclusion \[
Q_c^{(2)} \succ_q B_c = B_q
\] This is a contradiction. Hence, \eqref{eq:t_e_Aq_intersects_Qc2} holds.

Second, we claim that \begin{equation}
    \label{eq:t_e_Aq_not_Qc2}
    A_q \neq Q_c^{(2)} = \{x, z\}
\end{equation} Indeed, otherwise passing to the complements inside $ S_q = S_c $ would give $ D_q = Q_c^{(3)} $, which together with Lemma~\ref{lemma:three_seven_efx_or_b_unanimity}~\ref{lemma:three_seven_efx_or_b_unanimity:a} and Definition~\ref{definition:critical_conflict}~\ref{definition:critical_conflict:ii} yields \[
\{b\} \overset{\ref{lemma:three_seven_efx_or_b_unanimity}\ref{lemma:three_seven_efx_or_b_unanimity:a}}{\succ_p} Q_c^{(3)} = D_q \overset{\ref{definition:critical_conflict}\ref{definition:critical_conflict:ii}}{\succ_p} B_q = \{b, e\}
\] which is a contradiction. This proves \eqref{eq:t_e_Aq_not_Qc2}.

Note that by Lemma~\ref{lemma:three_seven_basic_structure}~\ref{lemma:three_seven_basic_structure:c}, $ A_c \in \mathcal{T} $. Moreover, by \eqref{eq:t_e_Cq2_transversal}, there are only three possibilities for $ C_q^{(2)} $, namely the three elements of $ \mathcal{T} $. We consider them separately.

\subsubsection{Case $ C_q^{(2)} = \{x, u\} $}
\label{section:three_seven_pivots:t_e_x_u}

Suppose that $ C_q^{(2)} = \{x, u\} $. Since $ A_c \in \mathcal{T} $, conditions \eqref{eq:t_e_Ac_intersects_Cq2} and \eqref{eq:t_e_Ac_not_Cq2} force $ A_c = \{x, v\} $. We now determine $ P_q^{(2)} $. Suppose that $ x \notin P_q^{(2)} $. Since $ P_q^{(2)} $ intersects $ C_q^{(2)} = \{x, u\} $, it follows that $ u \in P_q^{(2)} $. By \eqref{eq:t_e_Ac_crosses_Pq}, $ A_c = \{x, v\} $ intersects $ P_q^{(2)} $, so from $ x \notin P_q^{(2)} $ we obtain that $ v \in P_q^{(2)} $. Hence, $ P_q^{(2)} = \{u, v\} \subseteq Q_c^{(3)} $ which is a contradiction to \eqref{eq:t_e_Pq2_crosses_Qc}. We conclude that $ x \in P_q^{(2)} $. Now observe that $ u \notin P_q^{(2)} $, because otherwise this would contradict $ |P_q^{(2)} \cap C_q^{(2)}| = 1 $ by Lemma~\ref{lemma:three_seven_basic_structure}~\ref{lemma:three_seven_basic_structure:b}. Since by \eqref{eq:t_e_Ac_crosses_Pq} we have $ A_c \cap P_q^{(3)} \neq \varnothing $ and \[
A_c = \{x, v\} \ni x \notin P_q^{(3)}
\] we have $ v \in P_q^{(3)} $. In other words, $ v \notin P_q^{(2)} $. Note that by \eqref{eq:t_e_Pq2_crosses_Qc} we have \[
P_q^{(2)} \cap \{y, u, v\} = P_q^{(2)} \cap Q_c^{(3)} \neq \varnothing
\] Taken together with $ u, v \notin P_q^{(2)} $, this means that $ y \in P_q^{(2)} $. Hence, \[
P_q^{(2)} = \{x, y\} \qquad\text{and}\qquad P_q^{(3)} = \{z, u, v\}
\] It follows that the intersection structures of the fair views of the two critical conflicts must be related as follows: \[
\begin{array}{c|cc}
    \cap & C_q^{(2)} & C_q^{(3)} \\ \hline
    P_q^{(2)} & \{\alpha\} = \{x\} & \{\beta\} = \{y\} \\
    P_q^{(3)} & \{\gamma\} = \{u\} & \{\theta, \omega\} = \{z, v\}
\end{array}
\] Consequently, Lemma~\ref{lemma:three_seven_basic_structure}~\ref{lemma:three_seven_basic_structure:c}, applied to the critical conflict between $ p $ and $ c $, yields \[
A_q \in \bigl\{\{x, z\}, \{x, v\}, \{y, u\}\bigr\}
\] Since the pair $ \{y, u\} $ is disjoint from $ Q_c^{(2)} = \{x, z\} $, from \eqref{eq:t_e_Aq_intersects_Qc2} it follows that $ A_q \neq \{y, u\} $. The case $ A_q = \{x, z\} $ is excluded by \eqref{eq:t_e_Aq_not_Qc2}. Hence, \[
A_q = A_c = \{x, v\}
\] Taking complements inside $ S_q = S_c $ gives $ D_q = D_c $. Recall that Definition~\ref{definition:critical_conflict}~\ref{definition:critical_conflict:ii}, applied to the critical conflict between $ p $ and $ c $, provides some $ d_c \in D_q = D_c $ such that \begin{equation}
    \label{eq:t_e_Dc_without_dc_Dq_without_dc_succc_Aq_Ac}
    D_c \setminus \{d_c\} = D_q \setminus \{d_c\} \overset{\ref{definition:critical_conflict}\ref{definition:critical_conflict:ii}}{\succ_c} A_q = A_c
\end{equation} However, by Definition~\ref{definition:critical_conflict}~\ref{definition:critical_conflict:iii}, applied to the critical conflict between $ p $ and $ q $, we have that $ c $ accepts $ A_c $ under $ \{A_c, B_c, D_c\} $. In particular, $ A_c \succ_c D_c \setminus \{d_c\} $. This is a contradiction to \eqref{eq:t_e_Dc_without_dc_Dq_without_dc_succc_Aq_Ac}.

\subsubsection{Case $ C_q^{(2)} = \{x, v\} $}

Suppose that $ C_q^{(2)} = \{x, v\} $. Since the reasoning so far is invariant under swapping $ u $ and $ v $, this case leads to a contradiction by the same proof as in Section~\ref{section:three_seven_pivots:t_e_x_u} above, but with the names of $ u $ and $ v $ interchanged.

\subsubsection{Case $ C_q^{(2)} = \{y, z\} $}

Suppose that $ C_q^{(2)} = \{y, z\} $. Recall that $ A_c \in \mathcal{T} $. By \eqref{eq:t_e_Ac_not_Cq2}, $ A_c \neq \{y, z\} $, which leaves only the possibilities $ A_c = \{x, u\} $ and $ A_c = \{x, v\} $. Both of these cases are disjoint from $ C_q^{(2)} $ and hence contradict \eqref{eq:t_e_Ac_intersects_Cq2}.

Overall, since all three possibilities for $ C_q^{(2)} $ lead to a contradiction, we conclude that $ t \neq e $.

\subsection{Ruling out the case $ t = y $}
\label{section:three_seven_pivots:t_y}

Suppose that $ t = y $. Then $ B_q = \{b, y\} $ and $ S_q = \{x, z, u, v, e\} $. We first claim that \begin{equation}
    \label{eq:t_y_x_notin_Cq2}
    x \notin C_q^{(2)}
\end{equation} Towards a contradiction, suppose that $ x \in C_q^{(2)} $. Recall that $ C_q^{(2)} = \{\alpha, \gamma\} $. Since $ t = y $, the second inequality in \eqref{eq:b_succq_alphabetagammat} therefore gives $ \{b\} \succ_p \{x, y\} $. Combining this with Lemma~\ref{lemma:three_seven_efx_or_b_unanimity}~\ref{lemma:three_seven_efx_or_b_unanimity:a} yields \[
\{b\} \overset{\eqref{eq:b_succq_alphabetagammat}}{\succ_p} \{x, y\} = P_c^{(2)} \overset{\ref{lemma:three_seven_efx_or_b_unanimity}\ref{lemma:three_seven_efx_or_b_unanimity:a}}{\succ_p} \{b\}
\] which is a contradiction. This proves \eqref{eq:t_y_x_notin_Cq2}.

We now determine $ C_q^{(2)} $. First, observe that $ e \in C_q^{(2)} $. Indeed, otherwise \eqref{eq:t_y_x_notin_Cq2} and $ C_q^{(2)} \subseteq S_q $ would imply $ C_q^{(2)} \subsetneq \{z, u, v\} = P_c^{(3)} $. Hence, Lemma~\ref{lemma:three_seven_efx_or_b_unanimity}~\ref{lemma:three_seven_efx_or_b_unanimity:b} together with \eqref{eq:ts_cq_pq_relationships} would give \[
\{b\} \overset{\ref{lemma:three_seven_efx_or_b_unanimity}\ref{lemma:three_seven_efx_or_b_unanimity:b}}{\succ_c} P_c^{(3)} \succ_c C_q^{(2)} \overset{\eqref{eq:ts_cq_pq_relationships}}{\succ_c} \{b\}
\] which is a contradiction. Thus, $ e \in C_q^{(2)} $.

Since $ C_q^{(2)} $ has size two and excludes $ x $ by \eqref{eq:t_y_x_notin_Cq2}, we can therefore write $ C_q^{(2)} = \{r, e\} $ for some $ r \in \{z, u, v\} $. Next, we show that $ r \neq z $. Indeed, otherwise Lemma~\ref{lemma:three_seven_efx_or_b_unanimity}~\ref{lemma:three_seven_efx_or_b_unanimity:d} read together with \eqref{eq:ts_cq_pq_relationships} would yield \[
\{b\} \overset{\ref{lemma:three_seven_efx_or_b_unanimity}\ref{lemma:three_seven_efx_or_b_unanimity:d}}{\succ_c} \{z, e\} = C_q^{(2)} \overset{\eqref{eq:ts_cq_pq_relationships}}{\succ_c} \{b\}
\] which is a contradiction. Hence, $ r \in \{u, v\} $.

Let $ s $ denote the other good in $ \{u, v\} $, so that $ \{r, s\} = \{u, v\} $. Passing to the complement of $ C_q^{(2)} $ inside $ S_q $ yields $ C_q^{(3)} = \{x, z, s\} $. By \eqref{eq:b_succq_Pq2_Pq3_Cq2_Cq3} read in conjunction with Lemma~\ref{lemma:three_seven_efx_or_b_unanimity}~\ref{lemma:three_seven_efx_or_b_unanimity:a}, we get \[
\{b\} \overset{\eqref{eq:b_succq_Pq2_Pq3_Cq2_Cq3}}{\succ_q} C_q^{(3)} = \{x, z, s\} \succ_q \{x, z\} = Q_c^{(2)} \overset{\ref{lemma:three_seven_efx_or_b_unanimity}\ref{lemma:three_seven_efx_or_b_unanimity:a}}{\succ_q} \{b\}
\] which is a contradiction. We conclude that $ t \neq y $.

\subsection{Ruling out the case $ t = z $}
\label{section:three_seven_pivots:t_z}

Suppose that $ t = z $. Then $ B_q = \{b, z\} $ and $ S_q = \{x, y, u, v, e\} $. We first claim that \begin{equation}
    \label{eq:t_z_x_in_pair_blocks}
    x \in P_q^{(2)} \cup C_q^{(2)} = \{\alpha, \beta, \gamma\}
\end{equation} We prove \eqref{eq:t_z_x_in_pair_blocks}. Towards a contradiction, suppose that \eqref{eq:t_z_x_in_pair_blocks} were false. Then \[
\{\alpha, \gamma\} = C_q^{(2)} \subseteq \{y, u, v, e\}
\] We determine $ C_q^{(2)} $ in terms of $ y, u, v, e $ more precisely. First, observe that $ e \in C_q^{(2)} $ because otherwise we would have \[
C_q^{(2)} \subsetneq \{y, u, v\} = Q_c^{(3)}
\] and hence $ \{b\} \succ_c Q_c^{(3)} \succ_c C_q^{(2)} $ by Lemma~\ref{lemma:three_seven_efx_or_b_unanimity}~\ref{lemma:three_seven_efx_or_b_unanimity:b} and monotonicity, which contradicts \eqref{eq:ts_cq_pq_relationships}. Next, note that $ C_q^{(2)} \neq \{y, e\} $. This is because otherwise by \eqref{eq:ts_cq_pq_relationships} and Lemma~\ref{lemma:three_seven_efx_or_b_unanimity}~\ref{lemma:three_seven_efx_or_b_unanimity:d} applied to the critical conflict between $ p $ and $ q $ we would have \[
\{b\} \overset{\ref{lemma:three_seven_efx_or_b_unanimity}\ref{lemma:three_seven_efx_or_b_unanimity:d}}{\succ_c} \{y, e\} = C_q^{(2)} \overset{\eqref{eq:ts_cq_pq_relationships}}{\succ_c} \{b\}
\] which is a contradiction. It follows that \[
C_q^{(2)} = \{r, e\}
\] for some $ r \in \{u, v\} $. Let $ s $ denote the other good in $ \{u, v\} $, so that $ \{r, s\} = \{u, v\} $. Since $ \fancyC_q $ is a binary partition of the conflict region $ S_q = \{x, y, u, v, e\} $, it follows that \[
C_q^{(3)} = S_q \setminus C_q^{(2)} = \{x, y, s\}
\] Hence, by \eqref{eq:ts_cq_pq_relationships} and Lemma~\ref{lemma:three_seven_efx_or_b_unanimity}~\ref{lemma:three_seven_efx_or_b_unanimity:a} we obtain \[
\{b\} \overset{\eqref{eq:ts_cq_pq_relationships}}{\succ_p} C_q^{(3)} = \{x, y, s\} \succ_p \{x, y\} = P_c^{(2)} \overset{\ref{lemma:three_seven_efx_or_b_unanimity}\ref{lemma:three_seven_efx_or_b_unanimity:a}}{\succ_p} \{b\}
\] which is a contradiction. This proves \eqref{eq:t_z_x_in_pair_blocks}.

Since $ t = z $, from the first inequality of \eqref{eq:b_succq_alphabetagammat} read in conjunction with \eqref{eq:t_z_x_in_pair_blocks} it follows that $ \{b\} \succ_q \{x, z\} $. Combining this with Lemma~\ref{lemma:three_seven_efx_or_b_unanimity}~\ref{lemma:three_seven_efx_or_b_unanimity:a} yields \[
\{b\} \succ_q \{x, z\} = Q_c^{(2)} \overset{\ref{lemma:three_seven_efx_or_b_unanimity}\ref{lemma:three_seven_efx_or_b_unanimity:a}}{\succ_q} \{b\}
\] which is a contradiction. We conclude that $ t \neq z $.

\subsection{Ruling out the case $ t = u $}
\label{section:three_seven_pivots:t_u}

Suppose that $ t = u $. Then $ B_q = \{b, u\} $ and $ S_q = \{x, y, z, v, e\} $. We first determine the possible size-two block $ C_q^{(2)} $. By \eqref{eq:ts_cq_pq_relationships} and Lemma~\ref{lemma:three_seven_efx_or_b_unanimity}~\ref{lemma:three_seven_efx_or_b_unanimity:b}, we have \[
C_q^{(2)} \overset{\eqref{eq:ts_cq_pq_relationships}}{\succ_c} \{b\} \overset{\ref{lemma:three_seven_efx_or_b_unanimity}\ref{lemma:three_seven_efx_or_b_unanimity:b}}{\succ_c} P_c^{(2)}, P_c^{(3)}, Q_c^{(2)}, Q_c^{(3)}
\] Due to monotonicity, this means that $ C_q^{(2)} $ cannot be contained in the fair-view blocks of the critical conflict between $ p $ and $ q $. In other words, we have \begin{equation}
    \label{eq:t_u_Cq2_notin_xy_xz_yuv_zuv}
    C_q^{(2)} \nsubseteq X \qquad\forall X \in \bigl\{\{x, y\}, \{x, z\}, \{y, u, v\}, \{z, u, v\}\bigr\}
\end{equation} Similarly, by \eqref{eq:ts_cq_pq_relationships} and Lemma~\ref{lemma:three_seven_efx_or_b_unanimity}~\ref{lemma:three_seven_efx_or_b_unanimity:d}, we have \[
C_q^{(2)} \overset{\eqref{eq:ts_cq_pq_relationships}}{\succ_c} \{b\} \succ_c \{x, e\}, \{y, e\}, \{z, e\}
\] meaning that \begin{equation}
    \label{eq:t_u_Cq2_notin_xe_ye_ze}
    C_q^{(2)} \notin \bigl\{\{x, e\}, \{y, e\}, \{z, e\}\bigr\}
\end{equation} Now observe that \eqref{eq:t_u_Cq2_notin_xy_xz_yuv_zuv} and \eqref{eq:t_u_Cq2_notin_xe_ye_ze} leave only the following possibilities for $ C_q^{(2)} $: \[
C_q^{(2)} \in \bigl\{\{x, v\}, \{y, z\}, \{v, e\}\bigr\}
\] We handle these three cases separately:

\subsubsection{Case $ C_q^{(2)} = \{v, e\} $}

Suppose first that $ C_q^{(2)} = \{v, e\} $. Then $ C_q^{(3)} = \{x, y, z\} $, meaning that by \eqref{eq:ts_cq_pq_relationships} and Lemma~\ref{lemma:three_seven_efx_or_b_unanimity}~\ref{lemma:three_seven_efx_or_b_unanimity:a}, we get \[
\{b\} \overset{\eqref{eq:ts_cq_pq_relationships}}{\succ_p} C_q^{(3)} = \{x, y, z\} \succ_p \{x, y\} = P_c^{(2)} \overset{\ref{lemma:three_seven_efx_or_b_unanimity}\ref{lemma:three_seven_efx_or_b_unanimity:a}}{\succ_p} \{b\} 
\] which is a contradiction. Hence, this case is impossible.

\subsubsection{Case $ C_q^{(2)} = \{x, v\} $}
\label{section:three_seven_pivots:t_u_x_v}

Now suppose that $ C_q^{(2)} = \{x, v\} $. We determine $ P_q^{(2)} $. Since $ |C_q^{(2)} \cap P_q^{(2)}| = 1 $, only the following pairs are possible: \begin{equation}
    \label{eq:t_u_case2_subcases}
    P_q^{(2)} \in \bigl\{\{x, y\}, \{x, z\}, \{x, e\}, \{y, v\}, \{z, v\}, \{v, e\}\bigr\}
\end{equation} Note that by \eqref{eq:ts_cq_pq_relationships} and Lemma~\ref{lemma:three_seven_efx_or_b_unanimity}~\ref{lemma:three_seven_efx_or_b_unanimity:a} we have \[
P_q^{(2)} \overset{\eqref{eq:ts_cq_pq_relationships}}{\succ_p} \{b\} \overset{\ref{lemma:three_seven_efx_or_b_unanimity}\ref{lemma:three_seven_efx_or_b_unanimity:a}}{\succ_p} \{x, z\}, \{y, u, v\}
\] meaning that $ P_q^{(2)} \notin \{\{x, z\}, \{y, v\}\} $. Similarly, applying \eqref{eq:ts_cq_pq_relationships} and Lemma~\ref{lemma:three_seven_efx_or_b_unanimity}~\ref{lemma:three_seven_efx_or_b_unanimity:d} yields \[
P_q^{(2)} \overset{\eqref{eq:ts_cq_pq_relationships}}{\succ_p} \{b\} \overset{\ref{lemma:three_seven_efx_or_b_unanimity}\ref{lemma:three_seven_efx_or_b_unanimity:d}}{\succ_p} \{x, e\}
\] and hence $ P_q^{(2)} \neq \{x, e\} $.

Turning to the case $ P_q^{(2)} = \{z, v\} $, observe that then from $ \fancyP_q = \{\{z, v\}, \{x, y, e\}\} $ being EFX with respect to $ \prec_p $ it follows that $ \{z, v\} \succ_p \{x, y\} $. However, since $ \fancyP_c = \{\{x, y\}, \{z, u, v\}\} $ is EFX with respect to $ \prec_p $, we have the opposite inequality $ \{x, y\} \succ_p \{z, v\} $. This contradiction shows that $ P_q^{(2)} \neq \{z, v\} $.

Now consider the case $ P_q^{(2)} = \{v, e\} $. Then $ P_q^{(3)} = \{x, y, z\} $. Applying Lemma~\ref{lemma:three_seven_efx_or_b_unanimity}~\ref{lemma:three_seven_efx_or_b_unanimity:a} and \eqref{eq:b_succq_Pq2_Pq3_Cq2_Cq3}, we get \[
P_q^{(3)} = \{x, y, z\} \succ_q \{x, z\} = Q_c^{(2)} \overset{\ref{lemma:three_seven_efx_or_b_unanimity}\ref{lemma:three_seven_efx_or_b_unanimity:a}}{\succ_q} \{b\} \overset{\eqref{eq:b_succq_Pq2_Pq3_Cq2_Cq3}}{\succ_q} P_q^{(3)}
\] which is a contradiction. Hence, $ P_q^{(2)} \neq \{v, e\} $.

We have ruled out all possibilities of \eqref{eq:t_u_case2_subcases} except $ P_q^{(2)} = \{x, y\} $. Thus, in the following we can assume that $ P_q^{(2)} = \{x, y\} $. Since \[
\{\alpha\} = P_q^{(2)} \cap C_q^{(2)} = \{x, y\} \cap \{x, v\} = \{x\}
\] we have $ \alpha = x $. Hence, by the third inequality in \eqref{eq:b_succq_alphabetagammat}, \begin{equation}
    \label{eq:t_u_case2_b_succc_xu}
    \{b\} \succ_c \{\alpha, t\} = \{x, u\}
\end{equation} Recall that by Definition~\ref{definition:critical_conflict}~\ref{definition:critical_conflict:iii}, $ c $ accepts $ A_c $ under $ \{A_c, B_c, D_c\} $, so in particular we have $ A_c \succ_c B_c \setminus \{e\} = \{b\} $. Combined with \eqref{eq:t_u_case2_b_succc_xu}, this means that $ A_c \neq \{x, u\} $. It follows that, by Lemma~\ref{lemma:three_seven_basic_structure}~\ref{lemma:three_seven_basic_structure:c}, only the following two possibilities remain for $ A_c $: \begin{equation}
    \label{eq:t_u_case2_Ac_possibilities}
    A_c \in \bigl\{\{x, v\}, \{y, z\}\bigr\}
\end{equation} Suppose first that $ A_c = \{y, z\} $. Then $ D_c = \{x, u, v\} $. Since $ c $ accepts $ A_c $ under $ \{A_c, B_c, D_c\} $, \[
\{y, z\} = A_c \succ_c D_c \setminus \{u\} = \{x, v\}
\] However, from $ \fancyC_q = \{\{x, v\}, \{y, z, e\}\} $ being EFX with respect to $ \prec_c $ the opposite relationship follows: \[
\{x, v\} = C_q^{(2)} \succ_c C_q^{(3)} \setminus \{e\} = \{y, z\}
\] This is a contradiction. Hence, by \eqref{eq:t_u_case2_Ac_possibilities}, it remains to consider the case $ A_c = \{x, v\} $ and $ D_c = \{y, z, u\} $. We first make a few preparatory observations. Our goal is to show that $ q $ regards $ b $ as more valuable than any one-good deletion remainder of $ D_c $. Note that $ \beta = y $ because \[
\{\beta\} = P_q^{(2)} \cap C_q^{(3)} = \{x, y\} \cap \{y, z, e\} = \{y\}
\] Hence, by the first inequality in \eqref{eq:b_succq_alphabetagammat}, \begin{equation}
    \label{eq:t_u_case2_b_succq_yu}
    \{b\} \succ_q \{\beta, t\} = \{y, u\}
\end{equation} Next, observe that by \eqref{eq:b_succq_Pq2_Pq3_Cq2_Cq3}, \begin{equation}
    \label{eq:t_u_case2_b_succq_yz}
    \{b\} \overset{\eqref{eq:b_succq_Pq2_Pq3_Cq2_Cq3}}{\succ_q} C_q^{(3)} = \{y, z, e\} \succ_q \{y, z\}
\end{equation} We also note that by Lemma~\ref{lemma:three_seven_efx_or_b_unanimity}~\ref{lemma:three_seven_efx_or_b_unanimity:a}, \begin{equation}
    \label{eq:t_u_case2_b_succq_zu}
    \{b\} \overset{\ref{lemma:three_seven_efx_or_b_unanimity}\ref{lemma:three_seven_efx_or_b_unanimity:a}}{\succ_q} P_c^{(3)} = \{z, u, v\} \succ_q \{z, u\}
\end{equation} Since $ \{y, u\} $, $ \{y, z\} $ and $ \{z, u\} $ are exactly the sets that can be obtained from $ D_c $ by deleting one good, from \eqref{eq:t_u_case2_b_succq_yu}, \eqref{eq:t_u_case2_b_succq_yz} and \eqref{eq:t_u_case2_b_succq_zu} it follows that \[
D_c \setminus \{d\} \prec_q \{b\} \prec_q B_c
\] for all $ d \in D_c $. This contradicts Definition~\ref{definition:critical_conflict}~\ref{definition:critical_conflict:ii} for agent $ q $. We conclude that $ C_q^{(2)} \neq \{x, v\} $.

\subsubsection{Case $ C_q^{(2)} = \{y, z\} $}

Suppose that $ C_q^{(2)} = \{y, z\} $. The argument is similar to that of Section~\ref{section:three_seven_pivots:t_u_x_v} above. We again begin by determining $ P_q^{(2)} $. Since $ |C_q^{(2)} \cap P_q^{(2)}| = 1 $, the possible pairs are
\begin{equation}
    \label{eq:t_u_case3_subcases}
    P_q^{(2)} \in \bigl\{\{x, y\}, \{y, v\}, \{y, e\}, \{x, z\}, \{z, v\}, \{z, e\}\bigr\}
\end{equation} By \eqref{eq:ts_cq_pq_relationships} and Lemma~\ref{lemma:three_seven_efx_or_b_unanimity}~\ref{lemma:three_seven_efx_or_b_unanimity:a}, we have \[
P_q^{(2)} \overset{\eqref{eq:ts_cq_pq_relationships}}{\succ_p} \{b\} \overset{\ref{lemma:three_seven_efx_or_b_unanimity}\ref{lemma:three_seven_efx_or_b_unanimity:a}}{\succ_p} \{x, z\}, \{y, u, v\}
\] Consequently, $ P_q^{(2)} \neq \{x, z\} $ and $ P_q^{(2)} \neq \{y, v\} $. In a similar way, \eqref{eq:ts_cq_pq_relationships} and Lemma~\ref{lemma:three_seven_efx_or_b_unanimity}~\ref{lemma:three_seven_efx_or_b_unanimity:d} give \[
P_q^{(2)} \overset{\eqref{eq:ts_cq_pq_relationships}}{\succ_p} \{b\} \overset{\ref{lemma:three_seven_efx_or_b_unanimity}\ref{lemma:three_seven_efx_or_b_unanimity:d}}{\succ_p} \{z, e\}
\] which means that $ P_q^{(2)} \neq \{z, e\} $.

We now consider the case $ P_q^{(2)} = \{z, v\} $. As in Section~\ref{section:three_seven_pivots:t_u_x_v} above, from $ \fancyP_q = \{\{z, v\}, \{x, y, e\}\} $ being EFX with respect to $ \prec_p $ it follows that $ \{z, v\} \succ_p \{x, y\} $. On the other hand, due to $ \fancyP_c = \{\{x, y\}, \{z, u, v\}\} $ being EFX with respect to $ \prec_p $, we have the opposite inequality $ \{x, y\} \succ_p \{z, v\} $. This is a contradiction, so we conclude that $ P_q^{(2)} \neq \{z, v\} $.

We now study the case $ P_q^{(2)} = \{y, e\} $. Then $ P_q^{(3)} = \{x, z, v\} $. Applying Lemma~\ref{lemma:three_seven_efx_or_b_unanimity}~\ref{lemma:three_seven_efx_or_b_unanimity:a} and \eqref{eq:b_succq_Pq2_Pq3_Cq2_Cq3}, we get \[
P_q^{(3)} = \{x, z, v\} \succ_q \{x, z\} = Q_c^{(2)} \overset{\ref{lemma:three_seven_efx_or_b_unanimity}\ref{lemma:three_seven_efx_or_b_unanimity:a}}{\succ_q} \{b\} \overset{\eqref{eq:b_succq_Pq2_Pq3_Cq2_Cq3}}{\succ_q} P_q^{(3)}
\] which is a contradiction. Hence, $ P_q^{(2)} \neq \{y, e\} $.

We have excluded every possibility in \eqref{eq:t_u_case3_subcases} except $ P_q^{(2)} = \{x, y\} $. Thus, in the following we can assume \[
P_q^{(2)} = \{x, y\} \qquad\text{and}\qquad P_q^{(3)} = \{z, v, e\}
\] As in Section~\ref{section:three_seven_pivots:t_u_x_v} above, we start by determining the following relevant intersection of the two fair views in the critical conflict between $ p $ and $ c $. Note that $ C_q^{(3)} = \{x, v, e\} $. Since \[
\{\beta\} = P_q^{(2)} \cap C_q^{(3)} = \{x, y\} \cap \{x, v, e\} = \{x\}
\] we have $ \beta = x $. Since $ t = u $, the third inequality in \eqref{eq:b_succq_alphabetagammat} gives
\begin{equation}
    \label{eq:t_u_case3_b_succc_xu}
    \{b\} \succ_c \{\beta, t\} = \{x, u\}
\end{equation} Since $ c $ accepts $ A_c $ under $ \{A_c, B_c, D_c\} $, we have $ A_c \succ_c B_c \setminus \{e\} = \{b\} $. Together with \eqref{eq:t_u_case3_b_succc_xu}, this implies $ A_c \succ_c \{x, u\} $ and hence $ A_c \neq \{x, u\} $. Lemma~\ref{lemma:three_seven_basic_structure}~\ref{lemma:three_seven_basic_structure:c} therefore leaves only the possibilities
\begin{equation}
    \label{eq:t_u_case3_Ac_possibilities}
    A_c \in \bigl\{\{x, v\}, \{y, z\}\bigr\}
\end{equation} Suppose first that $ A_c = \{x, v\} $. Then $ D_c = \{y, z, u\} $. Since $ c $ accepts $ A_c $ under $ \{A_c, B_c, D_c\} $, \[
\{x, v\} = A_c \succ_c D_c \setminus \{u\} = \{y, z\}
\] However, $ \fancyC_q = \{\{y, z\}, \{x, v, e\}\} $ is EFX with respect to $ \prec_c $, and we obtain the opposite relationship \[
\{y, z\} = C_q^{(2)} \succ_c C_q^{(3)} \setminus \{e\} = \{x, v\}
\] This is a contradiction. Thus, by \eqref{eq:t_u_case3_Ac_possibilities}, we must have \[
A_c = \{y, z\} \qquad\text{and}\qquad D_c = \{x, u, v\}
\] To complete the proof, we show that no one-good deletion remainder of $ D_c $ can satisfy Definition~\ref{definition:critical_conflict}~\ref{definition:critical_conflict:ii} for agent $ q $, as in Section~\ref{section:three_seven_pivots:t_u_x_v} above. First, since $ \beta = x $ and $ t = u $, the first inequality in \eqref{eq:b_succq_alphabetagammat} gives \begin{equation}
    \label{eq:t_u_case3_b_succq_xu}
    \{b\} \succ_q \{\beta, t\} = \{x, u\}
\end{equation} Moreover, by \eqref{eq:b_succq_Pq2_Pq3_Cq2_Cq3} we have \begin{equation}
    \label{eq:t_u_case3_b_succq_xv}
    \{b\} \overset{\eqref{eq:b_succq_Pq2_Pq3_Cq2_Cq3}}{\succ_q} C_q^{(3)} = \{x, v, e\} \succ_q \{x, v\}
\end{equation} Finally, from Lemma~\ref{lemma:three_seven_efx_or_b_unanimity}~\ref{lemma:three_seven_efx_or_b_unanimity:a} it follows that \begin{equation}
    \label{eq:t_u_case3_b_succq_uv}
    \{b\} \overset{\ref{lemma:three_seven_efx_or_b_unanimity}\ref{lemma:three_seven_efx_or_b_unanimity:a}}{\succ_q} P_c^{(3)} = \{z, u, v\} \succ_q \{u, v\}
\end{equation} The bundles $ \{x, u\} $, $ \{x, v\} $ and $ \{u, v\} $ are exactly the one-good deletion remainders of $ D_c = \{x, u, v\} $. Hence, it follows from \eqref{eq:t_u_case3_b_succq_xu}, \eqref{eq:t_u_case3_b_succq_xv} and \eqref{eq:t_u_case3_b_succq_uv} that \[
D_c \setminus \{d\} \prec_q \{b\} \prec_q B_c
\] holds for every $ d \in D_c $. This contradicts Definition~\ref{definition:critical_conflict}~\ref{definition:critical_conflict:ii} for agent $ q $. Hence, $ C_q^{(2)} \neq \{y, z\} $.

We have now ruled out all three possible values of $ C_q^{(2)} $, and therefore conclude that $ t \neq u $.

\subsection{Ruling out the case $ t = v $}
\label{section:three_seven_pivots:t_v}

Suppose that $ t = v $. A contradiction follows by the same proof as in Section~\ref{section:three_seven_pivots:t_u} above, but with the names of $ u $ and $ v $ swapped. Indeed, the above assumptions and reasoning are invariant under interchanging $ u $ and $ v $. Hence, $ t \neq v $.

We conclude that $ t \notin \{e, y, z, u, v\} $ and hence $ t = x $, as desired. This completes the proof of Lemma~\ref{lemma:three_seven_pivots}.

\newcommand{\abs}[1]{|#1|}
\newcommand{\num}{\mathit{num}}

\newcommand{\EFXthree}{\mathit{EFX3}}
\newcommand{\SL}{\mathit{SL}}
\newcommand{\tv}{\tilde{v}}
\newcommand{\tA}{\tilde{A}} \newcommand{\tB}{\tilde{B}} \newcommand{\tC}{\tilde{C}}
\newcommand{\bv}{\bar{v}}
\renewcommand{\index}{\mathit{index}}
\newcommand{\var}{\mathit{var}}
\newcommand{\Gap}{\operatorname{Gap}}
\newcommand{\Monotone}{\mathit{Monotone}}
\newcommand{\Monotonicity}{\mathit{Monotonicity}}
\newcommand{\Transitive}{\mathit{Transitive}}
\newcommand{\Transitivity}{\mathit{Transitivity}}
\newcommand{\IfSingletonSomeLightSingletonElseTrue}
{\mathit{If\_Singleton\_Some\_light\_Singleton\_Else\_True}}
\newcommand{\allocations}{\mathit{allocation}}
\newcommand{\Dmin}{D_{\operatorname{min}}}
\providecommand{\R}{\mathbb{R}}
\newcommand{\assign}{\mathbin{\raisebox{0.05ex}{\mbox{\rm :}}\!\!=}}

\newlength{\mysetspacing}
\setlength{\mysetspacing}{0.04em}
\providecommand{\sset}[1]{\{\hspace{\mysetspacing} #1 \hspace{\mysetspacing}\}}

\newcommand{\SPASS}{{\sc Spass}}

\newcommand{\KMtwo}[2]{{\tiny #1}{\textcolor{red}{#2}}}

\newcommand{\hac}[1]{\textcolor{blue}{Hana says: #1}}
\newcommand{\hae}[1]{\textcolor{red}{#1}}
\newcommand{\shreyas}[1]{\textcolor{green!60!black}{Shreyas : #1}}
\newcommand{\instance}{(N,M,\{v_i\}_{i \in N})}

\section{SAT-assisted discovery and verification}\label{sec:old-sat-approach}

As noted in \Cref{sec:introduction}, we initially proved \Cref{thm:intro-counterexample} and \Cref{thm:intro-seven} using computer-assisted methods and tools. In fact, the resolution of the EFX problem through a counterexample was a surprising result that was unexpectedly produced by a SAT solver. Subsequently, we used the Z3 SMT solver of \cite{Moura:TACAS08-337} to run experiments and identify combinatorial patterns in the computer-generated counterexample. Relying on this analysis, we were able to derive a self-contained and considerably simpler counterexample in a classical, pen-and-paper way (see \Cref{sec:conceptual-proof}).

While the exposition in the main body of this paper yields insights into fair division theory, we believe that the surprising resolution of a long-standing open combinatorial question by an application of SAT solvers is of general and independent interest to the wider theoretical CS community. Our experience and follow up work such as that of \cite{efxormms2026} and \cite{alkassar2026completeefxallocationsexist} suggests that it is possible that other open problems are amenable to similar techniques employing SAT and SMT solvers and automated theorem provers. In this appendix, we outline our original SAT-solver based attack on the EFX problem, in order to provide enough intuition and methodological detail to reproduce our results and inspire similar tool-driven approaches for other combinatorial problems.

\paragraph*{Tools Used} 
For the sake of completeness and convenience, we briefly describe the tools we used.
\begin{itemize}
    \item SAT solvers are tools that take as input a propositional formula and check whether or not it is satisfiable, i.e., the formula admits an assignment of variables that makes the formula true. If the formula is satisfiable, SAT-solvers produce a concrete example of an assignment to the variables which certify this. Otherwise a proof of unsatisfiability is produced, usually via resolution calculus. We refer the reader to \cite{KroeningStrichmanDecisionProcedures, BradleyMannaCalculusOfComputation} for more background. In this work we used the SPASS-SAT solver of \cite{BrombergerEtAl19} and cross-verified our results with the SAT solver CaDiCaL (\cite{BiereFFFFP24}).
    \item The DRAT-trim checker of \cite{WetzlerHH14} is a tool for checking that a claimed unsatisfiability proof produced by a sat solver is indeed correct.
    \item The LEAN~4 theorem prover of \cite{lean4ITP} is an interactive theorem prover that we use to verify the correctness of our SAT formula, i.e., that it indeed encodes the desired statement. Such theorem provers have a small trusted logical kernel checking proofs written in a much larger programming language that allows encoding sophisticated mathematical theories.
    \item Additionally we use C++ and Python programs to generate and verify satisfying assignments to our SAT formula.
\end{itemize}

\paragraph*{Overview of our encoding.} We encode the negation of EFX existence as a SAT formula. Its satisfiability yields a counterexample, namely, a valuation function for which no EFX allocation exists. Its unsatisfiability establishes universal existence of EFX, for all valuation functions on the chosen number of agents and goods. A direct encoding of the three-agent, seven-good case as a satisfiability problem yields a CNF formula with approximately 25{,}000 variables and more than 6 million clauses. Theoretical refinements reduce the formula for three agents and seven goods to approximately 12{,}000 variables and 2.5 million clauses. Applying standard SAT preprocessing techniques, in particular unit propagation and subsumption elimination, further reduces the number of clauses to approximately 700{,}000. 

\paragraph*{Computer-assisted results and their verification} 
\begin{itemize}
    \item For the case of 3 agents and 7 goods, we established that an EFX allocation exists for any monotone valuation profile. The proof of unsatisfiability of the SAT encoding was generating using SPASS-SAT running for approximately 30 hours. The correctness of the proof was verified using DRAT-trim and CaDiCal.
    \item For the case of 3 agents and 8 goods, SPASS-SAT found a satisfying assignment in approximately 20 hours, which we translated into a valuation profile that admits no EFX allocation. We verified the correctness of the profile using three independent programs that enumerated all possible allocations.  %
\end{itemize}

SAT formulations have to precisely reflect every single detail to be correct and can be very fragile in this respect. This is partly because they lack any inherent semantic meaning, and partly because they are large. The slightest change to them can alter the results and what they imply about the problem being solved. For a problem such as EFX existence, it therefore becomes particularly important to have a reliable way of ensuring the correctness of the SAT encoding. To achieve this, we formalized in LEAN the EFX existence question together with the relevant fair-division theory used by our encoding. We note that this is not principally an autoformalization, although some later clean up and extensions involved them. To the best of our knowledge, this is the first encoding of any part of envy-free fair division theory in LEAN. Then, we proved in LEAN the correctness of our formulation.

The remainder of this appendix is structured as follows. We start by describing some theoretical foundations which prove useful in optimizing the size of our SAT encoding of the EFX problem. We then describe the SAT encoding itself and briefly outline the structure of the formal lean proof of the aforementioned theoretical optimizations, the lowering of the EFX problem into the SAT formula, and the equivalence between the satisfiability (resp., unsatisfiabiltiy) of the SAT formula and the existence (resp., non-existence) of a counterexample. Finally, the experimental artefacts can be found in Section~\ref{subsec:Companion Material}. We defer all proof details to the lean formalization.

\subsection{Theoretical Optimizations for the SAT Encoding}
\label{sec:theoefx}

In general, $m$ items can be allotted to $n$ agents in $n^m$ ways. For $n = 3$ and $m = 7$ we have $3^7 = 2187$ distinct allocations, a manageable number.
Actually, the number is less since all bundles must be non-empty in an EFX allocation, see Section~\ref{allocations}. However, the number of valuations is enormous.
Note that a set of size seven has $\binom{7}{3} = 35$ subsets of three elements and that the values of these subsets may have any order. So the number of valuations is at least $35! \approx (35/e)^{35} \ge 10^{35}$.

\subsection{Additional constraints}

We add additional constraints aimed at restricting the search space and preventing the solver from exploring mathematically inconsistent cases. This technique contributes to significantly speeding up the running time of the solver; see Section~\ref{Experiments}.

\subsubsection{Nondegeneracy and bundle-nonemptieness constraints}

We recall Proposition~\ref{prop:nondegenerate_perturbation}, restated below:

\nondegperturb*

It follows that for the purposes of establishing EFX it is sound to enforce nondegeneracy of the valuation profile, i.e., bundles to have pairwise distinct valuations by each agent. Accordingly, hereinafter we assume that our instance is nondegenerate. This allows us to also include the following simple observation into our SAT encoding:

\begin{observation}
    In an EFX allocation every agent must receive a non-empty bundle, provided that the valuation profile is nondegenerate and $ m \ge n $.
\end{observation}
\begin{proof}
    Assume that there is an agent that receives an empty bundle. Let $ v $ be its valuation. Then there is another agent that receives a bundle $ B $ containing at least two items. Then for any $ g \in B $, we have $v(B \setminus g) > 0 = v(\emptyset)$, contradicting EFX.
\end{proof}

It follows that any agent receives at most four items for $m=6$, five items for $m=7$, and six items for $m=8$.

\subsubsection{Fixing an Item Order for Agent Zero}\label{item order} 

Next, we note that the EFX property is invariant under renaming agents and goods. By renaming goods, we can postulate that the goods are sorted according to agent 0's valuation:

\begin{observation} 
    Every valuation profile admits a renaming of goods so that the valuation $ v_0 $ of agent $ 0 $ satisfies \[
        v_0(\{0\}) < v_0(\{1\}) < \ldots < v_0(\{m-1\})
    \] The (non-)existence of EFX is preserved under this renaming.
\end{observation}

\subsubsection{Leveled Valuations}\label{leveled valuations}

Recall that in a nondegenerate instance with seven items, no EFX allocation assigns more than five items to any agent.
In other words, the EFX-conditions never involve sets of size six or seven. It therefore seems plausible that the valuations of the sets of size six and seven are irrelevant.
In fact, we may even assume that all sets of size five are more valuable than any set of smaller cardinality. This holds since in a comparison $v(A \setminus g) < v(B)$,
the set on the left has size at most four. The set on the right might have size five, but then the set on the left is the empty
set\footnote{In a partition $(A,B,C)$ with $\abs{B} = 5$, we must have $\abs{A} = \abs{C} = 1$ for $m=7$.} and the inequality holds by monotonicity.

In general, for an integer $k$, let $v$ be an arbitrary monotone valuation, and define \begin{equation*}
    \bar{v}(A) = \begin{cases} v(A)   & \text{if $\abs{A} < k$} \\
     T (\abs{A} - k + 1) + v(A) & \text{if $\abs{A} \ge k$} \end{cases},
\end{equation*} where $T$ is larger than $v(A)$ for any $A$. Then $\bv(A) < \bv(B)$ whenever $\abs{A} < \abs{B}$ and $\abs{B} \ge k$. Also $\bv$ is monotone. Indeed, assume $A \subset B$. Then $\abs{A} < \abs{B}$. If $\abs{B} < k$, monotonicity follows from the monotonicity of $v$, if $\abs{B} \ge k$, monotonicity holds by construction. Moreover, if $v$ is non-degenerate, then $\bv$ is non-degenerate. We call $\bv$ a \emph{leveled valuation} for levels greater or equal to $k$.

\begin{lemma} Let $k = m - 2$. If $(A,B,C)$ is an EFX allocation for $\bv_0$, $\bv_1$, and $\bv_2$, it is an EFX allocation for $v_0$, $v_1$, and $v_2$.  \end{lemma}
\begin{proof} Let $(A, B, C)$ be an EFX allocation for the valuations $\bv_0$, $\bv_1$ and $\bv_2$. The largest cardinality of any set in an EFX allocation is $m - 2$. Consider agent 0. We need to show $v_0(B \setminus g) \le v_0(A)$ for any $g \in B$. 
We have $\bv_0(B \setminus g) \le \bv_0(A)$, and $B \setminus g$ has cardinality at most $m-3$. Thus $v_0(B \setminus g) = \bv_0(B \setminus g)$. If $A$ has cardinality at most $m-3$, $v_0(A) = \bv_0(A)$, and, hence,  $v_0(B \setminus g) = \bv_0(B \setminus g) \le \bv_0(A) = v_0(A)$. If $A$ has cardinality $m-2$, $B$ has
  cardinality one, and, hence, $B \setminus g = \emptyset$. Then the inequality follows from
  monotonicity.
\end{proof}

\subsection{Encoding EFX into SAT} \label{sec:encoding}

EFX is complicated because it involves a quantifier alternation.
\[ \forall \,\text{valuations}\, v \ \exists \,\text{allocation}\, X \ \forall \,\text{agents} \,i,j \ \forall \,\text{goods} \,g\in X_j \;:\; v_i(X_j \setminus g) < v_i(X_i). \] If the number $m$ of goods as well as the number $n$ of agents is fixed, the above formula can actually be turned into a finite Boolean combination. We will encode its negation
\[ \exists \,\text{valuation}\, v \ \forall \,\text{allocations}\, X \ \exists \,\text{agents} \,i,j \ \exists \,\text{good} \,g\in X_j \;:\;\neg (v_i(X_j \setminus g) < v_i(X_i)). \] and check for satisfiability. The number of potential allocations is rather small compared to the number of potential valuations; see Section~\ref{sec:theoefx}. So we decided on a SAT-formulation where valuations correspond to assignments to propositional variables, and the remaining part of the formula is expressed by Boolean constraints over these variables. 

\subsubsection{The Variables}\label{sec:variables}
Let $v$ be one of the three valuations. Number the subsets of $[m]$ from $0$ to $P - 1$, where $P = 2^m$. We identify the sets with bitstrings of length $m$ and let the number of a set be the binary value of the bitstring.
This numbering obeys monotonicity, i.e., if $A \subset B$, then $A < B$. We use the same symbol for the set and its number.

For each agent $i$ and pair $A$ and $B$ of subsets with $A < B$, we introduce a variable\footnote{We now use superscripts for the agents as we will need the subscripts for the sets.}   $x^i_{AB}$. The variable is true iff $v_i(A) < v_i(B)$ and is false if $v_i(A) > v_i(B)$. The number of variables is $P(P-1)/2$ for each agent.
For $m=7$, the number is $2^7 (2^7 - 1)/2 = 2^{13} - 2^6 = 8128$. Recall that for the actual encoding we only use the variables $x^i_{AB}$ where $A<B$.
In order to improve readability and to get rid of a case analysis, if in the below formalization a variable $x^i_{AB}$ occurs with $A>B$ it is
meant to be replaced by $\neg x^i_{BA}$. This is correct, because variables
$x^i_{AA}$ are not generated at all since we only consider partitions of the goods and only non-degenerate valuations.%

A standard format for specifying CNF-formulas is the DIMACS format. In this format, variables are represented by positive integers. If a variable is negated, the negation of its number is used. We first encoded the variable $x^i_{AB}$ with $i \in \sset{0,1,2}$ by
the number $i \cdot 2^{2m} + A \cdot 2^m + B$. Experiments showed that a dense encoding works better. 

Consider a pair $(A,B)$ with $0 \le A < B < P = 2^m$. In the lexicographic ordering $(0,1), (0,2), \ldots (0, P-1), (1,2), \ldots$ of all pairs, the pair $(A,B)$ is the $\ell$-th pair, where
$\ell = P(P-1)/2 - (P-A)(P- A - 1)/2  + (B - A) = PA - A(A+1)/2 + B - A$. Namely, there are $P(P-1)/2$ pairs altogether, and $(P-A)(P-A -1)/2$ pairs $(i,k)$ where $i > A$, and  $(A,B)$ is the $(B - A)$-th pair with first index $A$.
Alternatively, out of the $PA$ pairs $(i,j)$ with  $i < A$ and $0 \le j < P$, there are $\sum_{0 \le i < A} (i + 1) = A(A+1)/2$ which have $j \le i$. The numbering of the pairs starts at one. Let $\index(A,B)$ be the index of the pair $(A,B)$ in this numbering.
Then, for $0 \le i \le 2$ and $0 \le A < B < P$, the variable
\[ x^i_{AB} \quad\text{corresponds to}\quad i \cdot P (P-1)/2 + \index(A,B). \]
\noindent
In this way, variable names are then integers between $1$ and $3 P(P -1)/2$. For $m=7$, we get $3 \cdot 2^6 \cdot (2^7 -1) = 24384$ variables.

\subsubsection{The Valuations}

We need to encode monotonicity and transitivity of the valuations $v_i$.

\begin{description}
  \item[Monotonicity:]  $\bigwedge_{i \in \sset{0,1,2}} \bigwedge_{A \subset B} x^i_{AB}$.
  \item[Transitivity:] $\bigwedge_i \bigwedge_{A,B,C}  (x^i_{AB} \wedge x^i_{BC}) \rightarrow x^i_{AC}$ or $\bigwedge_{A,B,C}  \neg x^i_{AB} \vee \neg x^i_{BC} \vee x^i_{AC}$.
\end{description} For the case $m=7$, transitivity results in a conjunction of $3 \cdot (2^7)^3 = 3 \cdot 2^{21} \approx 6$ million clauses,
and monotonicity is a conjunction of approximately 6200 clauses.\footnote{For a set of size $\ell$, there are $2^{7 - \ell} - 1$ proper supersets. So the number of pairs $(A,B)$ with $A \subset B$ is 
\[ \sum_\ell \binom{7}{\ell} (2^{7 -\ell} - 1) = \sum_\ell \binom{7}{\ell} 2^\ell - 2^7 = 3^7 - 2^7 = 2059. \] } Of course, it would suffice to require monotonicity only for sets whose cardinality differs by one. But then the other inequalities have to be inferred with the use of transitivity. 
Since the number of monotonicity clauses is much smaller than the number of transitivity clauses, we have chosen to be verbose in the case of monotonicity.

\subsubsection{Item Ordering}
\label{ordering}

In Section~\ref{item order} we argued that without loss of generality,  agent zero values $\{g_0\} < \{g_1\} < \ldots < \{g_{m-1}\}$. This is expressed by
\[ \text{Item Order} \assign \bigwedge_{0 \le i < j < 7} x^0_{\{g_i\}\{g_j\}}\]

\subsubsection{Leveled Valuations} 
The following formula states that valuation $v \in \sset{0,1,2}$ is leveled from $k$ on. every set with at least $k$ goods is
preferred to every set with fewer goods. On any level $\ge k$, we order the sets by their number.
\[ \text{Leveled Val$_i$} \assign \bigwedge_{A,B:\; \abs{A} < \abs{B} \; \wedge \; \abs{B} \ge k \text{\ or\ } \abs{A} = \abs{B} \ge k \; \wedge\; A < B} x^i_{AB}\] In fact, the following simpler formula also works:
\[ \text{Simplified Leveled Val} \assign \bigwedge_{A,B: \abs{A} < \abs{B} \wedge \abs{B} \ge k } x^i_{AB} \]
Intuitively, from size $k$ on the valuation increases with the number of
goods, which restricts the search space. The formula does not constrain
comparisons between two distinct sets of the same size $\ell\ge k$. This
is no loss: no other subformula contains a variable $x^i_{AB}$ with
$\abs{A}=\abs{B}\ge k$, so these comparisons may be fixed arbitrarily,
for instance by a fixed order on the sets, without affecting
satisfiability.

For leveled valuations, transitivity
$x^i_{AB}\wedge x^i_{BC}\rightarrow x^i_{AC}$ only needs to be required
when $A$, $B$ and $C$ all have fewer than $k$ goods. Call a set
\emph{large} if it has at least $k$ goods. If $C$ is large and $A$ is
not, the conclusion holds. If $A$ is large and $C$ is not, an antecedent
fails: $x^i_{BC}$ if $B$ is large, and $x^i_{AB}$ otherwise. If $A$ and
$C$ are large and $B$ is not, $x^i_{AB}$ fails.

\subsubsection{Allocations}
\label{allocations}

For the case $m=7$, there are $3^7$ ways of allocating 7 items to 3 agents.
But we only need to consider the allocations where all three sets are non-empty. There are 3 maps that map all items to the same agent, and there are $3 \cdot 2^7$ maps that map all items to at most two agents.

\subsubsection{Putting It Together} \label{sec:puttogether}
We want to know whether the following formula is satisfiable.
\[\text{Monotonicity} \wedge
  \text{Transitivity} \wedge \text{Item Ordering} \wedge \text{Levelled Valuations} \wedge \neg \EFX.\]
A satisfying assignment would encode a counterexample to EFX.

\subsection{Formal Verification of the Foundations and the Encoding} \label{sec:formver}

As noted before SAT formulas are incredibly brittle and  easy to get wrong. The slightest errors in them can render the output of the SAT solver meaningless. In this section, we take the constructions described in Sections~\ref{sec:theoefx} and~\ref{sec:encoding} and formally encode them in the Lean theorem prover~\cite{lean4ITP} on top of Lean's mathematics library mathlib~\cite{mathlib4}. Together they establish the formal correctness of the SAT encoding of the EFX problem, which is particularly important in the case of the SAT based proof of the universal existence of EFX allocations in the 3 agents and 7 goods case.\footnote{Since the output of the sat solver, unsat and its resolution proof, is not  a certificate that can be checked for correctness independent of the formula.}
\footnote{We don't directly verify the C++ code that generated the formula itself, since C++ code doesn't have meaningful formal semantics}

\begin{listing}[h]
    \begin{leancode}        
variable {n m : ℕ} [NeZero n] [NeZero m]
variable (hnm : n < m)

include hnm in
theorem full_correctness :
    (dimacsCNF (n := n) (m := m)).sat ↔
      ∃ v : Ag n → ValFunction (G m),
        ∀ a : Alloc (Ag n) (G m),
          a.complete → ¬ (AllocInstance.mk v a).EFX := by ...

include hnm in
theorem reduced_correctness :
    (reducedDimacsCNF (n := n) (m := m)).sat ↔
      ∃ v : Ag n → ValFunction (G m),
        ∀ a : Alloc (Ag n) (G m),
          a.complete → ¬ (AllocInstance.mk v a).EFX :=
  (reduced_dimacs_equisatisfiable hnm).trans (full_correctness hnm)
    \end{leancode}
    \caption{\label{listing:LeanSATEncoding} \small The end to end correctness theorems. Note that we proved separate theorems for the SAT encodings that respectively used, the full transitivity clause and the one with size \emph{reduction} due to our levelled valuation optimization.}
\end{listing}

In the following paragraphs we have a brief overview of the structure of the proof.

\paragraph*{The theories encoded.}

Concretely, the lean formalization encodes
\begin{itemize}
    \item The theory of EFX allocations in the discrete setting.
    
    \item The proofs of correctness of the optimizations in ~\ref{sec:theoefx} and ~\ref{leveled valuations}.
    \item A model of SAT formulas and their satisfiability, and their invariance under injective mappings of its literals, as well as the correctness of the dense encoding in the DIMACS format.
    \item A three stage lowering of the EFX problem's resolution to the satisfiability of SAT formulas, using proofs of refinement.
    \item A pair of end to end theorems combining all of the above to show that our SAT formula indeed encodes the EFX problem. 
\end{itemize}

Next we briefly outline the proof structure of the correctness of the SAT formula. Concretely we show that the satisfiability of our optimized SAT formula is equivalent to the existence of a valuation function triplet for which no EFX allocation exists.

\begin{itemize}
    \item The first step maps the mathematical theory of allocations and valuations into a skeletal theory. At this stage valuation functions are reduced to ordering relations, and allocations are non-empty and have bounded cardinality. We show that an EFX counterexample in this skeletal theory is an EFX counterexample in the original theory. In other words, the mathematical theory refines the skeletal theory. The basic definitions of this skeletal theory are shown in \ref{listing:LeanSATEncoding}.
    \item  In the second step, the skeletal theory is mapped to the abstract SAT formulas with indices labelled by tuples of agents and goods. Here we prove that the satisfiability of the abstract SAT formula yields a skeletal valuation function that is an EFX counterexample in the skeletal theory.
    \item  The final step proves that these tuples map injectively into integers according to the mapping of \ref{sec:variables}. Since the satisfiability of SAT formulas is not affected by a renaming of its literals, we complete the formal proof that says that we have indeed correctly constructed our SAT formula. This formal proof is complemented by the DRAT-trim verification of the unsat core for the formula in the 3 agents - 7 goods case. Together these two pieces, the lean proof and the drat-trim verification serve as a complete proof of the existence of EFX allocations for any monotone valuation on 3 agents and 7 goods. 
\end{itemize}

\subsection{Experiments}
\label{Experiments}

For the experiments, we use SPASS-SAT (\cite{BrombergerEtAl19}) to compute reduced CNFs and to check their satisfiability.
Reducing the generated CNF is always possible, because the monotonicity, leveling, and item-ordering clauses are all units.
The difference in clauses between the generated CNF and the reduced CNF shows the impact of our theoretical foundations as well.
SPASS-SAT is the SAT solver out of the SPASS workbench
that is contained in the SMT solver SPASS-SATT, as well as in members of the SPASS-SCL family.
Results at the SAT level are confirmed by CaDiCal~3.0.0~\cite{BiereFFFFP24} and the proofs shown correct by DRAT-trim~v05.22.2023~\cite{WetzlerHH14}.
In addition, the satisfiability result for three agents and eight goods is translated from propositional variables back to the actual valuations. They have been checked
by separate code probing all possible assignments of the eight goods. It is presented and discussed in detail in Section~\ref{sec:eightitems}.
SPASS-SAT and CaDiCal have been compiled and executed on a Debian~12 Linux server equipped with AMD EPYC 9754 128-Core processors and 2~TB of main memory.

\subsubsection{Six Items}

We report on experiments with three agents and six items.
In particular, to demonstrate the effects of our theoretical foundations on the running time.

We present an excerpt from our experiments.
We have 6084 variables. The following table shows the number of initial clauses, clauses after reduction by SPASS-SAT,
the running time, and the memory usage for the values of $k=5$ and $k=4$ for leveling and finally with the
addition of the item order.
Both refinements leveling and item order add only units to the CNF which typically always
improve the performance of an afterwards applied solver. However, a lower leveling value $k$
enables less transitivity clauses, see Section~\ref{leveled valuations}. This explains the huge difference
in the number of clauses between values $k=5$ and $k=4$, see below.
Therefore, we decided to go for $k = m-2$ leveling plus the item ordering for all other runs.

\begin{center}
  \begin{tabular}{|l|r|r|r|r|} \hline
    & generated clauses & after reduction & time (sec) & mem usage (MB) \\ \hline \hline
    $k = 5$  &  461835 & 110520 & 33.4 & 246 \\ \hline 
    $k = 4$  &  189723 &  47310 & 25.6  & 146 \\ \hline
    $k = 4$, item-ordering & 189735 & 43813  & 0.53 & 44 \\ \hline  
  \end{tabular}
\end{center}

\subsubsection{Seven Items}

We have $24\,384 = 3 \cdot 8\,128$ variables. For $k = 5$ and item-ordering, the number of clauses is 2\,596\,677 initially and 680\,779 clauses after CNF-reduction. SPASS-SAT required about 30~hours for a DRAT proof of
unsatisfiability of about 35~GB. Unsatisfiability is confirmed by CaDiCal after about 26 hours, producing a proof of unsatisfiability of 37~GB. Both proofs have been verified by DRAT-trim. The DRAT proof of SPASS-SAT
contains an unsatisfiable core of 82\,417 clauses, meaning that 82\,417 clauses out of the initial 680\,779 clauses are sufficient to derive unsatisfiability. In particular, this suggests that the number of generated transitivity clauses can
be further reduced.

\subsubsection{The Counterexample: Eight Items} \label{sec:eightitems}

We have $97\,920$ variables. For $k = 8$ and item-ordering, the initial number of clauses is 29\,202\,318. CNF-reduction eliminates all unit clauses present in the problem formulation and reduces the number of clauses to 8\,138\,126. Recall that the running time for 7 items was more than two working days, whilst the running time for 6 items was less than a second. We concluded that we need to give SPASS-SAT additional hints to have a chance of proving the existence of EFX for eight items. We analyzed the proof for seven items and experimented with adding the first seven unit clauses found in this run as guidance for eight items. To our surprise, SPASS-SAT after 20~hours as well as CadiCal after 3~hours found models.
The model found by SPASS-SAT has been retranslated into valuations for the three agents
and probed for all assignments by three independent procedures. They all confirm the non-existence of EFX on the valuations.

\paragraph*{The Approach.}
As mentioned in the introduction, we started with the belief that EFX allocations always exist. The resolution proof for $m = 6$ took less than a second, the resolution proof for $m = 7$ took more than two working days, so it seemed hopeless at first to try $m = 8$. A close inspection of the proof for $m = 7$ revealed that the first unit clauses were found about half an hour before the end of the proof.%

Seven unit clauses were found at first, namely (after retranslation to human readable form)
\[  \sset{g_0,g_1} <_0 \sset{g_6},\ \sset{g_1,g_2} <_0 \sset{g_6},\ \sset{g_1,g_3,g_4} <_0 \sset{g_6}    \]
for agent zero and
\[ \sset{g_4,g_5} <_1 \sset{g_1,g_6},\ \sset{g_1,g_3,g_4} <_1 \sset{g_0,g_5},\ \sset{g_0,g_2} <_1 \sset{g_0,g_4},\ \sset{g_0,g_2} <_1 \sset{g_1,g_4}\]
for agent one. Then, for a while, no additional unit clauses were found, and finally, the number of unit clauses grew quickly.

We suspected that these units could serve as a useful case distinction in a proof of unsatisfiability.
For $m = 7$ a case distinction on values for the above unit clauses resulted in short solving times between $2$ and $30$ minutes. So they seemed to be a plausible start
for a case distinction for $m = 8$.
We extrapolated what these units
could mean for $m = 8$. For the valuation of agent zero, we replaced $g_6$ by $g_7$, i.e., we kept the most valuable good on the right hand side of the inequalities.
For agent one, we found no meaningful generalization and simply kept these inequalities. Our idea was then to run SPASS on  all 
$2^7$ possible assignment to these $7$ units, possibly in parallel, with the hope that the assignment to these seven units would give SPASS a head start.
To explore this approach, we tried a small random subset of these $2^7$ valuations translated to $m=8$. We provide these valuation functions with our accompanying artefact. Somewhat surprisingly, the formula for one of them was satisfiable, which yields a counterexample to EFX
existence.

\newcommand{\val}{\mathit{val}}

\subsection{Verification of the Counterexample}

The C++-program linked in the artifacts in~\Cref{subsec:Companion Material} checks the counterexample.
The three input valuations are in a single file with $3\cdot256$ lines. Each line contains three values:
the name of the set, the set as a bitstring, and the rank of the set. We construct a two-dimensional array $\val$,
where $\val(i,A)$ is the value of the set $A$ for agent $i$. We then check monotonicity and the EFX-property by brute-force. The program is available in the artifacts (see \Cref{subsec:Companion Material}).

\subsection{The Companion Material}\label{subsec:Companion Material}
The artifacts can be accessed at \url{https://nextcloud.mpi-klsb.mpg.de/index.php/s/25x4Q8eQErYsZE4} (189MB) for the SAT encoding, \url{https://nextcloud.mpi-klsb.mpg.de/index.php/s/XEoGiPXSL6MJpyr} (18.5GB) for the ziped drat proof for three agents and seven
goods, and at \url{https://doi.org/10.5281/zenodo.23734798} for the LEAN development as a citeable artifact \cite{https://doi.org/10.5281/zenodo.23734798}.

The first archive contains:
\begin{itemize}
  \item Instances for three agents in DIMACS format. EFXn.*cnf specifies an instance with $n$ goods, EFX*.kn*cnf specifies an instance for $n$ goods and leveling starting at $k$ (leveling units are included), and EFX*.*i*cnf also includes item order. EFX<s1>.<s2>rcnf contains the cnf of EFX<s1>.<s2>cnf after reduction by SPASS-SAT -u, EFX8.model contains the model for EFX8.cnf, and 
EFX8.srcnf is the EF8.cnf file after setting of certain variables and after reduction by SPASS-SAT -u.
\item SPASS-SAT, a  static linux x86 binary of the sat solver,
\item BasicGen.cpp, the C++ generator for the DIMACS files. Setting $m$, $k$, $P=2^m$, in the source, compiling and executing will generate the respective DIMACS file
\item SPASS-SAT -x EFC8.mcnf  translates the model into a valuation, prints the valuation and checks
  it. For every assignments, it prints  where EFX is violated.
\item Val0/1/2ByCard.txt contain the valuations sorted in increasing order of cardinality and lexicographically for fixed cardinality. 
\end{itemize}

The third archive (\cite{https://doi.org/10.5281/zenodo.23734798}) contains a LEAN project accompanying Section~\ref{sec:formver} and Python script \texttt{lra\_rdl.py} that generates SMTLIB input and calls Z3. The parameters for the number $n$ of agents and the number $m$ of goods can be adapted at the end of the script.

\end{document}